\pdfoutput=1
\documentclass{article}
\PassOptionsToPackage{numbers, compress}{natbib}

\usepackage[preprint]{neurips_2026}
\usepackage{tabularx}
\usepackage{wrapfig}

\usepackage[utf8]{inputenc}
\usepackage[T1]{fontenc}
\usepackage{float}
\usepackage{hyperref}
\hypersetup{hidelinks}
\hypersetup{pdftitle={Causal Inference in Metric Spaces via Distance Profiles},pdfauthor={Zheye Song, Zhiheng Zhang}}
\usepackage{url}
\usepackage{booktabs}
\usepackage{placeins}
\usepackage{amsfonts}
\usepackage{nicefrac}
\usepackage{microtype}
\usepackage{xcolor}

\usepackage{algorithm}
\usepackage{algorithmic}
\usepackage{dsfont}

\usepackage{graphicx}
\usepackage{amsmath,amssymb,amsthm}
\usepackage{bm}
\usepackage{comment}
\usepackage[capitalize,noabbrev]{cleveref}
\usepackage{enumitem}
\usepackage{physics}
\usepackage{bbm}

\theoremstyle{plain}
\newtheorem{theorem}{Theorem}[section]
\newtheorem{proposition}[theorem]{Proposition}
\newtheorem{lemma}[theorem]{Lemma}
\newtheorem{corollary}[theorem]{Corollary}
\theoremstyle{definition}
\newtheorem{definition}[theorem]{Definition}
\newtheorem{example}[theorem]{Example}
\newtheorem{remark}[theorem]{Remark}
\newtheorem{assumption}[theorem]{Assumption}

\DeclareMathOperator*{\E}{\mathbb{E}}
\DeclareMathOperator*{\Prob}{\mathbb{P}}

\DeclareMathOperator*{\Var}{Var}

\newcommand{\ind}{\perp\!\!\!\perp}

\newcommand{\dto}{\xrightarrow{\,d\,}}

\makeatletter
\providecommand{\customlabel}[2]{%
  \protected@write\@auxout{}{\string\newlabel{#1}{{#2}{\thepage}{#2}{#1}{}}}%
  \hypertarget{#1}{#2}%
}
\makeatother

\title{Causal Inference in Metric Spaces via Distance Profiles}

\author{%
  Zheye Song \\
  Department of Biostatistics \\
  University of North Carolina at Chapel Hill \\
  Chapel Hill, NC, USA \\
  \texttt{songzy@unc.edu} \\
  \And
  Zhiheng Zhang\thanks{Corresponding author.} \\
  School of Statistics and Data Science \\
  Shanghai University of Finance and Economics \\
  Shanghai, China \\
  \texttt{zhangzhiheng@mail.shufe.edu.cn} \\
}

\begin{document}

\maketitle
\setcounter{footnote}{0}

\begin{abstract}
Modern interventions increasingly produce object-valued outcomes, such as distributions, covariance matrices, networks, and trees. Existing causal frameworks often exploit special geometry such as one-dimensional Wasserstein quantiles or geodesic Fr\'echet means, whereas metric statistics studies random objects without confounding adjustment. We propose the \emph{profile causal effect} $\Delta(\omega,t)$, which describes where and at what scale a treatment shifts probability mass, including changes in dispersion, tails, or structure that leave mean-level summaries unchanged. Formally, $\Delta(\omega,t)=\Pr\{d(\omega,Y(1))\le t\}-\Pr\{d(\omega,Y(0))\le t\}$ contrasts the potential-outcome probabilities of the same anchor-centered metric ball and requires no structure beyond the metric. On grids fixed by a deterministic or auxiliary design, we combine cross-fitted doubly robust scores with high-dimensional Gaussian approximation and the multiplier bootstrap to obtain simultaneous confidence bands and a global test under regularity conditions. The bands bound treatment-induced changes in ball and shell probabilities simultaneously and, at fixed anchors, extend to all radii. A protocol guides the choice of anchors and radii. Across the experiments, the global test detects concentration effects that mean-quantile and Fr\'echet-mean summaries miss. Its empirical size is near or below the nominal level, even under confounding, where naive metric tests are substantially oversized.
\end{abstract}

\section{Introduction}
\label{sec:intro}

Modern interventions increasingly produce object-valued outcomes: a dispatch policy changes an urban mobility network, a treatment changes a brain connectivity graph, and a pricing rule changes the distribution and spatial pattern of waiting times. Object-oriented data analysis and Fr\'echet regression formalize such observations as statistical objects or metric-space-valued responses rather than Euclidean vectors~\citep{marron2014overview,petersen2019frechet,dubey2024metric}. Recent causal work has begun to move beyond independent scalar outcomes through network interference models, text-as-outcome designs, distributional treatment effects, and geodesic random-object effects, but these approaches typically rely on domain-specific dependence structure, learned representations, quantile geometry, or Fr\'echet/geodesic structure~\citep{ogburn2024causal,egami2022texts,lin2023causal,kurisu2024geodesic}. In these settings, the scientific question is rarely whether a scalar mean changes. A policy may leave the average waiting time nearly unchanged while increasing tail risk in underserved neighborhoods; a therapy may preserve a Fr\'echet mean while reducing dispersion in a network-valued biomarker. Scalar summaries can therefore erase precisely the causal signal of interest.

Object-valued outcomes also make the choice of inferential coordinates part of the design. For scalar or Euclidean outcomes, the outcome coordinates are usually fixed by the measurement scale. For random objects, many informative descriptors are distribution-relative: centrality and outlyingness are defined relative to a reference population or empirical data cloud, while signal relevance depends on the chosen anchors and radii. Depth-based analysis formalizes such distribution-relative descriptors, and metric halfspace depth extends this perspective to object data in general metric spaces~\citep{zuo2000general,dai2023tukey}. Metric statistics provides a complementary representation through distance profiles: for an anchor \(\omega\), \(F_\omega(t)=\Pr\{d(\omega,X)\le t\}\) records the probability mass of metric balls centered at \(\omega\)~\citep{dubey2024metric}. Thus a profile coordinate is not an intrinsic attribute of a single object; it is a metric-neighborhood descriptor defined relative to a reference law or design.

We adopt this idea in a causal form. For each anchor--radius pair \(g=(\omega,t)\), define the potential profile indicator \(Z^{(a)}(g)=\mathds{1}\{d(\omega,Y(a))\le t\}\). The event \(d(\omega,Y(a))\le t\) asks whether the potential outcome falls inside the metric neighborhood \(B(\omega,t)\). At a fixed profile location, the object-valued problem becomes an ordinary binary potential-outcome contrast. Varying \((\omega,t)\) turns these coordinates into a geometry-indexed causal process (\hyperref[fig:profile-geometry]{Figure~\ref*{fig:profile-geometry}}). This transformation preserves potential-outcome reasoning while avoiding linear operations, quantile maps, and geodesic calculus on the outcome space.

We define the \emph{profile causal effect}
\[
\Delta(\omega,t)
=
\Pr\{d(\omega,Y(1))\le t\}
-
\Pr\{d(\omega,Y(0))\le t\}.
\]
This estimand compares how the two potential-outcome laws assign probability to the same anchor-centered metric balls. It provides a common inferential language across distributions, covariance matrices, networks, trees, and related random objects, while preserving the native metric of each space.

\begin{figure}[htbp]
\centering
\includegraphics[width=0.917\textwidth]{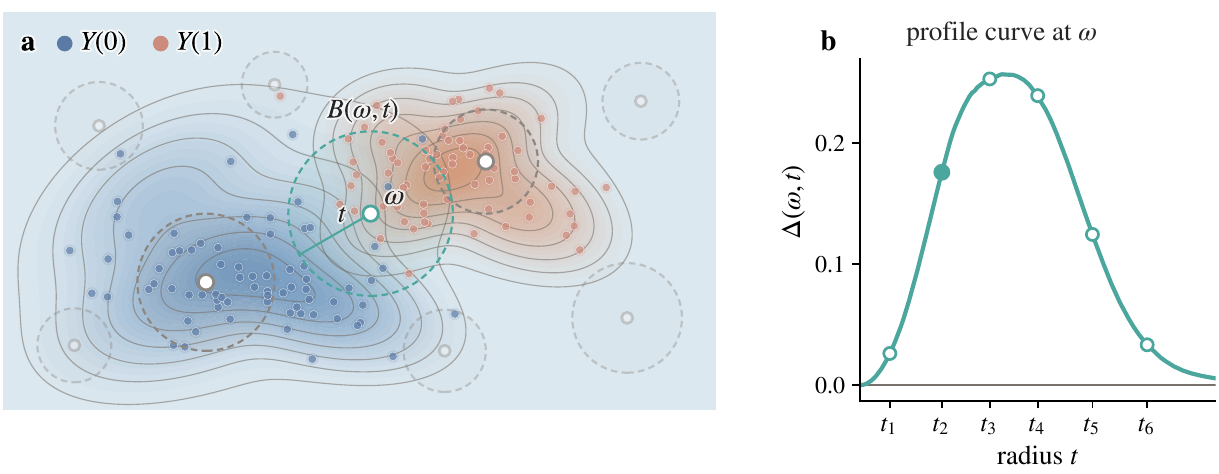}
\caption{Profile geometry. (a) Points are draws from the laws of \(Y(0)\) (blue) and \(Y(1)\) (orange), whose densities are shown by shading and contours. Each dashed circle is a metric ball \(B(\omega,t)\) around an anchor, and fainter circles are other anchor--radius cells of the inferential grid. The profile causal effect \(\Delta(\omega,t)\) is the probability that \(Y(1)\) falls in the ball minus the probability that \(Y(0)\) does. (b) The profile curve \(t\mapsto\Delta(\omega,t)\) at the anchor of the teal ball; markers are the grid radii \(t_1,\ldots,t_6\), and the filled one is the radius in (a).}
\label{fig:profile-geometry}
\end{figure}

Profile coordinates thus depend on the chosen anchors and radii, and the design becomes part of the inferential target. If anchors, radii, and trimming masks are chosen from the same observations used for score evaluation, the target process itself becomes main-sample adaptive. Rather than solving a full selective-inference problem, we separate grid construction from score evaluation: we fix the finite inferential grid through a deterministic or auxiliary design, leaving the main sample solely for cross-fitted doubly robust score evaluation~\citep{rasines2023splitting}. This design-exogenous finite-grid target matches the regime in which high-dimensional Gaussian approximation and the multiplier bootstrap have formal guarantees~\citep{chernozhukov2013gaussian}; with the anchors fixed in the same way, the bands extend to all radii (\hyperref[subsec:all_radii]{Section~\ref*{subsec:all_radii}}).

Causal interpretation also requires adjustment for confounding. Distance-based group comparison can reveal association, but it does not identify a causal effect when covariates affect both treatment and the object-valued outcome. Under conditional exchangeability and positivity, adjustment for the propensity score suffices to remove confounding~\citep{rosenbaum1983central}. We therefore embed metric profiles into the potential-outcome framework and estimate them with orthogonal doubly robust scores, drawing on semiparametric efficiency, augmented inverse-probability weighting (AIPW), cross-fitting, and random-object causal inference~\citep{hahn1998role,bang2005doubly,chernozhukov2018double,bhattacharjee2025doubly}.

Existing approaches address important but distinct pieces of the problem. Distributional causal inference gives a sharp construction in \(W_2(\mathbb R)\) by exploiting one-dimensional quantile geometry~\citep{lin2023causal}. Geodesic causal inference defines interpretable mean-level effects when geodesics and Fr\'echet means are available~\citep{kurisu2024geodesic}. Metric statistics develops distance-profile tools for random-object exploration and testing in general metric spaces~\citep{dubey2024metric}, but it does not address treatment assignment, confounding, or potential outcomes. At a fixed anchor, \(d(\omega,Y)\) is a scalar outcome, so each profile coordinate is a distributional treatment effect for that distance. Its identification, efficient influence function, and doubly robust estimation are those of a binary treatment effect~\citep{hahn1998role,chernozhukov2013counterfactual,belloni2017program}, and because the reduction uses only the metric, they hold in every separable metric space. At a single coordinate, efficiency also maximizes local power among Wald tests built from regular estimators (\hyperref[thm:pointwise-eff-power]{Proposition~\ref*{thm:pointwise-eff-power}}). The problems specific to object-valued outcomes arise when coordinates are combined: calibrating them jointly across anchors and radii, including coordinates with vanishing variance, and designing the grid on which they are compared.

We make four contributions.
\begin{enumerate}[label=(\roman*), leftmargin=*]
\item \textbf{Profile causal estimands for random objects.} We define anchor--radius profile causal effects for outcomes in separable metric spaces, turning object-valued causal questions into a geometry-indexed collection of binary potential-outcome contrasts.
\item \textbf{Doubly robust simultaneous inference.} We construct cross-fitted doubly robust scores and prove simultaneous confidence bands and max-statistic global tests on inferential grids fixed by an auxiliary design under explicit nuisance-rate conditions. For a fixed set of anchors, the bands extend to all radii, including coordinates with zero variance.
\item \textbf{Interpretation and design.} We show how simultaneous bands translate into probability statements about balls and shells, describe what finitely many profile coordinates cannot recover, and give a protocol for choosing anchors and radii and for reading primary and sensitivity designs.
\item \textbf{Empirical evidence.} We study the finite-grid procedure on five simulated outcome spaces and on two semi-synthetic benchmarks built from real object geometry. The max test has empirical size near or below the nominal level, including under confounding, detects concentration effects that mean-quantile and Fr\'echet-mean proxies miss, and has power that depends on the anchor design.
\end{enumerate}

\section{Related Work}
\label{sec:related}

\paragraph{Metric statistics and random objects.}
Our methodological starting point is metric statistics for random objects, especially distance-profile representations for metric-space-valued data~\citep{dubey2024metric}. Distance profiles describe random objects through anchor-centered metric-ball probabilities and support nonparametric exploration and testing in metric spaces. This line is object-general, but it is not a confounding-adjusted causal framework: it studies random-object laws and distance-based comparisons rather than potential outcomes under treatment assignment. We use the same metric-neighborhood coordinates, but define them for potential outcomes, estimate them with doubly robust scores, and calibrate them jointly on a finite grid fixed by an auxiliary design or, at fixed anchors, over all radii.

\paragraph{Distributional causal inference.}
Our work intersects with causal inference for distribution-valued outcomes, where existing methods exploit the one-dimensional Wasserstein geometry of quantile functions~\citep{lin2023causal}. This construction is sharp when such quantile geometry is available. Our profile effect targets a different object: the probabilities assigned by potential-outcome laws to anchor-centered metric balls. In \(W_2(\mathbb R)\), Lin-type estimands describe horizontal mean-quantile deformation, whereas the profile map describes anchor-relative changes in the law of distribution-valued outcomes.

\paragraph{Geodesic causal inference.}
Geodesic causal inference defines treatment effects through geodesics connecting Fr\'echet means of potential-outcome laws~\citep{kurisu2024geodesic}. This gives an interpretable mean-level target when the outcome space has usable geodesics and stable Fr\'echet means. Our framework differs in both target and required geometry. The profile effect map \(\Delta(\omega,t)\) is indexed by anchors and radii and can detect changes in dispersion, tails, and multimodality that leave a Fr\'echet mean unchanged. Pointwise identification requires no outcome-space geometry beyond the metric.

\paragraph{Distributional treatment effects and kernel tests.}
Inference on distributional treatment effects is well developed for scalar outcomes~\citep{chernozhukov2013counterfactual,belloni2017program,byambadalai2024distributional}, and kernel methods compare interventional outcome laws through mean embeddings, with doubly robust tests and one-step estimators~\citep{martinez2023drxkte,fawkes2024drkernel,luedtke2024hilbert,jain2026conditional}. Finite-location tests evaluate the witness function at selected points~\citep{jitkrittum2016interpretable}; \citet{zenati2026drme} calibrate a single chi-square test at witness locations learned from a local-power criterion.\footnote{\citet{zenati2026drme} (arXiv, 8 May 2026) is concurrent and independent work; an earlier version of this paper was submitted for peer review on 4 May 2026.} At a fixed anchor, the profile curve is the distribution function of the scalar distance \(d(\omega,Y)\), so our pointwise scores belong to the first of these lines. What we calibrate is different: many such effects, indexed by the anchors and radii of one metric, with a simultaneous interval for every retained coordinate rather than a single kernel statistic. Under the ball-indicator kernel, witness values are linear combinations of profile coordinates (\hyperref[app:ball_kernel]{Appendix~\ref*{app:ball_kernel}}).

\section{Preliminaries}
\label{sec:defs}

Let \((\Omega,d)\) be a separable metric space. The outcome \(Y\) and potential outcomes \(Y(0),Y(1)\) take values in \(\Omega\). For an anchor \(\omega\in\Omega\) and radius \(t\in\mathbb R_+\), the distance profile of a random object \(X\) is
\[
F_\omega(t)=\Pr\{d(\omega,X)\le t\}.
\]
If \(\Omega\) is totally bounded and \((\Omega,d^\theta)\) is of strong negative type (\hyperref[def:snt]{Definition~\ref*{def:snt}}) for some \(\theta>0\), two laws whose distance profiles agree at every anchor and every radius are equal~\citep[Proposition~1]{dubey2024metric}. This law-level property concerns the full anchor family rather than a finite grid; pointwise identification and estimation below use only the profile indicator at a fixed anchor--radius pair.

Let \(\mathcal W\subseteq\Omega\) be an anchor domain and \(\mathcal T\subseteq\mathbb R_+\) a radius domain. The population working domain is
\[
\mathcal G=\mathcal W\times\mathcal T,
\qquad
g=(\omega,t)\in\mathcal G.
\]
This continuous domain motivates the profile process. Formal simultaneous inference is carried out on finite grids and, for a fixed set of anchors, over all radii.

\begin{definition}[Grid layers]\label{def:grid_defs}
We use three nested grids, \(\mathcal G_n^{\mathrm{inf}}\subseteq\mathcal G_n^{\mathrm{cand}}\subseteq\mathcal G_n^{\mathrm{est}}\subset\mathcal G\). The \emph{estimation grid} \(\mathcal G_n^{\mathrm{est}}\) carries point estimates, visualizations, and diagnostics; the \emph{candidate grid} \(\mathcal G_n^{\mathrm{cand}}\) is the finite set of cells considered for formal inference; and the \emph{inferential grid} \(\mathcal G_n^{\mathrm{inf}}\) is the set of cells retained after pilot trimming, on which bands and tests are computed.
\end{definition}

Let \(\mathcal F_n^{\mathrm{aux}}\) denote the sigma-field generated by all auxiliary design information used to construct anchors, radii, pilot trimming masks, and inferential grids. Candidate and inferential grids may be deterministic or \(\mathcal F_n^{\mathrm{aux}}\)-measurable. The main inferential requirement is \emph{design exogeneity}: the data used for main-sample score evaluation do not choose anchors, radii, or the cells of the inferential grid. Conditionally on \(\mathcal F_n^{\mathrm{aux}}\), the inferential grid is treated as fixed.

In our experiments, trimming uses an auxiliary pilot sample with index set \(\mathcal I_{\mathrm{pilot}}\). For \(g=(\omega,t)\in\mathcal G_n^{\mathrm{cand}}\), define
\[
\hat p_{\mathrm{pilot}}(g)
=
\frac{1}{|\mathcal I_{\mathrm{pilot}}|}
\sum_{i\in\mathcal I_{\mathrm{pilot}}}
\mathds{1}\{d(\omega,Y_i)\le t\}.
\]
For a fixed \(c_p\in(0,1/2)\),
\[
\mathcal G_n^{\mathrm{inf}}
=
\{g\in\mathcal G_n^{\mathrm{cand}}:
\hat p_{\mathrm{pilot}}(g)\in[c_p,1-c_p]\}.
\]
This rule removes near-degenerate cells, whose ball probabilities are close to zero or one, and thereby stabilizes the finite target. Because it acts on marginal ball probabilities, it does not by itself bound the score variance from below; \hyperref[thm:finite_grid_inference]{Theorem~\ref*{thm:finite_grid_inference}} assumes such a bound separately, in (A4).

\section{Causal Framework and Identification}
\label{sec:framework}

\paragraph{Setup and assumptions.}
For each unit we observe \(O=(X,A,Y)\), where \(X\) is a covariate vector, \(A\in\{0,1\}\) is treatment, and \(Y\in\Omega\) is the observed random object. Let \(Y(0),Y(1)\) denote potential outcomes. For \(g=(\omega,t)\), define
\[
Z^{(a)}(g)=\mathds{1}\{d(\omega,Y(a))\le t\},
\qquad
Z(g)=\mathds{1}\{d(\omega,Y)\le t\}.
\]
We assume consistency and no interference, conditional exchangeability \(Y(a)\perp A\mid X\) for \(a\in\{0,1\}\), and positivity: for some \(\epsilon>0\), \(\epsilon\le \Pr(A=1\mid X)\le1-\epsilon\) almost surely.

\paragraph{Profile causal effect.}
The pointwise profile causal effect is
\[
\Delta(g)
=
\mathbb E\{Z^{(1)}(g)\}
-
\mathbb E\{Z^{(0)}(g)\}.
\]

The next two examples are population-level separations between the profile estimand and neighboring causal targets. In \hyperref[sec:experiments]{Section~\ref*{sec:experiments}} we instantiate them as finite-grid stress tests: the Wasserstein collapse experiment mirrors \hyperref[ex:lin-separation]{Example~\ref*{ex:lin-separation}}, and the mean-zero examples on graphs and symmetric positive definite (SPD) matrices mirror \hyperref[ex:kurisu-separation]{Example~\ref*{ex:kurisu-separation}}.

\begin{example}[A separation from mean-quantile effects]\label{ex:lin-separation}
Let \(\Omega=W_2(\mathbb R)\). Suppose
\[
Y(0)=
\begin{cases}
N(-1,1), & \text{with probability }1/2,\\
N(1,1), & \text{with probability }1/2,
\end{cases}
\qquad
Y(1)=N(0,1)\quad\text{almost surely.}
\]
For Gaussian distributions with unit variance, \(Q_{N(\mu,1)}(u)=\mu+\Phi^{-1}(u)\). Hence \(\mathbb E\{Q_{Y(0)}(u)\}=Q_{Y(1)}(u)\), so the mean-quantile contrast is zero. With anchor \(\omega=N(0,1)\), however, \(W_2(Y(1),\omega)=0\) almost surely and \(W_2(Y(0),\omega)=1\) almost surely. Thus, for every \(t\in(0,1)\), \(\Delta(\omega,t)=1\). The profile effect and the mean-quantile contrast target different aspects of causal change.
\end{example}

\begin{example}[A canonical separation from Fr\'echet-mean displacement]\label{ex:kurisu-separation}
Consider a geodesic metric space with two points \(y_-\) and \(y_+\) symmetric around a midpoint \(m\), with \(d(y_-,m)=d(y_+,m)=r>0\). Assume that \(m\) is the unique Fr\'echet mean of the symmetric two-point law. Let \(Y(0)\) put mass \(1/2\) on each of \(y_-\) and \(y_+\), while \(Y(1)=m\) almost surely. The Fr\'echet-mean displacement is zero. With anchor \(\omega=m\), however, \(d(Y(1),\omega)=0\) and \(d(Y(0),\omega)=r\), so \(\Delta(\omega,t)=1\) for every \(t\in(0,r)\). This mirrors the graph and SPD mean-zero examples of \hyperref[subsec:estimand_separation]{Section~\ref*{subsec:estimand_separation}}.
\end{example}

Let
\[
\pi(x)=\Pr(A=1\mid X=x),
\qquad
m_a(g;x)=\mathbb E\{Z(g)\mid A=a,X=x\},
\]
and write \(\eta=(\pi,m_0,m_1)\), with true value \(\eta_0\).

\begin{theorem}[Pointwise identification and orthogonal score]\label{thm:ident}
Under the causal assumptions above, for each fixed \(g\in\mathcal G\),
\[
\Delta(g)=\mathbb E\{m_1(g;X)-m_0(g;X)\}.
\]
Moreover, the uncentered AIPW score
\[
\phi(O;g,\eta)
=
m_1(g;X)-m_0(g;X)
+
\left\{
\frac{A}{\pi(X)}
-
\frac{1-A}{1-\pi(X)}
\right\}
\{Z(g)-m_A(g;X)\}
\]
satisfies \(\mathbb E\{\phi(O;g,\eta_0)\}=\Delta(g)\). Its centered version
\[
\psi(O;g)=\phi(O;g,\eta_0)-\Delta(g)
\]
is the efficient influence function for the binary profile coordinate. The score is Neyman orthogonal and doubly robust in the usual sense for an average treatment effect.
\end{theorem}

\hyperref[thm:ident]{Theorem~\ref*{thm:ident}} applies the standard argument for a binary outcome to \(Z^{(a)}(g)\) (\hyperref[app:proofs]{Appendix~\ref*{app:proofs}}). Two of its consequences are used later. First, with estimated nuisance functions, the bias of the score is a product of propensity and outcome-regression errors, which leads to the rate conditions in \hyperref[sec:est]{Section~\ref*{sec:est}}. Second, the variance of \(\psi(O;g)\) sets the scale on which coordinates are standardized and compared.

\section{Estimation and Inference}
\label{sec:est}

We estimate \(\Delta(g)\) by evaluating the doubly robust score over a finite grid. Point estimates and visual diagnostics may use the larger grid \(\mathcal G_n^{\mathrm{est}}\). Simultaneous confidence bands and formal max-statistic tests use only the inferential grid \(\mathcal G_n^{\mathrm{inf}}\). At a fixed coordinate the score is the standard AIPW score for a binary outcome; what follows calibrates many such scores jointly on a possibly growing grid.

\subsection{Finite-Grid Inference}
\label{subsec:finite_grid}

\paragraph{Conditional-CDF nuisance learning.}
For a fixed anchor \(\omega_k\), let \(D_k=d(\omega_k,Y)\). The nuisance target is the conditional CDF
\[
m_a(\omega_k,t;x)
=
\Pr(D_k\le t\mid A=a,X=x).
\]
Estimating this object directly preserves monotonicity in the radius coordinate \(t\). The theory is compatible with any conditional distribution learner satisfying the finite-grid rate and correlation conditions below. The experiments use a Distributional Random Forest (DRF) weighted empirical CDF (ECDF) learner~\citep{cevid2022drf} and, for the five-space table, an anchor-specific kernel-weighted ECDF learner; we compare the two in \hyperref[app:learner_robustness]{Appendix~\ref*{app:learner_robustness}}.

\paragraph{Implementation summary.}
The implementation first constructs anchors, radii, and trimming masks from deterministic rules or auxiliary information. It then fits nuisance functions by cross-fitting on the main sample, evaluates the AIPW scores on \(\mathcal G_n^{\mathrm{est}}\), and computes standard errors, simultaneous bands, and max-statistic cutoffs on \(\mathcal G_n^{\mathrm{inf}}\) using a Gaussian multiplier bootstrap applied to the centered scores. Full pseudocode appears in \hyperref[app:implementation]{Appendix~\ref*{app:implementation}}.

\paragraph{Regularity.}
The results below use the following high-level conditions, stated in full in \hyperref[app:regularity_proofs]{Appendix~\ref*{app:regularity_proofs}}.
\begin{itemize}[leftmargin=*, label={}]
\item \customlabel{assumption:A1}{(A1)} \textit{Causal identification}: consistency, no interference, conditional exchangeability, and positivity hold for the profile indicators.
\item \customlabel{assumption:A2}{(A2)} \textit{Design exogeneity}: anchors, radii, and trimming masks are measurable with respect to an auxiliary sigma-field \(\mathcal F_n^{\mathrm{aux}}\), independent of the main score-evaluation sample.
\item \customlabel{assumption:A3}{(A3)} \textit{Cross-fitted nuisance control}: nuisance learners are trained off-fold, propensities are clipped, and the fold-wise conditional-CDF and propensity errors satisfy finite-grid squared-rate and product-rate conditions.
\item \customlabel{assumption:A4}{(A4)} \textit{Grid size and nondegeneracy}: \(\mathcal G_n^{\mathrm{inf}}\) is nonempty with high probability, the oracle coordinate variances are bounded away from zero, and the candidate grid satisfies \(\log^7(|\mathcal G_n^{\mathrm{cand}}|\,n)=o(n)\).
\item \customlabel{assumption:A5}{(A5)} \textit{Correlation consistency}: the correlation matrix of the centered estimated scores approximates that of the standardized oracle scores at the rate given in \hyperref[app:regularity_proofs]{Appendix~\ref*{app:regularity_proofs}}.
\end{itemize}

With the main sample split into folds, let \(\hat\phi_i(g)=\phi(O_i;g,\hat\eta^{(-j)})\) for \(i\) in fold \(j\), where \(\hat\eta^{(-j)}\) is fitted without fold \(j\). For \(g\in\mathcal G_n^{\mathrm{inf}}\), let
\[
\hat\Delta(g)=n^{-1}\sum_{i=1}^n \hat\phi_i(g),
\qquad
\sigma^2(g)=\mathrm{Var}\{\psi(O;g)\mid\mathcal F_n^{\mathrm{aux}}\}.
\]
Let
\[
N_n^{\mathrm{cand}}=|\mathcal G_n^{\mathrm{cand}}|,
\qquad
p_n=|\mathcal G_n^{\mathrm{inf}}|.
\]
All Gaussian, bootstrap, band, and max-statistic statements are over the inferential grid of dimension \(p_n\), while logarithmic rate bounds use the conservative candidate-grid envelope \(N_n^{\mathrm{cand}}\). Let \(\hat\sigma^2(g)=n^{-1}\sum_{i=1}^n\{\hat\phi_i(g)-\hat\Delta(g)\}^2\) be the empirical variance of the cross-fitted scores, and let \(\hat c_{1-\alpha}\) be the conditional \((1-\alpha)\)-quantile, given the data, of
\[
\max_{g\in\mathcal G_n^{\mathrm{inf}}}\frac{|n^{-1/2}\sum_{i=1}^n\xi_i\{\hat\phi_i(g)-\hat\Delta(g)\}|}{\hat\sigma(g)},
\]
where the \(\xi_i\) are independent standard normal.

\hyperref[cor:primitive_grid]{Corollary~\ref*{cor:primitive_grid}} in \hyperref[app:regularity_proofs]{Appendix~\ref*{app:regularity_proofs}} shows that slightly stronger nuisance rates imply (A3) and (A5).

\begin{theorem}[Finite-grid simultaneous inference]\label{thm:finite_grid_inference}
Under (A1)--(A4), on the high-probability event where the inferential grid is nonempty and nondegenerate,
\[
\sup_{g\in\mathcal G_n^{\mathrm{inf}}}
\Bigl|\sqrt n\{\hat\Delta(g)-\Delta(g)\}-n^{-1/2}\sum_{i=1}^n\psi(O_i;g)\Bigr|
=
o_P\{(\log(N_n^{\mathrm{cand}}\vee2))^{-1/2}\}.
\]
Moreover, with \(T_n^\circ=\max_{g\in\mathcal G_n^{\mathrm{inf}}}\sqrt n|\hat\Delta(g)-\Delta(g)|/\sigma(g)\),
\[
\sup_{t\in\mathbb R}
\Bigl|\Pr\bigl(T_n^\circ\le t\mid\mathcal F_n^{\mathrm{aux}}\bigr)-\Pr\Bigl(\max_{g\in\mathcal G_n^{\mathrm{inf}}}|Z_n(g)|\le t\Bigm|\mathcal F_n^{\mathrm{aux}}\Bigr)\Bigr|
\to_P0,
\]
where, given \(\mathcal F_n^{\mathrm{aux}}\), \(Z_n\) is a centered Gaussian vector whose covariance is the correlation matrix of \(\{\psi(O;g)\}_{g\in\mathcal G_n^{\mathrm{inf}}}\).
\end{theorem}

\begin{corollary}[Simultaneous bands and global max test]\label{cor:main_scb_global}
Under (A1)--(A5), on the same event,
\[
\Pr\Bigl\{|\hat\Delta(g)-\Delta(g)|\le\frac{\hat c_{1-\alpha}\hat\sigma(g)}{\sqrt n}\ \text{for all }g\in\mathcal G_n^{\mathrm{inf}}\Bigm|\mathcal F_n^{\mathrm{aux}}\Bigr\}
\to_P1-\alpha .
\]
If \(\Delta(g)=0\) for all \(g\in\mathcal G_n^{\mathrm{inf}}\), the test that rejects when \(\max_{g\in\mathcal G_n^{\mathrm{inf}}}\sqrt n|\hat\Delta(g)|/\hat\sigma(g)>\hat c_{1-\alpha}\) therefore has conditional size tending to \(\alpha\) in probability.
\end{corollary}

\hyperref[thm:finite_grid_inference]{Theorem~\ref*{thm:finite_grid_inference}} and \hyperref[cor:main_scb_global]{Corollary~\ref*{cor:main_scb_global}} apply the standard argument for maxima of many orthogonal scores conditionally on \(\mathcal F_n^{\mathrm{aux}}\); by (A2), the inferential grid is then fixed (\hyperref[app:regularity_proofs]{Appendix~\ref*{app:regularity_proofs}}). Orthogonality and cross-fitting give the linearization under (A1) and (A3)~\citep{chernozhukov2018double}. Under (A4), the high-dimensional Gaussian approximation of \citet{chernozhukov2013gaussian} gives the law of the maximum; under (A5), Gaussian comparison and anti-concentration~\citep{chernozhukov2013gaussian,cck2015comparison} carry it to the multiplier bootstrap and its quantile.

In finite samples the inferential grid may be empty, or an estimated standard deviation on it may vanish. In that case the procedure returns the trivial interval \([-1,1]\) for every retained coordinate and does not reject; under (A3) and (A4) this happens with probability tending to zero.

\subsection{Design and Interpretation}
\label{subsec:design_interp}

\paragraph{Designing anchors and radii.}
The results of \hyperref[subsec:finite_grid]{Section~\ref*{subsec:finite_grid}} condition on the design: the anchors and their number, the radii, and any trimming rule are fixed before the main sample is used. Because they determine which ball probabilities are compared, choosing them is part of the scientific question. A scientific reference, such as a reference activity distribution, a connectivity template, or a prototype graph, can serve as an anchor directly; otherwise, a medoid-start MaxMin rule spreads anchors over the auxiliary objects using only their geometry. Radii can be scientifically meaningful tolerances or quantiles of auxiliary distances. For testing, the relevant signal at \(g\) is \(|\Delta(g)|/\sigma(g)\), where \(\sigma^2(g)\) is the variance of the efficient influence function; adding coordinates can raise the largest signal, while the critical value of the maximum can only increase (\hyperref[app:geometry_anchor_theory]{Appendix~\ref*{app:geometry_anchor_theory}}). \hyperref[subsec:design_rules]{Appendix~\ref*{subsec:design_rules}} gives detailed selection rules.

We recommend declaring one primary design and treating further independently fixed designs as sensitivity analyses: a rejection in the primary design is the main finding, and a rejection only in a sensitivity design is a sensitivity finding. Designs can legitimately disagree, because different coordinates ask different questions; a coordinate shared by two designs has the same population effect in both, although its reported significance can differ with the critical value. A joint claim across designs requires calibrating their union with a common critical value. We recommend reporting the global max test first and then, for each highlighted coordinate, the metric, the anchor or its nearest observed exemplar, the radius and its auxiliary percentile, and the estimate with its simultaneous interval.

\paragraph{Probability-scale conclusions.}
Let \([L(g),U(g)]\) denote a simultaneous band. On the event that it covers every coordinate, the following statements hold together; they are projections of one confidence set and need no further multiplicity correction. The sign of \(\Delta(g)\) is determined wherever the band excludes zero. For radii \(s<t\) at the same anchor, the \emph{shell effect} \(\Delta(\omega,t)-\Delta(\omega,s)\), the effect on the probability that \(d(\omega,Y)\) lies in \((s,t]\), is enclosed by \([L(\omega,t)-U(\omega,s),\,U(\omega,t)-L(\omega,s)]\). The largest absolute effect over a set of radii is at most the largest absolute band endpoint over that set; when this bound is below a tolerance \(\tau\), it establishes equivalence at \(\tau\), whereas failure to reject zero does not. On a finite grid these statements concern the retained coordinates (\hyperref[app:effect_outputs]{Appendix~\ref*{app:effect_outputs}}). A positive \(\Delta(\omega,t)\) means that treatment raises the potential-outcome probability of the ball, not that individual outcomes move toward \(\omega\), and a pattern of coordinates need not identify a unique mechanism (Appendices~\ref{app:effect_outputs} and~\ref{app:geometry_anchor_theory}).

\subsection{All Radii at Fixed Anchors}
\label{subsec:all_radii}

When the anchor set is small and fixed, the radius grid can be dropped. Let \(\omega_1,\ldots,\omega_K\) be anchors fixed by the auxiliary design. Write \(\theta_h\) for the arm-specific profiles \(F_{ak}(t)=\Pr\{d(\omega_k,Y(a))\le t\}\) and their contrasts \(\Delta_k(t)=F_{1k}(t)-F_{0k}(t)\), indexed jointly by \(h\in\mathcal H\) (arm or contrast, anchor, and radius \(t\ge0\)), and let \(\hat\theta_h\) be their cross-fitted AIPW estimators. Distance distributions may have atoms, and some coordinates may have zero efficient-score variance, as in \hyperref[ex:lin-separation]{Example~\ref*{ex:lin-separation}}.

The results hold under one of three training regimes, stated in full in \hyperref[ass:complete_learning]{Assumption~\ref*{ass:complete_learning}}: (i) the propensity score is known and the fitted conditional CDFs converge to a common, possibly incorrect, limit; (ii) the propensity score follows a correctly specified finite-dimensional model fitted by maximum likelihood on balanced cross-fitting folds, with the same stability requirement on the conditional CDFs; (iii) both nuisance functions are consistent and the product of their errors is \(o_P(n^{-1/2})\). Let \(u_h\) be the centered AIPW score at the limit of the fitted conditional CDFs. The influence function \(f_h\) is \(u_h\) in regimes (i) and (iii), and is efficient in regime (iii); in regime (ii) it is
\[
f_h=u_h-C_hI^{-1}S,
\qquad
C_h=\mathbb E(u_hS^{\mathsf T}),
\qquad
I=\mathbb E(SS^{\mathsf T}),
\]
where \(S\) is the score of the propensity model at the true parameter. The correction vanishes when the fitted conditional CDFs are consistent; other fold designs lead to different covariances (\hyperref[app:complete_radius]{Appendix~\ref*{app:complete_radius}}).

Inference uses calibration rows \(\hat f_i(h)\) that estimate \(f_h(O_i)\): the centered cross-fitted scores in regimes (i) and (iii), from which the estimated propensity-score projection is subtracted in regime (ii). Because some coordinates may have zero or vanishing variance, the bands use two fixed floors \(s_0,c_{\min}>0\). Let \(\hat s_h=\max\{(n^{-1}\sum_{i=1}^n\hat f_i(h)^2)^{1/2},s_0\}\), let \(\hat c_{1-\alpha}\) now denote the conditional \((1-\alpha)\)-quantile, given the data, of
\[
\sup_{h\in\mathcal H}\frac{|n^{-1/2}\sum_{i=1}^n\xi_i\hat f_i(h)|}{\hat s_h},
\]
where the \(\xi_i\) are independent standard normal, and let \(\hat c=\max(\hat c_{1-\alpha},c_{\min})\). The law is fixed, and probabilities are conditional on the auxiliary design.

\begin{theorem}[Simultaneous bands over all radii]\label{thm:all_radii}
Suppose \hyperref[ass:complete_learning]{Assumption~\ref*{ass:complete_learning}} holds and the calibration rows match the training regime. Then \(\sqrt n(\hat\theta-\theta)\rightsquigarrow G_f\) in \(\ell^\infty(\mathcal H)\), where \(G_f\) is a tight centered Gaussian process with covariance \(\mathbb E(f_hf_{h'})\), and
\[
\liminf_{n\to\infty}\Pr\Bigl\{\sup_{h\in\mathcal H}\frac{\sqrt n\,|\hat\theta_h-\theta_h|}{\hat s_h}\le\hat c\Bigr\}\ge1-\alpha .
\]
If \(G_f\not\equiv0\), this probability with \(\hat c_{1-\alpha}\) in place of \(\hat c\) converges to \(1-\alpha\); if \(G_f\equiv0\), the displayed probability converges to one.
\end{theorem}

The theorem requires no coordinate to have positive variance; its proof is in \hyperref[app:complete_radius]{Appendix~\ref*{app:complete_radius}}. The two floors play different roles: \(s_0\) keeps the studentization stable where a variance vanishes, and \(c_{\min}\) keeps the cutoff from collapsing when the limiting process is degenerate. With \(s_0\) alone, the coverage can fall to \((1-\alpha)^2\) (\hyperref[app:complete_radius]{Appendix~\ref*{app:complete_radius}}).

A conditional empirical CDF weighted by a fixed positive kernel satisfies the stability condition of regimes (i) and (ii) even when its limit is not the true conditional CDF, and cell-wise empirical estimators on a finite covariate space satisfy regime (iii) (Propositions~\ref{prop:fixed_kernel_stability} and~\ref{prop:finite_strata}). Such fits are step functions, so the supremum over radii is computed exactly at finitely many jump points. The experiments in \hyperref[sec:experiments]{Section~\ref*{sec:experiments}} use the finite-grid procedure.

\section{Experiments}
\label{sec:experiments}

We organize the experiments around four questions:
\begin{enumerate}[label=(Q\arabic*), leftmargin=*]
\item Calibration: does the finite-grid max test hold its level on object spaces as different as distributions, covariance matrices, graphs, and trees?
\item Estimand separation: which changes does the profile test detect that mean-quantile and Fr\'echet-mean summaries miss, and which do those summaries detect better?
\item Confounding: does doubly robust adjustment keep the test at its level when covariates drive both treatment and outcome?
\item Real geometry: does the procedure stay calibrated and informative when the outcomes inherit the geometry of real connectomes and networks?
\end{enumerate}
\hyperref[subsec:calibration]{Section~\ref*{subsec:calibration}} addresses (Q1), \hyperref[subsec:estimand_separation]{Section~\ref*{subsec:estimand_separation}} addresses (Q2) and (Q3), and \hyperref[subsec:real_geometry]{Section~\ref*{subsec:real_geometry}} addresses (Q4).

\subsection{Finite-Grid Calibration and Convergence}
\label{subsec:calibration}

We first ask whether the finite-grid max test of \hyperref[subsec:finite_grid]{Section~\ref*{subsec:finite_grid}} is calibrated in finite samples once the object-space problem is reduced to retained coordinates. The diagnostic uses five spaces: Euclidean vectors, Wasserstein quantile outcomes, \(3\times3\) SPD matrices, graphs, and rooted trees compared through a tree dissimilarity (\hyperref[app:simulation_details]{Appendix~\ref*{app:simulation_details}}).

In all five spaces, \(X\sim N(0,I_3)\), treatment follows a bounded logistic propensity score, and the alternative applies a fixed object-space shift. Auxiliary design information fixes anchors, radii, and trimming masks before main-sample score evaluation, and the formal statistic is evaluated only on \(\mathcal G_n^{\mathrm{inf}}\):
\[
T_n^{\max}
=
\max_{g\in\mathcal G_n^{\mathrm{inf}}}
\frac{\sqrt n|\hat\Delta(g)|}{\hat\sigma(g)}.
\]

\hyperref[tab:main_finite_grid_calibration]{Table~\ref*{tab:main_finite_grid_calibration}} shows conservative or near-nominal null calibration across the five spaces. Medoid MaxMin spreads anchors using only the geometry of the auxiliary objects; Guarded Shell-Signal also uses pilot estimates of the standardized signal within distance shells (\hyperref[tab:anchor_strategy_catalog]{Table~\ref*{tab:anchor_strategy_catalog}}). Across the two designs, null rejection stays below \(0.05\) in every space, and power differs only for trees (\(0.900\) versus \(0.944\)); in the insurance benchmark below, power differs more between designs.

\begin{table}[t]
\centering
\caption{Five-space finite-grid calibration at \(n=3200\) (kernel-weighted ECDF learner, \(500\) repetitions per cell). Size, Power: rejection rates of the level-\(0.05\) max test; SCB: null coverage of the \(90\%\) simultaneous band on \(\mathcal G_n^{\mathrm{inf}}\).}
\label{tab:main_finite_grid_calibration}
\small
\setlength{\tabcolsep}{0pt}
\begin{tabular*}{0.8\textwidth}{@{\extracolsep{\fill}}lcccccc@{}}
\toprule
& \multicolumn{3}{c}{Medoid MaxMin} & \multicolumn{3}{c}{Guarded Shell-Signal} \\
\cmidrule(lr){2-4}\cmidrule(lr){5-7}
Space & Size & Power & SCB & Size & Power & SCB \\
\midrule
Euclidean & 0.028 & 1.000 & 0.936 & 0.040 & 1.000 & 0.910 \\
Wasserstein & 0.022 & 1.000 & 0.956 & 0.042 & 1.000 & 0.922 \\
SPD & 0.012 & 1.000 & 0.966 & 0.034 & 1.000 & 0.924 \\
Graph & 0.042 & 1.000 & 0.922 & 0.044 & 1.000 & 0.896 \\
Tree & 0.040 & 0.900 & 0.924 & 0.032 & 0.944 & 0.920 \\
\bottomrule
\end{tabular*}
\end{table}

We then increase the sample size in a separate bivariate Euclidean design under the DRF weighted ECDF learner at \(n=3200,6400,12800\). Under the null, the max test rejects at \(0.046\), \(0.052\), and \(0.056\), while simultaneous-band coverage remains near the nominal \(95\%\).

Under the controlled alternative, rejection is already saturated at \(1.000\), so the informative trend is not increasing power but decreasing estimation error and band width. Weighted \(L^2\) error falls from \(0.0109\) to \(0.0051\), sup-norm error from \(0.0317\) to \(0.0149\), and mean band width from \(0.0806\) to \(0.0378\). Errors are measured against a Monte Carlo approximation of the true profile from \(100{,}000\) oracle draws. \hyperref[app:simulation_details]{Appendix~\ref*{app:simulation_details}} gives this design, and Figures~\ref{fig:app_convergence_diagnostics} and~\ref{fig:app_null_calibration_by_n} plot these results.

The empirical nuisance diagnostics of \hyperref[app:learner_robustness]{Appendix~\ref*{app:learner_robustness}} decrease with \(n\): \(\hat r_m\), the largest out-of-fold root mean squared error of the conditional-CDF fits, falls from \(0.193\) to \(0.120\), and the squared-rate diagnostic from \(0.796\) to \(0.305\). The scaled remainder also decreases, from \(1.554\) to \(1.058\), but remains of order one at these sample sizes.

\subsection{Estimand Separation and Confounding Adjustment}
\label{subsec:estimand_separation}

We next ask whether the profile estimand records causal changes that neighboring summaries may miss. The separating experiments are built from two controlled alternatives: one changes horizontal quantile location, and one collapses an object law toward a central prototype while keeping the mean-quantile contrast weak. In these experiments and in the confounding experiments below, \(n\) is the total sample size; half of each sample is used for nuisance fitting and score evaluation, and the rest for anchor design and trimming.

In the Lin-style location-shift positive control, the mean-quantile proxy rejects with rate \(1.000\), while the DRF profile test rises from \(0.470\) at \(n=200\) to \(0.855\) at \(n=400\). In the Wasserstein collapse setting, which mirrors \hyperref[ex:lin-separation]{Example~\ref*{ex:lin-separation}}, the profile test rejects with rate \(1.000\) at both sample sizes, while the mean-quantile proxy rejects at rates of only \(0.040\)--\(0.045\) (\hyperref[fig:main_estimand_confounding_drf]{Figure~\ref*{fig:main_estimand_confounding_drf}}). In the graph and SPD mean-zero examples, which mirror \hyperref[ex:kurisu-separation]{Example~\ref*{ex:kurisu-separation}}, finite-sample profile rejection reaches \(0.980\) and \(0.990\) at \(n=400\), while the Fr\'echet-mean proxy remains near nominal (\hyperref[fig:app_related_work]{Figure~\ref*{fig:app_related_work}}). These results illustrate the estimand distinction built into the framework: profile effects target probability shifts over metric neighborhoods, not only horizontal quantile deformation or Fr\'echet-mean displacement.

The confounding experiments keep the profile null true while making covariates affect both treatment and object-valued outcomes. In a confounded Wasserstein null at \(n=400\), the doubly robust profile test rejects at \(0.030\), whereas a naive metric two-sample test, which targets the observed-group difference, rejects at \(0.674\). In the graph confounded null, the corresponding rates are \(0.040\) and \(0.236\). The gap reflects the difference between the causal null and the observational null: a metric comparison of the observed groups detects distributional differences, whereas the causal null requires adjustment for confounding.

\begin{figure}[t]
\centering
\includegraphics[width=0.86\textwidth]{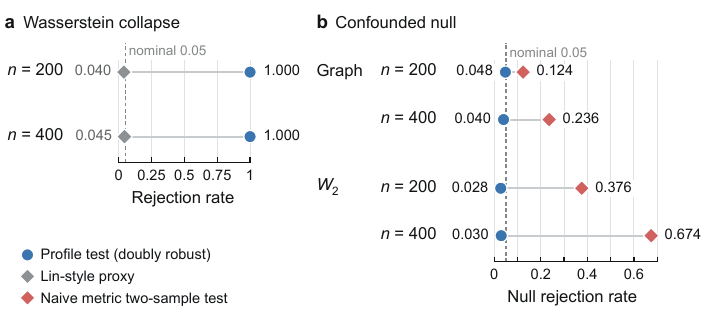}
\caption{Estimand separation and confounding adjustment under DRF. Panel (a) compares the profile test with the mean-quantile (Lin-style) proxy in the Wasserstein collapse setting, with \(200\) repetitions per sample size. Panel (b) compares the doubly robust profile test with a naive metric two-sample test under confounded nulls, with \(500\) repetitions per sample size; \(W_2\) is the Wasserstein space. Here \(n\) is the total sample size, and all tests have level \(0.05\).}
\label{fig:main_estimand_confounding_drf}
\end{figure}

\subsection{Real-Geometry Semi-Synthetic Benchmarks}
\label{subsec:real_geometry}

Finally, we ask whether the same finite-grid machinery can operate on real object geometry while keeping known truth. The ABIDE semi-synthetic experiment uses ABIDE PCP CPAC functional-connectivity derivatives with Craddock-200 ROI time series~\citep{dimartino2014autism,craddock2013neuro,craddock2012whole}, keeps real connectomes and covariates, represents outcomes as SPD connectivity objects, and injects controlled collapse or prototype-shift alternatives (\hyperref[app:realgeom_benchmarks]{Appendix~\ref*{app:realgeom_benchmarks}}).

At effect zero, rejection is conservative at \(0.025\). Under the collapse alternative, rejection rises from \(0.430\) at effect \(0.05\) to \(0.955\) at effect \(0.10\). Under the prototype-shift alternative, rejection rises from \(0.485\) at effect \(0.10\) to \(0.970\) at effect \(0.20\) (\hyperref[fig:main_abide_semisynth_drf]{Figure~\ref*{fig:main_abide_semisynth_drf}}).

In the connectome setting, a profile effect is a change in the probability that a brain connectivity object lies near a lower-dispersion or prototype-centered SPD neighborhood. This is a real-geometry illustration of the anchor--radius language. On this substrate we also examine how dividing a fixed pool between auxiliary and main samples affects power (\hyperref[subsec:allocation]{Appendix~\ref*{subsec:allocation}}).

\begin{figure}[t]
\centering
\includegraphics[width=0.82\textwidth]{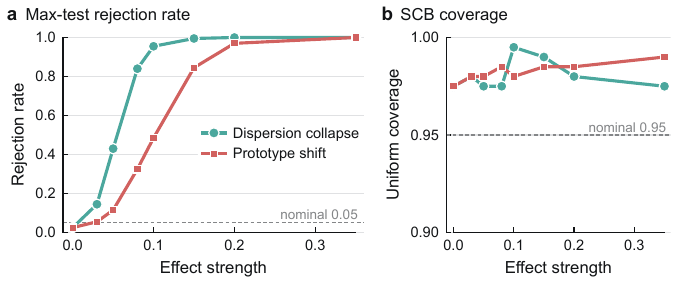}
\caption{ABIDE real-connectome semi-synthetic benchmark under DRF: (a) max-test rejection rate and (b) uniform coverage of the simultaneous band, against effect strength, with \(200\) repetitions per effect strength and level \(0.05\). Coverage is evaluated against the average profile effect over the main-sample units, computed from both potential outcomes.}
\label{fig:main_abide_semisynth_drf}
\end{figure}

The insurance benchmark gives a second real-geometry check using restricted processed data derived from the social-network insurance experiment of~\citet{cai2015social}: the real insurance network, covariates, and observed treatment structure serve as the substrate for graph-valued potential outcomes, compared by the Frobenius distance between adjacency matrices. Under a structural alternative that changes local motif and bridge organization rather than only edge mass, structural power is \(0.416\) for Medoid MaxMin, \(0.562\) for Guarded Shell-Signal, and \(0.586\) for Balanced LHS, a pilot-conditioned rule selected by Latin-hypercube screening.

The corresponding null sizes are conservative at \(0.004\), \(0.018\), and \(0.014\). In this setting, the causal effect has a network interpretation: treatment changes the chance that an ego-network outcome falls near anchor neighborhoods representing different local bridge and motif structures. The benchmark is consistent with the design considerations in \hyperref[app:geometry_anchor_theory]{Appendix~\ref*{app:geometry_anchor_theory}}: more informative pilot-conditioned designs can improve power, while inference remains tied to the finite inferential grid.

\hyperref[app:real_object_related_work]{Appendix~\ref*{app:real_object_related_work}} also reports descriptive profile analyses of observational ABIDE and NHANES data.

\section{Discussion and Limitations}
\label{sec:discussion}

We have developed a profile-based framework for causal inference with object-valued outcomes in separable metric spaces. Each anchor--radius coordinate is a binary potential-outcome contrast with standard identification and doubly robust estimation. What is specific to object-valued outcomes is the joint problem: calibrating the coordinates simultaneously, designing the grid on which they are compared, and reading the resulting bands as probability statements about balls and shells.

The guarantees cover finite grids fixed by a deterministic or auxiliary design and, for each fixed law, all radii at a fixed set of anchors. Inference over a continuum of anchors, uniformity over laws approaching degeneracy, and valid same-sample adaptive anchor search are open problems. The radius direction is tractable because, at a fixed anchor, the ball indicators are nested in the radius, which controls their complexity without conditions on the space (\hyperref[app:complete_radius]{Appendix~\ref*{app:complete_radius}}). Across anchors there is no such ordering: a band over a continuum of anchors would require entropy control of the class of metric balls, which can fail in infinite-dimensional spaces, conditional distance-CDF nuisances that are accurate uniformly in the anchor, and a variance floor over the whole domain. \hyperref[prop:grid-approx-anchor]{Proposition~\ref*{prop:grid-approx-anchor}} bounds how well a finite anchor net approximates the population map, but it does not control estimation error between net points. When auxiliary and analysis samples are drawn from one fixed pool, the pool must be divided between design and analysis; we examine the resulting tradeoff on the connectome substrate in \hyperref[subsec:allocation]{Appendix~\ref*{subsec:allocation}}. In discrete object spaces such as graphs and trees, trimming is part of the inferential target, and sharper theory for discrete outcome spaces is left for future work.

Finitely many profile coordinates do not determine the outcome law: an effect can lie between grid radii, and different laws can have identical complete profiles at finitely many anchors (\hyperref[app:effect_outputs]{Appendix~\ref*{app:effect_outputs}}). A test that does not reject is therefore not evidence that the potential-outcome laws are equal. With full anchor families, distance profiles separate laws under the strong-negative-type condition of \hyperref[sec:defs]{Section~\ref*{sec:defs}}; weaker and more local separation conditions are a natural next step.

The experiments show what each procedure responds to. When treatment shifts quantile location, the mean-quantile proxy is more powerful, rejecting at rate \(1.000\) against \(0.470\) for the profile test at \(n=200\); when treatment concentrates the outcome law around a prototype while leaving mean-level summaries unchanged or nearly so, profile tests detect effects that mean-quantile and Fr\'echet-mean proxies miss. Under confounded causal nulls, doubly robust adjustment keeps rejection near the nominal level, whereas a naive metric two-sample test, which targets the observational null, rejects frequently. On semi-synthetic outcomes built from real connectome and insurance-network geometry, the test is conservative under the null, and in the insurance benchmark its power depends on the anchor design. More broadly, distance profiles extend potential-outcome reasoning to heterogeneous random objects without imposing a common algebraic structure on the outcome space.

\bibliography{example_paper}
\bibliographystyle{unsrtnat}

\appendix

\section{Full Regularity Conditions and Proofs}
\label{app:regularity_proofs}

This appendix contains the full assumptions and proof details behind \hyperref[thm:finite_grid_inference]{Theorem~\ref*{thm:finite_grid_inference}} and \hyperref[cor:main_scb_global]{Corollary~\ref*{cor:main_scb_global}}. The main text states the regularity conditions as (A1)--(A5); here we give the corresponding formal assumptions, the auxiliary-conditional finite-grid theorem, and the proofs of identification, double robustness, Gaussian approximation, simultaneous bands, and max-statistic testing. Conditions (A1)--(A5) correspond to the formal assumptions as follows.
\begin{itemize}[leftmargin=*, label={}]
\item (A1): \hyperref[ass:id]{Assumption~\ref*{ass:id}};
\item (A2): Assumptions~\ref{ass:grid_det} and~\ref{ass:meas_design};
\item (A3): \hyperref[ass:uniform_remainder]{Assumption~\ref*{ass:uniform_remainder}};
\item (A4): the event \(\mathcal A_n\), the variance bound, and the dimension condition of \hyperref[ass:grid_det]{Assumption~\ref*{ass:grid_det}};
\item (A5): \hyperref[ass:gauss_boot]{Assumption~\ref*{ass:gauss_boot}}.
\end{itemize}

\subsection{Causal Identification Assumptions}
\paragraph{Setup and causal assumptions.}
The pointwise identification results below require only the standard potential-outcome assumptions adapted to the profile indicators.

\begin{assumption}[Identification]\label{ass:id}
\leavevmode
\begin{enumerate}[label=(\alph*), leftmargin=*]
    \item \textbf{SUTVA (consistency and no interference).}
    For each $a\in\{0,1\}$, $Y=Y(a)$ almost surely on the event $\{A=a\}$, and one unit's potential outcomes do not depend on other units' treatment assignments.

    \item \textbf{Conditional exchangeability.}
    For each $a\in\{0,1\}$,
    \[
    Y(a)\ind A\mid X.
    \]
    Consequently, for every fixed $g=(\omega,t)\in\mathcal{G}$,
    \[
    Z^{(a)}(g)\ind A\mid X.
    \]

    \item \textbf{Positivity.}
    There exists $\epsilon>0$ such that
    \[
    \epsilon\le\Prob(A=1\mid X)\le1-\epsilon
    \]
    almost surely.
\end{enumerate}
\end{assumption}

\begin{remark}[Pointwise identification does not require strong negative type]\label{rem:snt_not_ident}
\hyperref[ass:id]{Assumption~\ref*{ass:id}} is sufficient for identifying each fixed-anchor--radius functional $\Delta(\omega,t)$. Strong negative type, defined in \hyperref[def:snt]{Definition~\ref*{def:snt}}, plays a different role: it clarifies when a sufficiently rich full profile family can separate probability laws. We therefore use strong negative type only for global profile interpretation, not as a prerequisite for the pointwise identification and estimation developed below.
\end{remark}

\subsection{Asymptotic Theory on Inferential Grids}

Define the centered oracle score, equivalently the pointwise efficient influence function, at evaluation point $g \in \mathcal{G}_n^{\mathrm{inf}}$ as
\[
\psi(O; g) := \phi(O; g, \eta_0) - \Delta(g),
\]
and let
\[
\psi_{\mathcal{G}_n^{\mathrm{inf}}}(O)
:=
\bigl(\psi(O; g) : g \in \mathcal{G}_n^{\mathrm{inf}}\bigr)^\top
\]
be the corresponding score vector.

Let
\[
\sigma^2(g):=\Var\{\psi(O;g)\mid\mathcal F_n^{\mathrm{aux}}\},
\qquad g\in\mathcal G_n^{\mathrm{inf}},
\]
with deterministic grids covered as the special case in which the conditioning sigma-field is trivial.

\begin{assumption}[Inferential grid and nondegeneracy]\label{ass:grid_det}
The design-exogenous candidate grid $\mathcal{G}_n^{\mathrm{cand}}$ is \(\mathcal F_n^{\mathrm{aux}}\)-measurable and satisfies the dimension-growth condition
\[
\log^7\!\bigl(|\mathcal{G}_n^{\mathrm{cand}}|\,n\bigr)=o(n).
\]
Let $\mathcal{G}_n^{\mathrm{inf}}$ be the inferential grid from \hyperref[def:grid_defs]{Definition~\ref*{def:grid_defs}}. We assume there exists an event $\mathcal A_n$, measurable with respect to \(\mathcal F_n^{\mathrm{aux}}\) and satisfying $\Pr(\mathcal A_n)\to1$, such that on $\mathcal A_n$,
\[
1\le |\mathcal G_n^{\mathrm{inf}}|,
\qquad
\inf_{g\in\mathcal G_n^{\mathrm{inf}}}\sigma(g)\ge c_\sigma>0.
\]
\end{assumption}

\begin{assumption}[Measurability and design exogeneity]\label{ass:meas_design}
The outcome space \((\Omega,d)\) is separable and the random objects
\(Y,Y(0),Y(1)\) are Borel-measurable. The auxiliary sigma-field
\(\mathcal F_n^{\mathrm{aux}}\), generated by the design/pilot/reference
information used to construct anchors, radii, and trimming masks,
is independent of the main analysis sample. The candidate and
inferential grids are \(\mathcal F_n^{\mathrm{aux}}\)-measurable and
admit measurable finite enumerations. All maxima over these
grids are understood through these enumerations. Conditionally on
\(\mathcal F_n^{\mathrm{aux}}\), the inferential grid is fixed and
finite, and the main sample remains i.i.d. When the auxiliary information
is sample-based, a pilot random element may be used to generate
\(\mathcal F_n^{\mathrm{aux}}\).
\end{assumption}

\begin{remark}[Variance lower bound]\label{rem:variance_lower}
Since the residual \(Z(g)-m_A(g;X)\) has conditional mean zero given
\((A,X)\), the law of total variance gives
\[
\begin{aligned}
\sigma^2(g)
={}&
\Var\{m_1(g;X)-m_0(g;X)\}\\
&+
\E\left[
\frac{m_1(g;X)\{1-m_1(g;X)\}}{\pi(X)}
+
\frac{m_0(g;X)\{1-m_0(g;X)\}}{1-\pi(X)}
\right].
\end{aligned}
\]
If, uniformly over retained \(g\), the conditional ball probability
\(m_a(g;X)\) of one arm lies in \([\kappa,1-\kappa]\) on a set of covariate
values of probability at least \(q>0\), then
\(\sigma^2(g)\ge q\kappa(1-\kappa)\). The pilot trimming of
\hyperref[sec:defs]{Section~\ref*{sec:defs}} controls only marginal ball probabilities and does
not by itself ensure this condition.
\end{remark}

\begin{assumption}[Uniform nuisance rates]\label{ass:uniform_remainder}
The cross-fitting partition is deterministic or generated independently of the observations. The number of folds \(J\) is fixed, and the fold sizes satisfy \(\min_{j\le J}|\mathcal I_j|\ge c_J n\) for some constant \(c_J>0\). For every validation fold \(j\), the nuisance estimators \(\hat m_a^{(-j)}\) and \(\hat\pi^{(-j)}\) are trained without observations in \(\mathcal I_j\). The fitted propensity scores used in score evaluation are clipped at a fixed level \(\epsilon_{\rm clip}\in(0,\epsilon)\), where \(\epsilon\) is the positivity margin in \hyperref[ass:id]{Assumption~\ref*{ass:id}}, so that
\[
\epsilon_{\rm clip}\le \hat\pi^{(-j)}(X)\le 1-\epsilon_{\rm clip}
\quad\text{almost surely},
\]
and the fitted conditional profile regressions satisfy
\(0\le \hat m_a^{(-j)}(g;X)\le 1\). Define the fold-wise rates
\[
r_{m,n}
:=
\max_{j\le J}\max_{a\in\{0,1\}}\max_{g\in\mathcal G_n^{\mathrm{cand}}}
\|\hat m_a^{(-j)}(g;\cdot)-m_a(g;\cdot)\|_{P,2},
\qquad
r_{\pi,n}
:=
\max_{j\le J}\|\hat\pi^{(-j)}-\pi\|_{P,2},
\]
where \(\hat\pi^{(-j)}\) denotes the clipped fitted propensity actually used in score evaluation.
Let \(L_n=\log(N_n^{\mathrm{cand}}\vee2)\). We assume that these random rates converge uniformly over the candidate grid in probability:
\[
(r_{m,n}^2 + r_{\pi,n}^2)L_n^2
=o_P(1),
\qquad
\sqrt{n}\, r_{m,n}r_{\pi,n}\sqrt{L_n}
=o_P(1).
\]
\end{assumption}

\begin{remark}[Uniform remainder control]\label{rem:remainder_control}
\hyperref[ass:uniform_remainder]{Assumption~\ref*{ass:uniform_remainder}} adapts standard sample-splitting guarantees \citep{chernozhukov2018double} to a growing candidate grid. As established in \hyperref[app:proofs_est]{Appendix~\ref*{app:proofs_est}}, these rate conditions make the orthogonality bias and the empirical-process fluctuation uniformly negligible: \(\sup_{g\in\mathcal G_n^{\mathrm{cand}}}|\sqrt n\{\hat\Delta(g)-\Delta(g)\}-n^{-1/2}\sum_{i=1}^n\psi(O_i;g)|=o_P(L_n^{-1/2})\). Cross-fitting separates nuisance training from score evaluation; the linearization also uses Neyman orthogonality, estimated-overlap clipping, and the finite-grid product-rate conditions above. For instance, if the estimators achieve polynomial rates \(r_{m,n}=O_P(n^{-a})\) and \(r_{\pi,n}=O_P(n^{-b})\) with \(a+b>1/2\), this assumption holds provided
\[
L_n
=
o\!\left(n^{\min\{a,b,\,2(a+b-1/2)\}}\right).
\]
\end{remark}

For the feasible studentized process, center the cross-fitted scores at
their pooled mean \(\hat\Delta(g)=n^{-1}\sum_{i=1}^n\hat\phi_i(g)\) and define
\[
\hat\sigma^2(g)
:=
\frac1n\sum_{i=1}^n
\{\hat\phi_i(g)-\hat\Delta(g)\}^2,
\qquad g\in\mathcal G_n^{\mathrm{inf}}.
\]

For \(g,g'\in\mathcal G_n^{\mathrm{inf}}\), let \(\Gamma_n\) be the auxiliary-conditional oracle correlation matrix with entries
\[
\Gamma_n(g,g')
=
\E\left[
\frac{\psi(O;g)}{\sigma(g)}
\frac{\psi(O;g')}{\sigma(g')}
\;\middle|\;
\mathcal F_n^{\mathrm{aux}}
\right],
\]
let \(\hat\Gamma_n\) be the empirical correlation matrix of the centered estimated scores, with entries
\[
\hat\Gamma_n(g,g')
=
\frac1n
\sum_{i=1}^n
\frac{\hat\phi_i(g)-\hat\Delta(g)}{\hat\sigma(g)}
\frac{\hat\phi_i(g')-\hat\Delta(g')}{\hat\sigma(g')},
\]
and write \(\|M\|_{\max}=\max_{g,g'}|M(g,g')|\) for a matrix indexed by \(\mathcal G_n^{\mathrm{inf}}\).

\begin{assumption}[Correlation consistency]\label{ass:gauss_boot}
On the event \(\mathcal A_n\) of \hyperref[ass:grid_det]{Assumption~\ref*{ass:grid_det}},
\(\|\hat\Gamma_n-\Gamma_n\|_{\max}L_n^2=o_P(1)\).
\end{assumption}

The remaining regularity used below follows from the preceding assumptions. On \(\mathcal A_n\) and conditionally on \(\mathcal F_n^{\mathrm{aux}}\), the standardized oracle scores \(\psi(O_i;g)/\sigma(g)\) are i.i.d.\ across \(i\), centered, and bounded by a constant \(C_\psi\) that depends only on the positivity margin and on \(c_\sigma\) (Assumptions~\ref{ass:id}, \ref{ass:grid_det}, and~\ref{ass:meas_design}); the relative studentization error satisfies \(\max_g|\hat\sigma(g)/\sigma(g)-1|L_n=o_P(1)\) by \eqref{eq:correct_scale_rate}; and, because \(Z_{\mathcal G_n^{\mathrm{inf}}}\sim N(0,\Gamma_n)\) has unit variances, \(V_n=\max_{g\in\mathcal G_n^{\mathrm{inf}}}|Z_{\mathcal G_n^{\mathrm{inf}}}(g)|\) satisfies
\begin{equation}\label{eq:anticoncentration}
\sup_{t\in\mathbb R}\Pr\{|V_n-t|\le\varepsilon\mid\mathcal F_n^{\mathrm{aux}}\}\le C_{\rm ac}\,\varepsilon\sqrt{L_n},\qquad\varepsilon>0,
\end{equation}
for an absolute constant \(C_{\rm ac}\) \citep[Theorem~3 and Comment~5]{cck2015comparison}. \hyperref[cor:primitive_grid]{Corollary~\ref*{cor:primitive_grid}} gives nuisance-rate conditions under which \hyperref[ass:gauss_boot]{Assumption~\ref*{ass:gauss_boot}} holds.

\begin{theorem}[Auxiliary-conditional Gaussian approximation]\label{thm:clt}
Under Assumptions~\ref{ass:id}, \ref{ass:grid_det}, \ref{ass:meas_design}, and~\ref{ass:uniform_remainder}, let
\(p_n=|\mathcal G_n^{\mathrm{inf}}|\). On the event
\(\mathcal A_n\) in \hyperref[ass:grid_det]{Assumption~\ref*{ass:grid_det}}, \(1\le p_n\le N_n^{\mathrm{cand}}\).
All conditional statements below are made on \(\mathcal A_n\), and the
displayed conditional errors converge to zero in probability under the
joint law. The cross-fitted estimator on the inferential grid satisfies
\[
\sup_{g\in\mathcal G_n^{\mathrm{inf}}}
\left|
\sqrt n\bigl(\hat\Delta(g)-\Delta(g)\bigr)
-
\frac{1}{\sqrt n}\sum_{i=1}^n\psi(O_i;g)
\right|
=
o_P\!\bigl(L_n^{-1/2}\bigr).
\]
Furthermore, for the standardized maximum
\[
T_n^{\max,\circ}
:=
\max_{g\in\mathcal G_n^{\mathrm{inf}}}
\frac{\sqrt n\,|\hat\Delta(g)-\Delta(g)|}{\sigma(g)},
\]
there exists a centered Gaussian vector
\[
Z_{\mathcal G_n^{\mathrm{inf}}}
\sim
N(0,\Gamma_n)
\]
such that
\[
\sup_{t\in\mathbb R}
\left|
\Pr\!\left(T_n^{\max,\circ}\le t \mid \mathcal F_n^{\mathrm{aux}}\right)
-
\Pr\!\left(
\max_{g\in\mathcal G_n^{\mathrm{inf}}}
|Z_{\mathcal G_n^{\mathrm{inf}}}(g)|
\le t
\;\middle|\;
\mathcal F_n^{\mathrm{aux}}
\right)
\right|
\to_P 0.
\]
\end{theorem}

In the notation of \hyperref[thm:finite_grid_inference]{Theorem~\ref*{thm:finite_grid_inference}}, \(T_n^\circ=T_n^{\max,\circ}\) and \(Z_n=Z_{\mathcal G_n^{\mathrm{inf}}}\).

Simultaneous inference is calibrated by a Gaussian multiplier bootstrap applied to the centered cross-fitted scores.

\begin{corollary}[Simultaneous confidence bands via multiplier bootstrap]\label{cor:scb}
Under the assumptions and event convention of \hyperref[thm:clt]{Theorem~\ref*{thm:clt}} and \hyperref[ass:gauss_boot]{Assumption~\ref*{ass:gauss_boot}}, let $\{\xi_i\}_{i=1}^n$ be i.i.d.\ $\mathcal N(0,1)$ multipliers generated independently of the data, and define
\[
\hat Z_n^\star(g)
:=
\frac{1}{\sqrt n}
\sum_{i=1}^n
\xi_i\,
\frac{\hat\phi_i(g)-\hat\Delta(g)}{\hat\sigma(g)},
\qquad
g\in\mathcal G_n^{\mathrm{inf}}.
\]
Let
\[
\hat c_{1-\alpha}^{\mathrm{SCB}}
:=
\inf\left\{
c:
\Pr_\xi\!\left(
\max_{g\in\mathcal G_n^{\mathrm{inf}}}
|\hat Z_n^\star(g)|
\le c
\mid O_1^n,\mathcal F_n^{\mathrm{aux}}
\right)
\ge 1-\alpha
\right\}.
\]
Then the simultaneous confidence band
\[
\mathcal C_{n,1-\alpha}(g)
:=
\left[
\hat\Delta(g)\pm
\frac{\hat c_{1-\alpha}^{\mathrm{SCB}}\hat\sigma(g)}{\sqrt n}
\right],
\qquad
g\in\mathcal G_n^{\mathrm{inf}},
\]
satisfies
\[
\Pr\!\left(
\Delta(g)\in\mathcal C_{n,1-\alpha}(g)
\ \text{for all }g\in\mathcal G_n^{\mathrm{inf}}
\mid \mathcal F_n^{\mathrm{aux}}
\right)
\to_P 1-\alpha .
\]
The corresponding unconditional coverage statement follows by averaging over the auxiliary information.
\end{corollary}

\subsection{Global Omnibus Testing via a Max-Statistic}\label{subsec:global_test}

To align with the high-dimensional Gaussian approximation used in \hyperref[sec:est]{Section~\ref*{sec:est}}, we formulate the formal global test with a supremum-type statistic on the inferential grid.

We test the grid-indexed null hypothesis $H_0^{\max}:\ \Delta(g)=0$ for all $g\in\mathcal G_n^{\mathrm{inf}}$. The corresponding omnibus max-statistic is:
\[
T_n^{\max} := \max_{g\in\mathcal G_n^{\mathrm{inf}}} \frac{\sqrt n\,|\hat\Delta(g)|}{\hat\sigma(g)}.
\]

\begin{theorem}[Auxiliary-conditional max-statistic global test]\label{thm:global_test}
Under the assumptions and event convention of \hyperref[cor:scb]{Corollary~\ref*{cor:scb}}, and under the null hypothesis $H_0^{\max}$, we have, conditionally on the auxiliary sigma-field,
\[
\sup_{t\in\mathbb R}
\left|
\Pr\!\left(T_n^{\max}\le t\mid\mathcal F_n^{\mathrm{aux}}\right)
-
\Pr\!\left(\max_{g\in\mathcal G_n^{\mathrm{inf}}}|Z_{\mathcal G_n^{\mathrm{inf}}}(g)|\le t\mid\mathcal F_n^{\mathrm{aux}}\right)
\right|
\to_P 0.
\]
Let \(\hat c_{1-\alpha}^{\max}:=\hat c_{1-\alpha}^{\mathrm{SCB}}\). Then the max test $\varphi_n^{\max}:=\mathds{1}\{T_n^{\max}>\hat c_{1-\alpha}^{\max}\}$ satisfies
\[
\Pr\bigl(\varphi_n^{\max}=1\mid\mathcal F_n^{\mathrm{aux}}\bigr)\to_P\alpha.
\]
\end{theorem}

\begin{remark}[Familywise error]\label{rem:fwer}
Reject the coordinate null \(\Delta(g)=0\) when \(0\notin\mathcal C_{n,1-\alpha}(g)\). Because \(\hat c_{1-\alpha}^{\max}=\hat c_{1-\alpha}^{\mathrm{SCB}}\), the max test rejects exactly when some coordinate is rejected. On the event that \(\mathcal C_{n,1-\alpha}\) covers every \(\Delta(g)\), no coordinate with \(\Delta(g)=0\) is rejected. Hence, for every law satisfying the assumptions and every configuration of true coordinate nulls, the probability of rejecting at least one true null is at most the probability that the band misses some \(\Delta(g)\), which tends to \(\alpha\) by \hyperref[cor:scb]{Corollary~\ref*{cor:scb}}. The coordinatewise rule therefore controls the familywise error rate in the strong sense, asymptotically.
\end{remark}

\subsection{Proofs of Identification and Double Robustness}\label{app:proofs}

Fix $(\omega,t)\in \mathcal{W}\times\mathcal{T}$. Define the observed profile indicator
\[
Z:=Z(\omega,t)=\mathds{1}\{d(\omega,Y)\le t\},
\]
and the potential profile indicators
\[
Z^{(a)}:=Z^{(a)}(\omega,t)=\mathds{1}\{d(\omega,Y(a))\le t\},\qquad a\in\{0,1\}.
\]
Let the nuisance functions be
\[
\pi(x):=\Prob(A=1\mid X=x),\qquad
m_a(x):=\E[Z\mid A=a,X=x].
\]
Our target estimand is the profile causal effect
\[
\Delta(\omega,t):=\E[Z^{(1)}]-\E[Z^{(0)}].
\]

\subsubsection{Identification (Outcome Regression Form)}

\begin{proposition}[Identification]\label{prop:id}
Under \hyperref[ass:id]{Assumption~\ref*{ass:id}} (SUTVA, conditional exchangeability, positivity),
for each $a\in\{0,1\}$,
\[
\E[Z^{(a)}]=\E\big[m_a(X)\big],
\]
and hence
\[
\Delta(\omega,t)=\E\big[m_1(X)-m_0(X)\big].
\]
\end{proposition}

\begin{proof}
By the law of iterated expectations,
\[
\E[Z^{(a)}]=\E\!\left[\E\!\left(Z^{(a)}\mid X\right)\right].
\]
By conditional exchangeability, $Z^{(a)}\ind A\mid X$, so
\[
\E\!\left(Z^{(a)}\mid X\right)=\E\!\left(Z^{(a)}\mid A=a,X\right).
\]
By consistency, on the event $\{A=a\}$ we have $Y=Y(a)$ and thus
$Z(\omega,t)=Z^{(a)}(\omega,t)$, implying
\[
\E\!\left(Z^{(a)}\mid A=a,X\right)=\E\!\left(Z\mid A=a,X\right)=m_a(X).
\]
Combining the displays yields $\E[Z^{(a)}]=\E[m_a(X)]$.
Taking the difference between $a=1$ and $a=0$ gives the result for $\Delta(\omega,t)$.
\end{proof}

\subsubsection{Efficient Influence Function}

Define the (observed-data) functional
\[
\Psi_a(P):=\E\big[m_a(X)\big],\qquad a\in\{0,1\},
\]
so that under \hyperref[prop:id]{Proposition~\ref*{prop:id}}, $\Psi_a(P)=\E[Z^{(a)}]$.
In the nonparametric model for $P(A,X,Y)$, the centered EIF for $\Psi_a(P)$ is
\[
\psi_a(O)
=
\frac{\mathds{1}(A=a)}{\pi_a(X)}\{Z-m_a(X)\}+m_a(X)-\Psi_a(P),
\]
where $\pi_1(X)=\pi(X)$ and $\pi_0(X)=1-\pi(X)$.
Consequently, the centered EIF for $\Delta(\omega,t)=\Psi_1(P)-\Psi_0(P)$ is
\[
\psi(O)=\psi_1(O)-\psi_0(O).
\]
This is the standard EIF for ATE-type functionals; see, e.g., \citep{hahn1998role}.

\subsubsection{Double Robustness}

For $a\in\{0,1\}$, define the uncentered AIPW score with generic nuisances
$\bar\pi,\bar m_a$:
\[
\phi_a(O;\bar m_a,\bar\pi)
=
\frac{\mathds{1}(A=a)}{\bar\pi_a(X)}\{Z-\bar m_a(X)\}+\bar m_a(X).
\]

\begin{proposition}[Double robustness]\label{prop:dr}
Assume $\bar\pi_a(X)$ is bounded away from $0$ almost surely.
Then for each $a\in\{0,1\}$,
\[
\E\big[\phi_a(O;\bar m_a,\bar\pi)\big]-\Psi_a(P)
=
\E\left[
\frac{\{m_a(X)-\bar m_a(X)\}\{\pi_a(X)-\bar\pi_a(X)\}}{\bar\pi_a(X)}
\right].
\]
In particular, $\E[\phi_a(O;\bar m_a,\bar\pi)]=\Psi_a(P)$ if either
(i) $\bar m_a=m_a$ a.s.\ or (ii) $\bar\pi=\pi$ a.s.
Consequently, under \hyperref[ass:id]{Assumption~\ref*{ass:id}}, $\E[\phi_1-\phi_0]=\Delta(\omega,t)$ if $\bar\pi=\pi$ a.s., or if $\bar m_0=m_0$ and $\bar m_1=m_1$ a.s.
\end{proposition}

\begin{proof}
We prove the case $a=1$; the case $a=0$ is analogous.
Condition on $X$ and apply iterated expectation:
\begin{align*}
\E\big[\phi_1(O;\bar m_1,\bar\pi)\mid X\big]
&=
\E\left[\frac{A}{\bar\pi(X)}(Z-\bar m_1(X))+\bar m_1(X)\ \Big|\ X\right]\\
&=
\frac{\E[A\mid X]}{\bar\pi(X)}
\E\big[Z-\bar m_1(X)\mid A=1,X\big]+\bar m_1(X)\\
&=
\frac{\pi(X)}{\bar\pi(X)}\{m_1(X)-\bar m_1(X)\}+\bar m_1(X).
\end{align*}
Subtracting $m_1(X)$ gives
\[
\E[\phi_1\mid X]-m_1(X)
=
\frac{\{m_1(X)-\bar m_1(X)\}\{\pi(X)-\bar\pi(X)\}}{\bar\pi(X)}.
\]
Taking expectations over $X$ yields the bias formula.
Since $\Psi_1(P)=\E[m_1(X)]$, if $\bar m_1=m_1$ a.s.\ or $\bar\pi=\pi$ a.s., the bias is zero.
\end{proof}

The product-bias display also gives Neyman orthogonality. To see this,
fix an arm \(a\) and perturb the nuisance functions along
\(\bar m_a=m_a+s h_m\) and \(\bar\pi_a=\pi_a+s h_\pi\). The bias becomes
\[
\E\left[
\frac{s^2 h_m(X)h_\pi(X)}{\pi_a(X)+s h_\pi(X)}
\right].
\]
Its derivative at \(s=0\) is zero. The contrast score inherits
the same first-order insensitivity because it is the difference of the
two arm-specific scores.

\subsection{Proofs for Estimation and Inference}\label{app:proofs_est}

We write $\psi_{\mathcal G_n^{\mathrm{cand}}}(O) = \bigl(\psi(O; g):g\in\mathcal G_n^{\mathrm{cand}}\bigr)$ for the centered oracle score vector on the design-exogenous candidate grid, and $\psi_{\mathcal G_n^{\mathrm{inf}}}(O)$ for its restriction to the inferential grid.
\subsubsection{Design Conditioning and Finite-Grid Reduction}

\begin{lemma}[Auxiliary conditioning reduces inference to a fixed finite vector]\label{lem:pilot_fixed_vector}
Under Assumptions~\ref{ass:id} and~\ref{ass:meas_design}, conditionally on
\(\mathcal F_n^{\mathrm{aux}}\), all inference statements over
\(\mathcal G_n^{\mathrm{inf}}\) are statements about the finite vector
\[
    \{\psi(O_i;g):g\in\mathcal G_n^{\mathrm{inf}},\ i=1,\ldots,n\}.
\]
For each fixed \(g\), the coordinate has conditional mean zero and
variance \(\sigma^2(g)\).
\end{lemma}

\begin{proof}
Because the grid is measurable with respect to
\(\mathcal F_n^{\mathrm{aux}}\), conditioning on the auxiliary sigma-field
turns \(\mathcal G_n^{\mathrm{inf}}\) into a nonrandom finite set. The
auxiliary sigma-field is independent of the main analysis sample, so the
conditional law of \(O_1,\ldots,O_n\) remains the product law of the
main-sample observations. Write the realized grid as
\(g_1,\ldots,g_{p_n}\). For each realized coordinate \(g_k\),
\[
    \E\{\psi(O;g_k)\mid \mathcal F_n^{\mathrm{aux}}\}
    =
    \E\{\phi(O;g_k,\eta_0)-\Delta(g_k)\}=0,
\]
where the last equality follows from \hyperref[thm:ident]{Theorem~\ref*{thm:ident}}. The variance statement is the definition of
\(\sigma^2(g_k)\). Since the grid is finite after conditioning, these
coordinatewise facts define a finite-dimensional score vector.
\end{proof}

\subsubsection{\texorpdfstring{Uniform Remainder under \hyperref[ass:uniform_remainder]{Assumption~\ref*{ass:uniform_remainder}}}{Uniform Remainder}}\label{app:uniform_remainder_proof}

\begin{lemma}[Conditional finite-class maximal inequality]\label{lem:conditional_finite_maximal}
Let \(\mathcal I\) be a validation fold with \(m=|\mathcal I|\), and
let \(\mathcal F_n=\{f_g:g\in\mathcal H_n\}\) be a finite class with
\(|\mathcal H_n|=N\). Condition on a sigma-field \(\mathcal S\) such
that the observations in \(\mathcal I\) are conditionally i.i.d. and
independent of \(\mathcal S\), while the functions \(f_g\) are
\(\mathcal S\)-measurable. Suppose, conditionally on \(\mathcal S\),
\[
    \E\{f_g(O)\mid\mathcal S\}=0,\qquad
    \|f_g\|_{P,2}\le \rho,\qquad
    \|f_g\|_\infty\le B
\]
for all \(g\in\mathcal H_n\), where \(B<\infty\) is fixed. Then
\[
    \sup_{g\in\mathcal H_n}
    \left|
    m^{-1}\sum_{i\in\mathcal I} f_g(O_i)
    \right|
    =
    O_P\left(
    \rho\sqrt{\frac{\log(N\vee2)}{m}}
    +
    \frac{\log(N\vee2)}{m}
    \right),
\]
where the probability is under the joint law after integrating over
\(\mathcal S\).
\end{lemma}

\begin{proof}
Fix \(g\) and condition on \(\mathcal S\). Bernstein's inequality for
bounded, conditionally independent, mean-zero summands gives, for every
\(x>0\),
\[
\Pr\left(
\left|
m^{-1}\sum_{i\in\mathcal I}f_g(O_i)
\right|
>
C\left\{\rho\sqrt{\frac{x}{m}}+\frac{Bx}{m}\right\}
\ \middle|\ \mathcal S
\right)
\le
2e^{-x},
\]
for a universal constant \(C\). Taking \(x=\log(N\vee2)+u\) and applying
the union bound over the \(N\) coordinates yields
\[
\Pr\left(
\sup_{g\in\mathcal H_n}
\left|
m^{-1}\sum_{i\in\mathcal I}f_g(O_i)
\right|
>
C'\left\{\rho\sqrt{\frac{\log(N\vee2)+u}{m}}
\!+\!
\frac{B\{\log(N\vee2)+u\}}{m}\right\}
\ \middle|\ \mathcal S
\right)
\le
2e^{-u}.
\]
Since \(B\) is fixed, this conditional tail bound implies the asserted
order in conditional probability, and therefore also in joint
probability.
\end{proof}

\begin{lemma}[Finite-grid orthogonal remainder]\label{lem:finite_grid_orthogonal_remainder}
Assume sample splitting with a fixed number of folds. For each validation
fold, nuisance estimators are trained without using observations in that
fold. Suppose the overlap bound in \hyperref[ass:id]{Assumption~\ref*{ass:id}} holds, the
indicator outcomes \(Z(g)\) are bounded, and the fitted propensities and
profile regressions are bounded as in \hyperref[ass:uniform_remainder]{Assumption~\ref*{ass:uniform_remainder}}.
Suppose further that the uniform nuisance rates in \hyperref[ass:uniform_remainder]{Assumption~\ref*{ass:uniform_remainder}} hold over
\(\mathcal G_n^{\mathrm{cand}}\), and that the dimension-growth condition 
in \hyperref[ass:grid_det]{Assumption~\ref*{ass:grid_det}} holds. Then
\[
\sup_{g\in\mathcal G_n^{\mathrm{cand}}}
\left|
\sqrt n\{\hat\Delta(g)-\Delta(g)\}
-
n^{-1/2}\sum_{i=1}^n\psi(O_i;g)
\right|
=
o_P(L_n^{-1/2}).
\]
\end{lemma}

\begin{proof}
Fix a fold \(j\) and let \(\mathcal S_j\) be the sigma-field generated by
\(\mathcal F_n^{\mathrm{aux}}\), the fold partition, the training
observations of fold \(j\), and any internal randomization of the
learners. Given \(\mathcal S_j\), the grid and the fitted nuisance tuple
\(\hat\eta^{(-j)}\) are fixed, and the validation observations in
\(\mathcal I_j\) are i.i.d.\ and independent of them. Write
\[
    e_{jg}(O)=\phi(O;g,\hat\eta^{(-j)})-\phi(O;g,\eta_0),
\]
and let \(P\) integrate over a new observation \(O\) with \(\mathcal S_j\)
held fixed, and let \(P_{n,j}\) be the empirical measure of fold \(j\). The fold contribution \(P_{n,j}e_{jg}\) splits into the
conditional bias \(B_{n,j}(g)=Pe_{jg}\) and the centered
empirical-process term \(E_{n,j}(g)=(P_{n,j}-P)e_{jg}\).

For the bias term, the AIPW algebra gives, uniformly in \(g\),
\[
\begin{aligned}
B_{n,j}(g)
&=
P\left[
\frac{\{\hat\pi(X)-\pi(X)\}\{\hat m_1(g;X)-m_1(g;X)\}}{\hat\pi(X)}
\right]\\
&\quad +
P\left[
\frac{\{\hat\pi(X)-\pi(X)\}\{\hat m_0(g;X)-m_0(g;X)\}}{1-\hat\pi(X)}
\right].
\end{aligned}
\]
True overlap and estimated clipping bound the
denominators away from zero, so
Cauchy--Schwarz yields
\[
    \sup_g |B_{n,j}(g)| \le C r_{\pi,n}r_{m,n}.
\]
Multiplying by \(\sqrt n\) and by
\(\sqrt{L_n}\) gives a vanishing bound by
\hyperref[ass:uniform_remainder]{Assumption~\ref*{ass:uniform_remainder}}.

For the centered empirical-process term, let
\(\rho_{n,j}=\max_{a,g}\|\hat m_a^{(-j)}(g;\cdot)-m_a(g;\cdot)\|_{P,2}+\|\hat\pi^{(-j)}-\pi\|_{P,2}\)
be the error of the nuisance fit used on fold \(j\); it is fixed given \(\mathcal S_j\) and satisfies \(\rho_{n,j}\le r_{m,n}+r_{\pi,n}\). The centered score
perturbation \(e_{jg}-Pe_{jg}\) is, conditionally on \(\mathcal S_j\), a bounded mean-zero
validation-fold function whose \(L_2(P)\) norm is bounded by
\(C\rho_{n,j}\), uniformly over the grid. The boundedness and
the Lipschitz \(L_2(P)\) bound follow from the clipped propensity
denominators, the bounded profile regressions, and the fact that
\(Z(g)\in[0,1]\). Applying
\hyperref[lem:conditional_finite_maximal]{Lemma~\ref*{lem:conditional_finite_maximal}} with
\(\rho=C\rho_{n,j}\) and using
\(|\mathcal I_j|\ge c_J n\) gives
\[
    \sup_g |E_{n,j}(g)|
    =
    O_P\!\left\{
    (r_{m,n}+r_{\pi,n})
    \sqrt{\frac{L_n}{n}}
    +
    \frac{L_n}{n}
    \right\}.
\]
The squared-rate logarithmic condition in
\hyperref[ass:uniform_remainder]{Assumption~\ref*{ass:uniform_remainder}} implies that the first term,
after multiplication by \(\sqrt n\), is
\(o_P(L_n^{-1/2})\). The dimension
condition in \hyperref[ass:grid_det]{Assumption~\ref*{ass:grid_det}} implies
\(L_n^{3/2}/\sqrt n=o(1)\), so the
second term has the same negligible scale. Hence
\(\sup_g\sqrt n |E_{n,j}(g)|
=o_P(L_n^{-1/2})\). The number of
folds is fixed, so taking the maximum over folds preserves the same
order. Combining the bias and empirical-process bounds proves the
claimed uniform linearization.
\end{proof}

\subsubsection{Studentization and Covariance Control}\label{app:primitive_covariance}

When the grid grows, the feasible studentized maximum stays close to its
oracle counterpart provided the relative error of the estimated standard
deviations is small compared with the reciprocal of the logarithmic
dimension. The boundedness of the AIPW score delivers this rate
directly.

\begin{lemma}[Estimated-score covariance on a finite grid]
\label{lem:primitive_covariance}
Suppose Assumptions~\ref{ass:id}, \ref{ass:grid_det}, \ref{ass:meas_design},
and~\ref{ass:uniform_remainder} hold, and let \(\rho_n=r_{m,n}+r_{\pi,n}\) and
\(L_n=\log(N_n^{\mathrm{cand}}\vee2)\). Let \(\hat\Sigma_n\) be the
empirical covariance matrix of the centered cross-fitted scores on
\(\mathcal G_n^{\mathrm{inf}}\), and let \(\Sigma_n\) be the conditional
covariance matrix of the oracle scores. Then
\[
\|\hat\Sigma_n-\Sigma_n\|_{\max}
=O_P\{\rho_n+\sqrt{L_n/n}+L_n/n\}.
\]
If, in addition, \(\inf_{g\in\mathcal G_n^{\mathrm{inf}}}\sigma(g)\ge
c_\sigma>0\), the same order bounds
\(\max_g|\hat\sigma(g)/\sigma(g)-1|\) and
\(\|\hat\Gamma_n-\Gamma_n\|_{\max}\).
\end{lemma}

\begin{proof}
For fold \(j\), let \(\mathcal S_j\) and \(e_{jg}\) be as in the proof of
\hyperref[lem:finite_grid_orthogonal_remainder]{Lemma~\ref*{lem:finite_grid_orthogonal_remainder}}. Clipping and
bounded regressions give a fixed envelope, and
\(\max_g Pe_{jg}^2\le C\rho_{n,j}^2\), where \(\rho_{n,j}\le\rho_n\) is
the error of the nuisance fit used on fold \(j\). Conditionally on \(\mathcal S_j\), Bernstein's inequality applied to
\(e_{jg}^2\), whose variance is at most \(CPe_{jg}^2\), together with a
union bound over the grid, gives
\[
\max_g P_{n,j}e_{jg}^2
=O_P\{\rho_n^2+\rho_n\sqrt{L_n/n}+L_n/n\}
=O_P(\rho_n^2+L_n/n).
\]
Since the number of folds is fixed, the pooled empirical \(L^2\) distance
between estimated and oracle scores is \(O_P(\rho_n+\sqrt{L_n/n})\)
uniformly in \(g\). By the Cauchy--Schwarz inequality and the boundedness
of the oracle score, the same bound holds for the changes in all
empirical first and second moments, and hence in the centered cross
products. For the oracle scores, a union bound over at most
\((N_n^{\mathrm{cand}})^2\) bounded products gives an empirical covariance
error of order \(O_P(\sqrt{L_n/n}+L_n/n)\). Combining the two bounds
proves the first claim. When the oracle standard deviations are bounded
below, the square root and the division by standard deviations are
uniformly Lipschitz on the relevant range, which gives the remaining
claims.
\end{proof}

\hyperref[ass:uniform_remainder]{Assumption~\ref*{ass:uniform_remainder}} implies \(\rho_nL_n=o_P(1)\), and
the dimension condition in \hyperref[ass:grid_det]{Assumption~\ref*{ass:grid_det}} implies
\(L_n^{3/2}/\sqrt n\to0\). \hyperref[lem:primitive_covariance]{Lemma~\ref*{lem:primitive_covariance}}
therefore yields
\begin{equation}\label{eq:correct_scale_rate}
\max_{g\in\mathcal G_n^{\mathrm{inf}}}
\left|\frac{\hat\sigma(g)}{\sigma(g)}-1\right|L_n=o_P(1),
\end{equation}
which is the studentization bound noted after \hyperref[ass:gauss_boot]{Assumption~\ref*{ass:gauss_boot}}. Because the
oracle-studentized maximum is \(O_P(\sqrt{L_n})\), replacing
\(\sigma(g)\) by \(\hat\sigma(g)\) changes it by \(o_P(L_n^{-1/2})\), the
scale controlled by Gaussian anti-concentration.

\begin{corollary}[Explicit sufficient rates]\label{cor:primitive_grid}
Suppose Assumptions~\ref{ass:id}, \ref{ass:grid_det}, and
\ref{ass:meas_design} hold, the cross-fitting partition, clipped
propensities, and bounded regressions are as in
\hyperref[ass:uniform_remainder]{Assumption~\ref*{ass:uniform_remainder}}, and
\[
(r_{m,n}+r_{\pi,n})L_n^2=o_P(1),
\qquad
\sqrt n\,r_{m,n}r_{\pi,n}\sqrt{L_n}=o_P(1).
\]
Then Assumptions~\ref{ass:uniform_remainder} and~\ref{ass:gauss_boot}
hold. Consequently, the conclusions of \hyperref[thm:clt]{Theorem~\ref*{thm:clt}},
\hyperref[cor:scb]{Corollary~\ref*{cor:scb}}, and \hyperref[thm:global_test]{Theorem~\ref*{thm:global_test}} hold under
these rate conditions.
\end{corollary}

\begin{proof}
The first rate condition implies
\((r_{m,n}^2+r_{\pi,n}^2)L_n^2=o_P(1)\), so
\hyperref[ass:uniform_remainder]{Assumption~\ref*{ass:uniform_remainder}} holds.
\hyperref[lem:primitive_covariance]{Lemma~\ref*{lem:primitive_covariance}} gives
\(\|\hat\Gamma_n-\Gamma_n\|_{\max}L_n^2
=O_P\{\rho_nL_n^2+\sqrt{L_n^5/n}+L_n^3/n\}=o_P(1)\), which is
\hyperref[ass:gauss_boot]{Assumption~\ref*{ass:gauss_boot}}.
\end{proof}

\subsubsection{\texorpdfstring{Proof of \hyperref[thm:clt]{Theorem~\ref*{thm:clt}} (High-Dimensional Gaussian Approximation)}{Proof of the Gaussian Approximation}}

\begin{lemma}[Perturbation transfer for Gaussian maxima]\label{lem:max_perturbation_transfer}
Let \(S_n\), \(T_n\), and \(V_n\) be real-valued statistics such that
\[
    \sup_t |\Pr(S_n\le t\mid\mathcal F_n^{\mathrm{aux}})
    -\Pr(V_n\le t\mid\mathcal F_n^{\mathrm{aux}})|\to_P0
\]
and assume there exist \(\mathcal F_n^{\mathrm{aux}}\)-measurable \(\delta_n\ge0\) with
\(\delta_n\sqrt{L_n}\to0\) such that
\[
    \Pr(|T_n-S_n|>\delta_n\mid\mathcal F_n^{\mathrm{aux}})\to_P0.
\]
If \(V_n\) satisfies the
anti-concentration bound \eqref{eq:anticoncentration}, then
\[
    \sup_t |\Pr(T_n\le t\mid\mathcal F_n^{\mathrm{aux}})
    -\Pr(V_n\le t\mid\mathcal F_n^{\mathrm{aux}})|\to_P0.
\]
\end{lemma}

\begin{proof}
Let \(\delta_n\) be as in the statement; it is fixed given
\(\mathcal F_n^{\mathrm{aux}}\). For every \(t\),
\(\{T_n\le t\}\subseteq \{S_n\le t+\delta_n\}\cup
\{|T_n-S_n|>\delta_n\}\) gives
\[
\Pr(T_n\le t\mid\mathcal F_n^{\mathrm{aux}})
\le
\Pr(S_n\le t+\delta_n\mid\mathcal F_n^{\mathrm{aux}})
+
\Pr(|T_n-S_n|>\delta_n\mid\mathcal F_n^{\mathrm{aux}}).
\]
The analogous inclusion
\(\{S_n\le t-\delta_n\}\subseteq\{T_n\le t\}\cup
\{|T_n-S_n|>\delta_n\}\) gives the lower bound. The Gaussian
approximation for \(S_n\) and the anti-concentration bound for \(V_n\)
control
\(\Pr(t-\delta_n<V_n\le t+\delta_n\mid\mathcal F_n^{\mathrm{aux}})\)
by \(2C_{\rm ac}\delta_n\sqrt{L_n}=o(1)\). Taking the
supremum over \(t\) proves the claim.
\end{proof}

\begin{proof}[Proof of Theorem~\ref{thm:clt}]
Condition on the auxiliary sigma-field and restrict attention to the
high-probability event in \hyperref[ass:grid_det]{Assumption~\ref*{ass:grid_det}}. By
\hyperref[lem:pilot_fixed_vector]{Lemma~\ref*{lem:pilot_fixed_vector}}, the inferential grid is a fixed
finite index set and the main-sample score vectors are conditionally
i.i.d. By \hyperref[lem:finite_grid_orthogonal_remainder]{Lemma~\ref*{lem:finite_grid_orthogonal_remainder}},
\[
\sup_{g\in\mathcal G_n^{\mathrm{inf}}}
\left|
\sqrt n\{\hat\Delta(g)-\Delta(g)\}
-
n^{-1/2}\sum_{i=1}^n\psi(O_i;g)
\right|
=o_P(L_n^{-1/2}).
\]
This proves the uniform linearization part of the theorem.

Define the oracle standardized maximum
\[
S_n^{\max}
=
\max_{g\in\mathcal G_n^{\mathrm{inf}}}
\left|
n^{-1/2}\sum_{i=1}^n\frac{\psi(O_i;g)}{\sigma(g)}
\right|.
\]
Enumerate the inferential grid as \(g_1,\ldots,g_{p_n}\). To apply a
coordinatewise maximum theorem to the absolute maximum, define the
signed augmented vector
\[
X_i^{\pm}
=
\left(
\frac{\psi(O_i;g_1)}{\sigma(g_1)},\ldots,
\frac{\psi(O_i;g_{p_n})}{\sigma(g_{p_n})},
-
\frac{\psi(O_i;g_1)}{\sigma(g_1)},\ldots,
-
\frac{\psi(O_i;g_{p_n})}{\sigma(g_{p_n})}
\right).
\]
Then \(S_n^{\max}\) is the coordinatewise maximum of
\(n^{-1/2}\sum_i X_i^{\pm}\). The augmented summands are conditionally
independent, centered, have unit marginal variances, are bounded by
\(C_\psi\), and have
dimension \(2p_n\le 2N_n^{\mathrm{cand}}\). The stated growth condition implies
\(\log^7\{(2N_n^{\mathrm{cand}}\vee2)n\}=o(n)\). The high-dimensional Gaussian
approximation theorem for maxima of sums of independent random vectors
\citep[Theorem 2.2]{chernozhukov2013gaussian} therefore gives
\[
\sup_{t\in\mathbb R}
\left|
\Pr(S_n^{\max}\le t\mid\mathcal F_n^{\mathrm{aux}})
-
\Pr(V_n\le t\mid\mathcal F_n^{\mathrm{aux}})
\right|
\to_P0,
\]
where
\[
V_n=
\max_{g\in\mathcal G_n^{\mathrm{inf}}}
|Z_{\mathcal G_n^{\mathrm{inf}}}(g)|.
\]
The estimator maximum \(T_n^{\max,\circ}\) differs from \(S_n^{\max}\)
by at most the linearization remainder divided by
\(\inf_g\sigma(g)\), which is still
\(o_P(L_n^{-1/2})\). More explicitly, let
\(R_n=|T_n^{\max,\circ}-S_n^{\max}|\). Since \(R_n\sqrt{L_n}\to_P0\), there exists a
deterministic sequence \(b_n\downarrow0\) such that
\(\Pr\{R_n>b_n/\sqrt{L_n}\}\to0\). Set
\(\delta_n=b_n/\sqrt{L_n}\). Then
\(\delta_n\sqrt{L_n}\to0\), and
\[
\Pr\{R_n>\delta_n\mid\mathcal F_n^{\mathrm{aux}}\}\to_P0
\]
because the expectation of this conditional probability is
\(\Pr(R_n>\delta_n)\to0\). Applying
\hyperref[lem:max_perturbation_transfer]{Lemma~\ref*{lem:max_perturbation_transfer}} proves the Gaussian
approximation for \(T_n^{\max,\circ}\).
\end{proof}

\subsubsection{\texorpdfstring{Proof of \hyperref[cor:scb]{Corollary~\ref*{cor:scb}} and \hyperref[thm:global_test]{Theorem~\ref*{thm:global_test}} (Cross-Fitted Multiplier Bootstrap)}{Proof of the Band and Test Results}}

\begin{lemma}[Multiplier max-bootstrap validity]\label{lem:multiplier_max_bootstrap}
Under Assumptions~\ref{ass:grid_det}, \ref{ass:uniform_remainder}, and
\ref{ass:gauss_boot},
\[
    \sup_{t\in\mathbb R}
    \left|
    \Pr_\xi\!\left(
    \max_{g\in\mathcal G_n^{\mathrm{inf}}}|\hat Z_n^\star(g)|
    \le t
    \mid O_1^n,\mathcal F_n^{\mathrm{aux}}
    \right)
    -
    \Pr\!\left(V_n\le t\mid\mathcal F_n^{\mathrm{aux}}\right)
    \right|
    \to_P0,
\]
where
\[
V_n=
\max_{g\in\mathcal G_n^{\mathrm{inf}}}
|Z_{\mathcal G_n^{\mathrm{inf}}}(g)|
\]
and \(\Pr_\xi\) denotes probability over the multiplier variables
conditionally on the observed data and the auxiliary sigma-field.
\end{lemma}

\begin{proof}
Condition on the observed data and the auxiliary sigma-field. The bootstrap statistic is a Gaussian vector
because it is a linear combination of independent standard normal
multipliers. Its conditional covariance matrix is exactly the empirical
correlation matrix \(\hat\Gamma_n\) of the centered estimated
standardized scores. \hyperref[ass:gauss_boot]{Assumption~\ref*{ass:gauss_boot}} gives
\[
    \delta_n=\|\hat\Gamma_n-\Gamma_n\|_{\max},\qquad
    \delta_nL_n^2=o_P(1),
\]
where \(\Gamma_n\) is the oracle correlation matrix. For the absolute
maximum, apply the same comparison to the signed augmented Gaussian
vectors \((Z,-Z)\) and \((\hat Z^\star,-\hat Z^\star)\) of dimension
\(2p_n\); their covariance matrices also differ by \(\delta_n\) in max
norm. For unit-variance Gaussian vectors, Gaussian comparison for maxima
\citep[Lemma 3.1]{chernozhukov2013gaussian} bounds the Kolmogorov
distance between the two maxima by
\(C\delta_n^{1/3}\{1\vee\log(2p_n/\delta_n)\}^{2/3}\), read as zero when
\(\delta_n=0\), which tends to zero in probability because
\(\delta_nL_n^2=o_P(1)\). This gives the
displayed conditional Kolmogorov convergence. The conversion from this
distributional comparison to bootstrap critical-value consistency is
recorded in \hyperref[lem:bootstrap_quantile_transfer]{Lemma~\ref*{lem:bootstrap_quantile_transfer}}.
\end{proof}

\begin{lemma}[Conditional bootstrap quantile transfer]\label{lem:bootstrap_quantile_transfer}
Let
\[
F_n(t)=\Pr(V_n\le t\mid\mathcal F_n^{\mathrm{aux}})
\]
and let \(\hat F_n(t)\) be a right-continuous random conditional
distribution function, measurable with respect to the main sample and
\(\mathcal F_n^{\mathrm{aux}}\), such
that
\[
    \sup_{t\in\mathbb R}|\hat F_n(t)-F_n(t)|\to_P0.
\]
Assume that \(F_n\) satisfies, on the same high-probability auxiliary event,
\[
    \sup_{t\in\mathbb R}\{F_n(t+\varepsilon)-F_n(t-\varepsilon)\}
    \le C_{\rm ac}\varepsilon\sqrt{L_n}
    \qquad\text{for all }\varepsilon>0.
\]
For fixed \(\alpha\in(0,1)\), define the generalized inverse quantile
\[
    \hat c_{1-\alpha}
    =
    \inf\{t:\hat F_n(t)\ge 1-\alpha\}.
\]
Then
\[
    F_n(\hat c_{1-\alpha})=1-\alpha+o_P(1).
\]
\end{lemma}

\begin{proof}
Let \(a_n=\sup_t|\hat F_n(t)-F_n(t)|\). By assumption \(a_n=o_P(1)\).
By the generalized-inverse definition of \(\hat c_{1-\alpha}\) and the
right-continuity of \(\hat F_n\),
\[
    F_n(\hat c_{1-\alpha})
    \ge
    \hat F_n(\hat c_{1-\alpha})-a_n
    \ge
    1-\alpha-a_n.
\]
For the upper bound, fix any \(\varepsilon_n\downarrow0\) such that
\(\varepsilon_n\sqrt{L_n}\to0\). By the definition of the
left quantile,
\(\hat F_n(\hat c_{1-\alpha}-\varepsilon_n)\le 1-\alpha\). Hence
\[
F_n(\hat c_{1-\alpha}-\varepsilon_n)
\le
1-\alpha+a_n.
\]
The anti-concentration bound yields
\[
F_n(\hat c_{1-\alpha})
\le
F_n(\hat c_{1-\alpha}-\varepsilon_n)
+
C_{\rm ac}\varepsilon_n\sqrt{L_n}
\le
1-\alpha+a_n+o(1).
\]
Combining the lower and upper bounds proves the claim.
\end{proof}

\begin{proof}[Proof of Corollary~\ref{cor:scb} and Theorem~\ref{thm:global_test}]
By \hyperref[thm:clt]{Theorem~\ref*{thm:clt}}, the oracle-studentized maximum of the
estimation error is approximated by \(V_n\). Equation~\eqref{eq:correct_scale_rate} gives
\[
\max_{g\in\mathcal G_n^{\mathrm{inf}}}
\sqrt n|\hat\Delta(g)-\Delta(g)|
\left|
\frac{1}{\hat\sigma(g)}-\frac{1}{\sigma(g)}
\right|
=o_P(L_n^{-1/2}),
\]
because the oracle-studentized maximum is \(O_P(\sqrt{L_n})\).
\hyperref[lem:max_perturbation_transfer]{Lemma~\ref*{lem:max_perturbation_transfer}} therefore gives the same
Gaussian approximation for the feasible studentized maximum
\[
Q_n
:=
\max_{g\in\mathcal G_n^{\mathrm{inf}}}
\frac{\sqrt n|\hat\Delta(g)-\Delta(g)|}{\hat\sigma(g)}.
\]
\hyperref[lem:multiplier_max_bootstrap]{Lemma~\ref*{lem:multiplier_max_bootstrap}} gives consistency of the
conditional multiplier distribution
\[
\hat F_n(t)=
\Pr_\xi\!\left(
\max_{g\in\mathcal G_n^{\mathrm{inf}}}|\hat Z_n^\star(g)|\le t
\mid O_1^n,\mathcal F_n^{\mathrm{aux}}
\right),
\]
and \hyperref[lem:bootstrap_quantile_transfer]{Lemma~\ref*{lem:bootstrap_quantile_transfer}} calibrates its quantile
against \(F_n(t)=\Pr(V_n\le t\mid\mathcal F_n^{\mathrm{aux}})\).

Because \(Q_n\) and the cutoff \(\hat c_{1-\alpha}^{\mathrm{SCB}}\) are
computed from the same sample, we conclude with a bracketing argument. Choose a deterministic
\(\eta_n\downarrow0\) such that the event
\(\mathcal E_n=\{\|\hat F_n-F_n\|_\infty\le\eta_n\}\) satisfies
\(\Pr(\mathcal E_n^c\mid\mathcal F_n^{\mathrm{aux}})=o_P(1)\), and set
\[
c_\pm(\eta_n)=F_n^{-1}(1-\alpha\pm2\eta_n).
\]
Because \(F_n\) is continuous, generalized-inverse monotonicity gives
\(c_-(\eta_n)\le\hat c_{1-\alpha}^{\mathrm{SCB}}\le c_+(\eta_n)\) on
\(\mathcal E_n\). Hence
\[
\begin{aligned}
\Pr\{Q_n\le c_-(\eta_n)\mid\mathcal F_n^{\mathrm{aux}}\}
-\Pr(\mathcal E_n^c\mid\mathcal F_n^{\mathrm{aux}})
&\le
\Pr\{Q_n\le\hat c_{1-\alpha}^{\mathrm{SCB}}
\mid\mathcal F_n^{\mathrm{aux}}\}\\
&\le
\Pr\{Q_n\le c_+(\eta_n)\mid\mathcal F_n^{\mathrm{aux}}\}
+\Pr(\mathcal E_n^c\mid\mathcal F_n^{\mathrm{aux}}).
\end{aligned}
\]
The Gaussian approximation for \(Q_n\), the continuity of \(F_n\), and
\(\eta_n\to0\) imply that both outer bounds converge in probability to
\(1-\alpha\), so
\[
\Pr\left\{
Q_n
\le \hat c_{1-\alpha}^{\mathrm{SCB}}
\ \middle|\ \mathcal F_n^{\mathrm{aux}}
\right\}
\to_P 1-\alpha.
\]
This event is exactly the simultaneous coverage event for
\(\mathcal C_{n,1-\alpha}(g)\) over the inferential grid. Averaging over
the auxiliary information gives the unconditional statement.

Under \(H_0^{\max}\), \(\Delta(g)=0\) for all cells of the inferential grid, so the
global test statistic \(T_n^{\max}\) is the same feasible studentized
maximum with \(\Delta(g)\) replaced by zero; hence the coverage statement gives
\(\Pr(\varphi_n^{\max}=1\mid\mathcal F_n^{\mathrm{aux}})=1-\Pr(Q_n\le\hat c_{1-\alpha}^{\mathrm{SCB}}\mid\mathcal F_n^{\mathrm{aux}})\to_P\alpha\). \hyperref[rem:fwer]{Remark~\ref*{rem:fwer}} follows from the coverage statement, because every false coordinate rejection is a coverage failure.
\end{proof}

\section{Inference over All Radii at Fixed Anchors}
\label{app:complete_radius}

This appendix proves \hyperref[thm:all_radii]{Theorem~\ref*{thm:all_radii}} and the supporting results summarized in \hyperref[subsec:all_radii]{Section~\ref*{subsec:all_radii}}. The law \(P\), the number of anchors \(K\), and the number of folds \(J\) are fixed. An independent auxiliary sample may choose the anchors once; all statements hold conditionally on that choice, provided the stated conditions hold for almost every realization. If the design is redrawn from a fixed law for each analysis, the unconditional coverage statements follow by integrating the conditional ones. Uniformity over design laws that change with \(n\) is not considered.

\subsection{Scores and Training Regimes}

Write \(D_k=d(\omega_k,Y)\), \(Z_{kt}=\mathds{1}\{D_k\le t\}\), \(F_{ak}(t)=\Pr\{d(\omega_k,Y(a))\le t\}\), and \(\Delta_k=F_{1k}-F_{0k}\); distances are finite almost surely. For an arm-specific conditional CDF \(m_a(k,t;x)\) and a propensity score \(\pi\), with \(\pi_1=\pi\) and \(\pi_0=1-\pi\), define
\[
\phi_a(O;k,t;m,\pi)
=
m_a(k,t;X)
+
\frac{\mathds{1}(A=a)}{\pi_a(X)}\{Z_{kt}-m_a(k,t;X)\},
\]
and \(\phi_\Delta=\phi_1-\phi_0\). Arm profiles and contrasts are indexed jointly by \(h\in\mathcal H=\{0,1,\Delta\}\times\{1,\ldots,K\}\times[0,\infty)\), and \(\theta_h\) denotes the corresponding \(F_{ak}(t)\) or \(\Delta_k(t)\). All statements also hold for the restriction to contrasts. Folds are indexed as in \hyperref[app:regularity_proofs]{Appendix~\ref*{app:regularity_proofs}}: \(\mathcal I_j\) is validation fold \(j\), \(n_j=|\mathcal I_j|\), \(P_{n,j}\) is its empirical measure, and a subscript \(j\) on a fitted nuisance, as in \(\hat m_{a,j}\), plays the role of the superscript \((-j)\) there.

\begin{assumption}[Sampling and learning over all radii]
\label{ass:complete_learning}
The observations are i.i.d., \hyperref[ass:id]{Assumption~\ref*{ass:id}} holds, and \((X,A,D_1,\ldots,D_K)\) takes values in a standard Borel space. The fold partition is independent of the data, \(J\ge2\) is fixed, and \(n_j/n\to\lambda_j>0\). For each validation fold \(j\), the entire fitting procedure, including tuning and algorithmic randomness, uses only the training complement and independent auxiliary information. All fitted propensities lie in \([\epsilon_*,1-\epsilon_*]\) for a fixed \(\epsilon_*>0\). Each fitted regression \(\hat m_{a,j}(k,t;x)\) is jointly measurable and, as a function of \(t\), a proper right-continuous CDF. There is a nonrandom proper conditional CDF \(\bar m_a(k,t;x)\), common to all folds, such that
\[
r_n
=
\max_{j,a,k}\sup_{t\ge0}
\|\hat m_{a,j}(k,t;\cdot)-\bar m_a(k,t;\cdot)\|_{P_X,2}
\to_P0.
\]
One of the following regimes holds.
\begin{enumerate}[label=(\roman*),leftmargin=*]
\item \emph{Known propensity.} The score uses the true \(\pi\); the limit \(\bar m\) may differ from the true conditional CDF \(m\).
\item \emph{Correct parametric propensity.} The folds are asymptotically balanced, \(\lambda_j=1/J\). A correctly specified model \(\pi_\beta\) of fixed dimension has an interior true parameter \(\beta_0\). In a fixed neighborhood of \(\beta_0\), overlap holds uniformly, \(\pi_\beta\) has bounded continuous first and second derivatives, and \(q_\beta=\partial_\beta\log\{\pi_\beta/(1-\pi_\beta)\}\) and its first derivative are bounded. With \(S_\beta=\{A-\pi_\beta(X)\}q_\beta(X)\), the information matrix \(I=P S_{\beta_0}S_{\beta_0}^{\mathsf T}\) is positive definite. Each off-fold estimate \(\hat\beta_j\) is consistent and satisfies \(P_{n,-j}S_{\hat\beta_j}=o_P(n^{-1/2})\), where \(P_{n,-j}\) is the empirical measure of the training complement of fold \(j\), and clipping is asymptotically inactive. The limit \(\bar m\) may again differ from \(m\).
\item \emph{Two consistent nuisances.} Here \(\bar m=m\), the propensity estimates may be nonparametric, and
\[
r_{\pi,n}=\max_j\|\hat\pi_j-\pi\|_{P_X,2}=o_P(1),
\qquad
\sqrt n\,r_nr_{\pi,n}=o_P(1).
\]
\end{enumerate}
\end{assumption}

The conditions concern population risks rather than training diagnostics. Monotonicity in \(t\) does not by itself imply consistency for the true conditional CDF; Propositions~\ref{prop:fixed_kernel_stability} and~\ref{prop:finite_strata} give learners that satisfy the stability requirement and regime~(iii), respectively.

Set \(u_h=\phi_h(\bar m,\pi)-\theta_h\). In regime~(ii) put
\[
S=S_{\beta_0},\qquad C_h=P(u_hS^{\mathsf T}),\qquad v_h=u_h-C_hI^{-1}S.
\]
The limiting influence function \(f_h\) is \(u_h\) in regime~(i), \(v_h\) in regime~(ii), and the efficient influence function \(\psi_h\) in regime~(iii).

\subsection{Atoms and the Complexity of the Radius Path}

\begin{lemma}[Entropy of a monotone path]\label{lem:monotone_path}
If \(0\le b_t(o)\le1\) is nondecreasing in a real index \(t\) for every \(o\), then for every probability measure \(Q\),
\[
N_{[\,]}\bigl(\eta,\{b_t:t\ge0\},L^2(Q)\bigr)\le2+2\eta^{-2},
\qquad 0<\eta<1.
\]
Finite sums of such paths with fixed bounded measurable coefficients have polynomial uniform covering numbers, with constants that do not depend on the particular paths.
\end{lemma}
\begin{proof}
Partition the range of \(t\mapsto Qb_t\) into intervals of length \(\eta^2/2\). For each nonempty interval, take the pointwise infimum \(l\) and supremum \(u\) of \(b_t\) over the indices mapped into it. Ordered index sets admit countable cofinal and coinitial sequences, so these envelopes are measurable, and monotone convergence gives \(Q(u-l)\le\eta^2/2\); since \(0\le u-l\le1\), also \(Q(u-l)^2\le\eta^2/2\). A jump of the path may skip intervals, but it never has to fit inside a single bracket, so atoms in the distance distributions do not affect the argument. For a coefficient \(c\) of either sign, write \(c^\pm=\max(\pm c,0)\) and use the bracket \([c^+l-c^-u,\,c^+u-c^-l]\). Products of finitely many covers handle finite sums, and finite unions over anchors and score types, as well as bounded products and squares, preserve polynomial entropy.
\end{proof}

The lemma is applied to each frozen fit separately; it does not assert that the union of all possible fits is a Donsker class.

\begin{lemma}[Off-fold perturbations]\label{lem:radius_perturbation}
Let \(e_{j,h}\) be the difference between the fitted and the limiting uncentered score on fold \(j\). Under \hyperref[ass:complete_learning]{Assumption~\ref*{ass:complete_learning}},
\[
\max_j\sup_h|\mathbb G_{n,j}e_{j,h}|=o_P(1),
\qquad
\max_j\sup_hP_{n,j}e_{j,h}^2=o_P(1),
\]
where \(\mathbb G_{n,j}=\sqrt{n_j}(P_{n,j}-P)\) and \(P\) integrates a new observation conditionally on the training set of fold \(j\).
\end{lemma}
\begin{proof}
Condition on the training set of one fold. The score difference has a fixed bounded envelope, and its largest \(L^2(P)\) norm is at most \(C(r_n+r_{\pi,n})\), with \(r_{\pi,n}=0\) in regime~(i) and \(r_{\pi,n}\) replaced by \(\|\hat\beta_j-\beta_0\|\) in regime~(ii); this follows by expanding the inverse propensities and using their lower bound. \hyperref[lem:monotone_path]{Lemma~\ref*{lem:monotone_path}} applies to the true and fitted conditional CDFs and to \(Z_{kt}\), with constants that do not depend on the training realization. Choose a deterministic \(\delta_n\downarrow0\) with \(\delta_n\ge n^{-1/4}\) such that the largest conditional \(L^2\) norm is at most \(\delta_n\) with probability tending to one. The local maximal inequality under uniform entropy \citep[Theorem~2.1]{vdvw2011local} bounds the conditional expected supremum by
\[
C\left\{\delta_n\sqrt{\log(C/\delta_n)}+\frac{\log(C/\delta_n)}{\sqrt{n_j}}\right\}=o(1).
\]
The conditional Markov inequality gives the first claim. The same inequality applied to the class of squares gives \(\sup_hP_{n,j}e_{j,h}^2\le\sup_hPe_{j,h}^2+O_P(n_j^{-1/2})=o_P(1)\). A union over the fixed number of folds completes the proof; no independence between the training sets of different folds is used.
\end{proof}

\subsection{Linearization and Process Limit}

\begin{theorem}[Uniform linearization and process limit]\label{thm:radius_process}
Under \hyperref[ass:complete_learning]{Assumption~\ref*{ass:complete_learning}},
\[
\Bigl\|\sqrt n(\hat\theta-\theta)-n^{-1/2}\sum_{i=1}^nf(O_i)\Bigr\|_{\mathcal H}=o_P(1),
\qquad
n^{-1/2}\sum_{i=1}^nf(O_i)\rightsquigarrow G_f
\quad\text{in }\ell^\infty(\mathcal H),
\]
where \(G_f\) is a tight centered Gaussian process with covariance \(P(f_hf_{h'})\). In regime~(ii),
\[
P(v_hv_{h'})=P(u_hu_{h'})-C_hI^{-1}C_{h'}^{\mathsf T}.
\]
No coordinate variance is required to be positive.
\end{theorem}

\begin{proof}
\textit{Regime (i).} With the true propensity, \(P\phi_a(\hat m_j,\pi)=F_{ak}(t)\) for every off-fold fit, so after subtracting the limiting score only the centered perturbations of \hyperref[lem:radius_perturbation]{Lemma~\ref*{lem:radius_perturbation}} remain. This proves the expansion in regime~(i).

\textit{Regime (iii).} The arm-specific bias equals
\[
P\phi_a(\hat m_j,\hat\pi_j)-F_{ak}(t)
=
P\left[\frac{(m_a-\hat m_{a,j})(\pi_a-\hat\pi_{a,j})}{\hat\pi_{a,j}}\right],
\]
which is bounded by \(\epsilon_*^{-1}r_nr_{\pi,n}\) uniformly in \(a,k,t\); the product-rate condition completes this case.

\textit{Regime (ii).} Write \(\delta_{ak}(t;x)=m_a(k,t;x)-\bar m_a(k,t;x)\). Conditioning on \((A,X)\) gives
\[
C_{1kt}=P\bigl[(1-\pi)\delta_{1k}(t;X)q^{\mathsf T}\bigr],
\qquad
C_{0kt}=-P\bigl[\pi\,\delta_{0k}(t;X)q^{\mathsf T}\bigr],
\]
and differentiating the mean score gives \(\partial_\beta P\phi_h(\bar m,\pi_\beta)|_{\beta_0}=-C_h\). The bounded second derivatives and bounded CDF values give a Taylor remainder \(O_P(\|\hat\beta_j-\beta_0\|^2)\), uniformly in \(h\), and replacing \(\bar m\) by \(\hat m_j\) changes the drift by at most \(Cr_n\|\hat\beta_j-\beta_0\|\). The estimating equation, consistency, and a local uniform law of large numbers for the score derivative yield \(\hat\beta_j-\beta_0=I^{-1}P_{n,-j}S+o_P(n^{-1/2})\). Hence the root-\(n\) expansion is
\[
\mathbb G_nu_h-\sqrt n\,C_hI^{-1}\sum_j\frac{n_j}{n}P_{n,-j}S+o_P(1)
\]
uniformly in \(h\). In the last average, an observation in fold \(\ell\) receives weight \(a_{n,\ell}=\sum_{j\ne\ell}\lambda_{n,j}/(1-\lambda_{n,j})\) relative to its weight \(1/n\) in \(P_nS\), where \(\lambda_{n,j}=n_j/n\). For equal folds this weight equals one; for asymptotically balanced folds it converges to one, and the squared \(L^2\) norm of the resulting replacement error is \(\operatorname{tr}(I)\sum_\ell\lambda_{n,\ell}(a_{n,\ell}-1)^2=o(1)\). This yields the projected score \(v_h\).

\textit{Process limit.} \hyperref[lem:monotone_path]{Lemma~\ref*{lem:monotone_path}} gives a finite uniform entropy integral and total boundedness in \(L^2(P)\). The local maximal inequality applied to the increment class with \(L^2\) norm at most \(\delta\) bounds the asymptotic modulus by \(C\delta\sqrt{\log(C/\delta)}\), which vanishes as \(\delta\to0\); together with the finite-dimensional central limit theorem this gives a tight limit. The finite-dimensional score \(S\) can be added to the class, and the projection coefficients \(C_hI^{-1}\) range over a bounded subset of a fixed Euclidean space, whose covering numbers are polynomial. The map \((z,s)\mapsto(z_h-C_hI^{-1}s)_h\) is continuous because \(\sup_h\|C_h\|<\infty\), which proves the projected limit and the covariance identity.
\end{proof}

\paragraph{Other training protocols.}
If the limiting fold proportions are unequal, the influence function of an observation in fold \(\ell\) is \(u_h-a_\ell C_hI^{-1}S\) with \(a_\ell=\sum_{j\ne\ell}\lambda_j/(1-\lambda_j)\). Since \(\sum_\ell\lambda_\ell a_\ell=1\), the pooled covariance is
\[
P(u_hu_{h'})+\Bigl(\sum_\ell\lambda_\ell a_\ell^2-2\Bigr)C_hI^{-1}C_{h'}^{\mathsf T}.
\]
With two folds of proportions \(w\) and \(1-w\) the coefficient equals \(1/\{w(1-w)\}-5\): it is \(-1\) for \(w=1/2\) but \(1.25\) for \(w=0.2\). If the propensity model is instead fitted on an independent sample of size \(M\) with \(n/M\to\rho\), the additional covariance is \(+\rho\,C_hI^{-1}C_{h'}^{\mathsf T}\). The projected covariance of \hyperref[thm:radius_process]{Theorem~\ref*{thm:radius_process}} therefore applies to balanced cross-fitting only.

\subsection{Multiplier Calibration and Protected Bands}

Let \(\hat U_i(h)=\hat\phi_i(h)-\hat\theta_h\). In regimes~(i) and~(iii), the calibration rows are \(\hat f_i(h)=\hat U_i(h)\). In regime~(ii), evaluate the off-fold propensity score \(\hat S_i=\{A_i-\pi_{\hat\beta_{j(i)}}(X_i)\}q_{\hat\beta_{j(i)}}(X_i)\), where \(j(i)\) is the fold of observation \(i\), and set
\[
\tilde S_i=\hat S_i-P_n\hat S,\quad
\hat I=P_n\tilde S\tilde S^{\mathsf T},\quad
\hat C_h=P_n\hat U(h)\tilde S^{\mathsf T},\quad
\hat f_i(h)=\hat U_i(h)-\hat C_h\hat I^{-1}\tilde S_i.
\]
For a deterministic \(a_n\downarrow0\), the inverse is used only when \(\lambda_{\min}(\hat I)\ge a_n\) and all quantities are finite; otherwise the calibration rows are set to zero and every interval is replaced by its full parameter range. This event has probability tending to zero. The projection changes the calibration rows, not the AIPW point estimate.

\begin{theorem}[Multiplier calibration over all radii]\label{thm:radius_multiplier}
Under the regime of \hyperref[ass:complete_learning]{Assumption~\ref*{ass:complete_learning}} that matches the calibration rows, independent standard normal multipliers \(\xi_i\) satisfy
\[
B_n^*(h)=n^{-1/2}\sum_{i=1}^n\xi_i\hat f_i(h)\rightsquigarrow_\xi G_f
\]
conditionally in probability in \(\ell^\infty(\mathcal H)\), and \(\sup_h|P_n\hat f(h)^2-Pf_h^2|=o_P(1)\).
\end{theorem}

\begin{proof}
The second conclusion of \hyperref[lem:radius_perturbation]{Lemma~\ref*{lem:radius_perturbation}}, the Cauchy--Schwarz inequality, and uniform laws of large numbers for bounded products give \(\hat I\to_PI\), \(\sup_h\|\hat C_h-C_h\|=o_P(1)\), and \(\sup_hP_n\{\hat f(h)-f_h\}^2=o_P(1)\). Covariance consistency alone does not yield process convergence, so the full perturbation is controlled as follows. Condition on the data and let \(\eta_{n,j}=\sup_h\{P_{n,j}e_{j,h}^2\}^{1/2}=o_P(1)\). The entropy bound of \hyperref[lem:monotone_path]{Lemma~\ref*{lem:monotone_path}} is uniform over probability measures and therefore applies to the empirical measure; Gaussian chaining with one base point gives
\[
E_\xi\sup_h\Bigl|n_j^{-1/2}\sum_{i\in\mathcal I_j}\xi_ie_{j,h}(O_i)\Bigr|
\le
\eta_{n,j}+C\int_0^{2\eta_{n,j}}\sqrt{\log(C/u)}\,du
=o_P(1).
\]
Global centering uses \(\hat\theta_h\) in place of \(\theta_h\): \(\hat U_i(h)-u_h(O_i)=e_{j(i),h}(O_i)-(\hat\theta_h-\theta_h)\), so centering contributes \(-(\hat\theta_h-\theta_h)\,n^{-1/2}\sum_i\xi_i\), which is \(O_P(n^{-1/2})\) uniformly in \(h\) by \hyperref[thm:radius_process]{Theorem~\ref*{thm:radius_process}}. For the projected rows, the remaining differences are a uniformly \(o_P(1)\) coefficient multiplying a fixed-dimensional Gaussian vector and a uniformly bounded coefficient multiplying the Gaussian sum of \(\tilde S_i-S_i=(\hat S_i-S_i)-P_n\hat S\); both parts have conditional second moment \(o_P(1)\), because the bounded derivatives imply \(\max_i\|\hat S_i-S_i\|=o_P(1)\) and \(P_nS=O_P(n^{-1/2})\).

For the oracle multiplier process, conditional finite-dimensional convergence follows from the law of large numbers. On the increment class \(\{f_h-f_{h'}:\|f_h-f_{h'}\|_{P,2}<\delta\}\), the largest empirical squared norm is at most \(\delta^2+o_P(1)\) by the uniform law for the product class, and the same entropy integral controls the conditional modulus. Finite nets and this modulus give conditional tightness and hence the convergence of the oracle multiplier process, which the perturbation bound transfers to the estimated rows. The product-class bounds also give the variance statement.
\end{proof}

Fix \(s_0>0\) and \(c_{\min}>0\), let \(\hat s_h=\max\{(P_n\hat f(h)^2)^{1/2},s_0\}\), and let \(\hat c_{1-\alpha}\) be the conditional \((1-\alpha)\)-quantile of \(\sup_h|B_n^*(h)|/\hat s_h\). The reported band is
\begin{equation}\label{eq:protected_radius_band}
\hat\theta_h\pm\max\{\hat c_{1-\alpha},c_{\min}\}\,\hat s_h/\sqrt n,
\qquad h\in\mathcal H,
\end{equation}
intersected with \([0,1]\) for arm profiles and with \([-1,1]\) for contrasts.

\begin{corollary}[Coverage including degenerate laws]\label{cor:radius_protected}
For every fixed law satisfying \hyperref[ass:complete_learning]{Assumption~\ref*{ass:complete_learning}}, the simultaneous coverage of \eqref{eq:protected_radius_band} has limit inferior at least \(1-\alpha\). If \(G_f\) is not identically zero, the coverage of the band without the lower bound \(c_{\min}\) converges to \(1-\alpha\). If \(G_f\equiv0\), the coverage of \eqref{eq:protected_radius_band} converges to one.
\end{corollary}

\begin{proof}
Uniform variance consistency and the positive scale floor give uniform consistency of \(1/\hat s_h\). Hence both the maximum of the studentized estimation error and the conditional multiplier maximum converge in distribution to
\[
V=\sup_{h\in\mathcal H}\frac{|G_f(h)|}{\max\{(Pf_h^2)^{1/2},s_0\}}.
\]
If \(G_f\) is not identically zero, the distribution function of \(V\) is continuous and strictly increasing on \((0,\infty)\) (see \hyperref[rem:norm_regularity]{Remark~\ref*{rem:norm_regularity}} and~\citep{gaenssler2007continuity}), so the conditional quantiles converge and Slutsky's lemma gives coverage \(1-\alpha\); enlarging the critical value cannot reduce it. If \(G_f\equiv0\), the studentized estimation error is \(o_P(1)\), whereas the protected cutoff is at least \(c_{\min}>0\).
\end{proof}

\begin{remark}[Regularity of the norm distribution]\label{rem:norm_regularity}
Let \(G\) be a tight centered Gaussian process that is not identically zero. Tightness yields a compact set of positive probability and finite covers of it by arbitrarily small balls; one such ball has positive probability, two independent copies of \(G\) fall into it with positive probability, and their difference has the law of \(\sqrt2\,G\). Hence every ball about zero has positive probability. On a countable norming set, the distribution function of each finite maximum of absolute coordinates is log-concave in the radius, by log-concavity of Gaussian measures on symmetric convex sets, and decreasing limits preserve log-concavity. The logarithm of the limiting distribution function is therefore finite and concave on \((0,\infty)\), hence continuous there. A nonzero coordinate excludes an atom at zero and keeps the distribution function strictly below one at every finite radius; if the log-distribution function were constant on an interval, concavity would keep it constant thereafter, contradicting its convergence to zero at infinity. This gives strict increase.
\end{remark}

\paragraph{A scale floor alone is insufficient.}
Consider regime~(i) with contrast-only calibration on a fixed window. Fix an anchor \(\omega\), let \(\pi=1/2\), \(d(\omega,Y(1))=0\), and \(d(\omega,Y(0))=2\), and take the window \([1/2,3/2]\), on which \(\Delta=1\) and the efficient score vanishes. Use two equal folds of size \(n/2\). Let \(r_n\downarrow0\), and let \(s_j\in\{-1,1\}\) be the sign of the treatment imbalance in the training complement of fold \(j\), with ties assigned \(+1\). On the window set \(\hat m_{1j}=1-r_n(1+s_j)/2\) and \(\hat m_{0j}=r_n(1-s_j)/2\), placing this mass at distance \(0\) and the remaining mass at distance \(2\), so that the fits are proper CDFs on the whole half-axis; they converge uniformly to the true regressions. The fitted contrast score on the window equals \(1+2r_ns_{j(i)}(A_i-1/2)\). If \(N_1\) and \(N_2\) are independent standard normal variables, then
\[
\sqrt n\,|\hat\Delta-1|/r_n\ \dto\ (|N_1|+|N_2|)/\sqrt2,
\qquad
\hat\sigma/r_n\to_P1,
\]
so the scale floor cancels between the bootstrap cutoff and the band width. Since \((|N_1|+|N_2|)/\sqrt2\) is the larger of the absolute values of two independent standard normal variables after an orthogonal rotation, the coverage without \(c_{\min}\) converges to \((1-\alpha)^2\). With any fixed \(c_{\min}>0\) it converges to one.

\subsection{Learners That Satisfy the Conditions}

\begin{proposition}[Fixed-kernel conditional empirical CDF]
\label{prop:fixed_kernel_stability}
Let \(0\le H(x,x')\le L<\infty\) be a measurable kernel whose parameters and preprocessing are deterministic or fixed by independent auxiliary information, and suppose that \(D_a(x)=E\{\mathds{1}(A=a)H(x,X)\}>0\) for \(P_X\)-almost every \(x\). From \(M\) training observations form
\[
\hat m_a(k,t;x)
=
\frac{\sum_{r=1}^M\mathds{1}(A_r=a)H(x,X_r)\mathds{1}\{D_{kr}\le t\}}
     {\sum_{r=1}^M\mathds{1}(A_r=a)H(x,X_r)},
\]
using a fixed proper CDF when the denominator vanishes, and let \(\bar m_a(k,t;x)=E\{\mathds{1}(A=a)H(x,X)\mathds{1}(D_k\le t)\}/D_a(x)\). For a universal constant \(C\),
\[
E\int\sup_{t\ge0}|\hat m_a(k,t;x)-\bar m_a(k,t;x)|^2\,dP_X(x)
\le
\int\min\left\{1,\frac{CL^2}{M D_a(x)^2}\right\}dP_X(x)
\longrightarrow0.
\]
Hence this learner satisfies the stability condition of regimes~(i) and~(ii), without requiring \(\bar m=m\), a density for \(X\), or moments of the distances.
\end{proposition}

\begin{proof}
Fix \(x\), and write \(\hat N_t\) and \(\hat D\) for the training-sample averages in the numerator and denominator, with population counterparts \(N_t\) and \(D=D_a(x)\). Symmetrization and ordering by \(D_{kr}\) turn the numerator fluctuation into a weighted Rademacher partial-sum process, and Doob's \(L^2\) inequality gives \(E\sup_t|\hat N_t-N_t|^2\le16L^2/M\); tied distances only reduce the set of partial sums. Also \(E|\hat D-D|^2\le L^2/M\). On the event \(|\hat D-D|\le D/2\), since \(N_t\le D\),
\[
\sup_t\left|\frac{\hat N_t}{\hat D}-\frac{N_t}{D}\right|
\le
\frac2D\Bigl\{\sup_t|\hat N_t-N_t|+|\hat D-D|\Bigr\},
\]
and Chebyshev's inequality bounds the probability of the complement by \(4L^2/(MD^2)\). The two events contribute at most \(8\cdot17\,L^2/(MD^2)\) and, because two CDFs differ by at most one, \(4L^2/(MD^2)\), so \(C=140\) suffices. Integrating in \(x\) and applying dominated convergence completes the proof; nonnegative weights and finite distances make the population ratio a proper CDF.
\end{proof}

A fixed Gaussian kernel is positive everywhere and gives a simple instance. The limit is the kernel-smoothed conditional CDF rather than the true conditional CDF, which is all that regimes~(i) and~(ii) require. Data-tuned kernels with several limiting bandwidths, forest-based weights, or fold-dependent standardization need a separate argument for the stability condition.

\begin{proposition}[Finite covariate space]\label{prop:finite_strata}
Suppose \(X\) has a fixed finite support and \(\min_{x,a}\Pr(X=x,A=a)>0\). Within each training fold, estimate the propensity score by the empirical treatment proportion in each covariate cell and the conditional CDF by the arm-specific empirical distance CDF in that cell. Clip the propensity at a fixed level outside the true overlap interval and use proper CDF fallbacks for empty cells. Then both uniform nuisance errors are \(O_P(n^{-1/2})\), so regime~(iii) holds over the whole radius half-axis.
\end{proposition}

\begin{proof}
All cell-by-arm counts are bounded below by a constant multiple of the training size with probability \(1-O(e^{-cn})\). Conditionally on a count, the distances in a cell are i.i.d., and the Dvoretzky--Kiefer--Wolfowitz inequality gives an \(O_P(n^{-1/2})\) error uniformly in the radius, including atoms. A finite union over cells, arms, anchors, and folds preserves this rate. The empirical treatment proportions have the same rate by binomial concentration, and their denominators are bounded below on the same event. Clipping is asymptotically inactive because its thresholds lie outside the true range, and empty-cell fallbacks occur with vanishing probability. The product of the two errors multiplied by \(\sqrt n\) is \(O_P(n^{-1/2})\).
\end{proof}

\subsection{Measurability and Computation}

Regular versions of the conditional distributions are used throughout, so the conditional CDFs are right-continuous in \(t\) and every supremum over \(t\ge0\) is determined by the nonnegative rationals. Weak convergence in \(\ell^\infty(\mathcal H)\) is understood in the outer-probability sense when the process is not Borel measurable; the scalar suprema and the multiplier quantiles have the measurable versions just described. Proper CDF limits imply \(\theta_{akt}\to1\), \(\theta_{\Delta kt}\to0\), and \(f_h\to0\) in \(L^2(P)\) as the radius of \(h\) tends to infinity. The entropy bound does not depend on the length of the radius interval, so the result covers the whole half-axis directly. At radius zero the closed-ball mass is retained.

If all fits, including fallbacks, are finite-support weighted empirical CDFs with weights that do not depend on \(t\), then all fitted scores, projections, means, variances, and multiplier paths are constant between consecutive points of the union of the evaluation, training, auxiliary, and fallback distances. Evaluating the closed-ball values at these points, together with the initial and final states, computes the supremum over radii exactly; tied distances enter together. \hyperref[alg:complete_radius]{Algorithm~\ref*{alg:complete_radius}} in \hyperref[app:implementation]{Appendix~\ref*{app:implementation}} summarizes the procedure. Smooth CDF fits require a different algorithm for the supremum. In simulation studies the coverage error is a supremum against the true curve, whose extrema need not occur at the observed distances.

\paragraph{Finite multiplier draws.}\phantomsection\label{app:mc_quantiles}
The theorems use exact conditional Gaussian quantiles. With \(B\to\infty\) independent draws, Monte Carlo quantiles are consistent when the limiting process is nonzero. For a fixed number of draws, a conservative choice allocates \(\delta_{\rm MC}<\alpha\), sets \(p_0=1-\alpha+\delta_{\rm MC}\) and \(\eta_B=\{\log(2/\delta_{\rm MC})/(2B)\}^{1/2}\), and uses the \(\lceil B(p_0+\eta_B)\rceil\)-th order statistic of the simulated maxima, without interpolation. By the Dvoretzky--Kiefer--Wolfowitz inequality with Massart's constant \citep{massart1990tight}, this order statistic lies above the exact \(p_0\)-quantile with conditional probability at least \(1-\delta_{\rm MC}\). If \(p_0+\eta_B>1\), the number of draws must be increased. The same rule can be used for the finite-grid cutoff of \hyperref[subsec:finite_grid]{Section~\ref*{subsec:finite_grid}}.

\section{Implementation Details and Design Rules}
\label{app:implementation}

This appendix gives the implementation details suppressed from the main text. The finite-grid theory concerns \(\mathcal G_n^{\mathrm{inf}}\), while the software also computes point-estimation and diagnostic quantities on the larger grid \(\mathcal G_n^{\mathrm{est}}\). \hyperref[alg:complete_radius]{Algorithm~\ref*{alg:complete_radius}} implements inference over all radii at fixed anchors (\hyperref[app:complete_radius]{Appendix~\ref*{app:complete_radius}}).

\subsection{Scalable Nuisance Estimation via Conditional Distribution Learning}\label{sec:nuisance_est}

A naive implementation of \hyperref[thm:ident]{Theorem~\ref*{thm:ident}} would fit a separate binary classifier for every anchor--radius pair $g=(\omega,t)$. Such a construction is unattractive both computationally and statistically. It is expensive, and more importantly, it need not preserve the monotonicity in $t$ required of a conditional CDF:
\[
m_a(\omega,t_1;x)\le m_a(\omega,t_2;x),
\qquad t_1<t_2.
\]

For a fixed anchor $\omega_k\in\mathcal W$, the transformed outcome
\[
D_k:=d(\omega_k,Y)
\]
is scalar-valued, and the nuisance target becomes the conditional CDF
\[
m_a(\omega_k,t;x)
=
\Prob(D_k\le t\mid A=a,X=x).
\]
This suggests estimating $m_a(\omega_k,\cdot\,;x)$ directly as a conditional distribution function rather than separately learning threshold-wise binary regressions.
Fitting the conditional CDF directly keeps the fitted profile monotone in \(t\).

The theory is compatible with conditional distribution learners that return coherent CDFs in $t$, provided the finite-grid nuisance-rate and correlation conditions of \hyperref[app:regularity_proofs]{Appendix~\ref*{app:regularity_proofs}} hold. We describe two structurally compatible classes; the experiments use the DRF weighted ECDF learner, except for the five-space table, which uses the anchor-specific kernel ECDF learner.

\paragraph{(i) Generic conditional distribution learners.}
One option is to employ a distributional learner such as Distributional Random Forests (DRF), which estimates conditional distributions through forest-induced weights and can therefore produce a valid CDF curve in $t$ by construction~\citep{cevid2022drf}.

\paragraph{(ii) Anchor-specific kernel-weighted ECDF learners.}
A second option is to estimate, for each anchor $\omega_k$, a kernel-weighted conditional empirical CDF. This directly targets the same nuisance object \(m_a(\omega_k,t;x)\) and likewise guarantees monotonicity in \(t\) by construction, paralleling the empirical-CDF viewpoint in classical empirical-process theory~\citep{van1996weak}.

These two implementations differ algorithmically, but they target the same conditional CDF nuisance and both produce structurally valid profile regressions. Their inferential use is governed by the finite-grid rate and correlation conditions of \hyperref[app:regularity_proofs]{Appendix~\ref*{app:regularity_proofs}}. We use the DRF weighted ECDF learner outside the five-space table because its conditional-CDF prediction errors decrease more clearly with the sample size (\hyperref[app:learner_robustness]{Appendix~\ref*{app:learner_robustness}}); the kernel ECDF learner is simpler and faster.

The complete cross-fitted estimation procedure is summarized in \hyperref[alg:dr_optimized]{Algorithm~\ref*{alg:dr_optimized}}.

\begin{algorithm}[tbp]
\caption{Finite-Grid Profile Causal Effect Inference}
\label{alg:dr_optimized}
\begin{algorithmic}[1]
\STATE \textbf{Input:} Main analysis sample $\mathcal{D}_n=\{(A_i,X_i,Y_i)\}_{i=1}^n$ and an independent auxiliary design/pilot sample $\mathcal{D}^{\mathrm{pilot}}$ with index set $\mathcal{I}_{\mathrm{pilot}}$.
\STATE \textbf{Hyperparameters:} Estimation grid $\mathcal{G}_n^{\mathrm{est}}=\mathcal{W}_{\mathrm{grid}}\times\mathcal{T}_{\mathrm{grid}}$, candidate subgrid $\mathcal{G}_n^{\mathrm{cand}}\subseteq\mathcal{G}_n^{\mathrm{est}}$, folds $J$, truncation $\epsilon_{\rm clip}$, trimming threshold $c_p\in(0,0.5)$, level $\alpha$, and number of multiplier draws $B$.

\vspace{1mm}
\STATE \textbf{Phase 1: Grid Construction (Pilot Sample)}
\STATE Use $\mathcal{D}^{\mathrm{pilot}}$ only to construct anchors, radii, and the inferential grid; it is not used in doubly robust score evaluation.
\FOR{each $g = (\omega_k, t) \in \mathcal{G}_n^{\mathrm{cand}}$}
    \STATE Compute marginal probability: $\hat{p}_{\mathrm{pilot}}(g) = \frac{1}{|\mathcal{I}_{\mathrm{pilot}}|} \sum_{i \in \mathcal{I}_{\mathrm{pilot}}} \mathds{1}\{d(\omega_k, Y_i) \le t\}$.
\ENDFOR
\STATE Form inferential grid: $\mathcal{G}_n^{\mathrm{inf}} \leftarrow \{g \in \mathcal{G}_n^{\mathrm{cand}} : \hat{p}_{\mathrm{pilot}}(g) \in [c_p, 1-c_p]\}$.

\vspace{1mm}
\STATE \textbf{Phase 2: Cross-Fitted Inference (Main Sample $\mathcal{D}_n$)}
\STATE Throughout Phase 2, if a required fit fails without a valid fallback, or a required computation returns a nonfinite value, report \([-1,1]\) for every retained coordinate, do not reject, and stop before using that result.
\STATE Cache distance matrix $D_{i,k} = d(\omega_k, Y_i)$ for $i \in \mathcal{D}_n$ and $\omega_k \in \mathcal{W}_{\mathrm{grid}}$.
\STATE Partition $\mathcal{D}_n$ into $J$ disjoint folds $\mathcal{I}_1, \dots, \mathcal{I}_J$.
\FOR{$j=1$ to $J$}
    \STATE Train $\hat{\pi}^{(-j)}(x)$ to predict $\Prob(A=1 \mid X=x)$ using $\mathcal{D}_n \setminus \mathcal{I}_j$.
    \FOR{treatment arm $a \in \{0, 1\}$ and anchor $\omega_k \in \mathcal{W}_{\mathrm{grid}}$}
        \STATE Train conditional CDF learner $\hat{F}_{a,k}^{(-j)}(\cdot \mid x)$ using $\{X_i, D_{i,k} : i \notin \mathcal{I}_j, A_i=a\}$.
    \ENDFOR
    \FOR{each validation unit $i \in \mathcal{I}_j$}
        \STATE Clip propensity: $\hat{\pi}_i = \min\{1-\epsilon_{\rm clip}, \max\{\epsilon_{\rm clip}, \hat{\pi}^{(-j)}(X_i)\}\}$.
        \STATE Set $\hat\pi_{1,i}=\hat\pi_i$ and $\hat\pi_{0,i}=1-\hat\pi_i$ for the arm-specific score denominators.
        \FOR{each $g=(\omega_k, t) \in \mathcal{G}_n^{\mathrm{est}}$}
            \FOR{$a \in \{0,1\}$}
                \STATE Predict nuisance: $\hat{m}_{a,i}(g) = \hat{F}_{a,k}^{(-j)}(t \mid X_i)$.
                \STATE Compute arm-specific score $\hat{\phi}_{a,i}(g)$ using the arm-specific AIPW formula.
            \ENDFOR
            \STATE Set causal score: $\hat{\phi}_i(g) = \hat{\phi}_{1,i}(g) - \hat{\phi}_{0,i}(g)$.
        \ENDFOR
    \ENDFOR
\ENDFOR
\STATE Set \(\hat\Delta(g)=n^{-1}\sum_{i=1}^n\hat\phi_i(g)\) for \(g\in\mathcal G_n^{\mathrm{est}}\), center the pooled validation scores at \(\hat\Delta(g)\), compute \(\hat\sigma^2(g)=n^{-1}\sum_{i=1}^n\{\hat\phi_i(g)-\hat\Delta(g)\}^2\), and use the same centered scores for multiplier-bootstrap calibration.
\STATE If \(\mathcal{G}_n^{\mathrm{inf}}\) is empty, or if \(\hat\sigma(g)=0\) for some \(g\in\mathcal{G}_n^{\mathrm{inf}}\), report \([-1,1]\) for every retained coordinate, do not reject, and stop.
\STATE Otherwise, for \(b=1,\ldots,B\), draw i.i.d.\ \(\xi_i^{(b)}\sim N(0,1)\) independently of the data and compute \(M^{(b)}=\max_{g\in\mathcal G_n^{\mathrm{inf}}}\bigl|n^{-1/2}\sum_{i=1}^n\xi_i^{(b)}\{\hat\phi_i(g)-\hat\Delta(g)\}/\hat\sigma(g)\bigr|\). Take \(\hat c_{1-\alpha}\) as the empirical \((1-\alpha)\)-quantile of \(M^{(1)},\ldots,M^{(B)}\), which approximates the conditional quantile as \(B\to\infty\), or use the finite-draw rule of \hyperref[app:mc_quantiles]{Appendix~\ref*{app:mc_quantiles}}. Report \(\hat\Delta(g)\pm\hat c_{1-\alpha}\hat\sigma(g)/\sqrt n\) for \(g\in\mathcal G_n^{\mathrm{inf}}\), and reject the global null when \(\max_{g\in\mathcal G_n^{\mathrm{inf}}}\sqrt n|\hat\Delta(g)|/\hat\sigma(g)>\hat c_{1-\alpha}\).
\STATE \textbf{Output:} Point estimates $\hat{\Delta}(g)$ on $\mathcal{G}_n^{\mathrm{est}}$ for visualization and diagnostics; standard errors, simultaneous bands, bootstrap cutoffs, and the formal max-statistic $T_n^{\max}=\max_{g\in\mathcal{G}_n^{\mathrm{inf}}}\sqrt n|\hat{\Delta}(g)|/\hat\sigma(g)$ on $\mathcal{G}_n^{\mathrm{inf}}$.
\end{algorithmic}
\end{algorithm}

For a fixed set of anchors, \hyperref[alg:complete_radius]{Algorithm~\ref*{alg:complete_radius}} gives the corresponding procedure over all radii (\hyperref[app:complete_radius]{Appendix~\ref*{app:complete_radius}}).

\begin{algorithm}[tbp]
\caption{Inference over All Radii at Fixed Anchors}
\label{alg:complete_radius}
\begin{algorithmic}[1]
\STATE Freeze the metric, the anchors \(\omega_1,\ldots,\omega_K\), and the auxiliary information. Choose a training regime from \hyperref[ass:complete_learning]{Assumption~\ref*{ass:complete_learning}} and fixed \(s_0,c_{\min}>0\).
\STATE Split the main sample into folds independently of the data; use balanced folds in the parametric-propensity regime.
\STATE Throughout, if a required fit fails without a valid fallback, or a required computation returns a nonfinite value, report the full parameter ranges and stop before using that result.
\FOR{each fold}
\STATE On the training complement only, fit the conditional CDFs as finite-support weighted empirical CDFs whose weights do not depend on the radius, using the same form for any fallback. Use the known propensity score in regime~(i) and fit it as specified in \hyperref[ass:complete_learning]{Assumption~\ref*{ass:complete_learning}} in regimes~(ii) and~(iii). Evaluate the arm-specific AIPW paths and their contrasts on the validation fold.
\ENDFOR
\STATE Average the off-fold scores to obtain \(\hat\theta_h\), and set \(\hat U_i(h)=\hat\phi_i(h)-\hat\theta_h\).
\STATE In the parametric-propensity regime, compute \(\hat S_i\), \(\tilde S_i\), \(\hat I\), and \(\hat C_h\) as in \hyperref[app:complete_radius]{Appendix~\ref*{app:complete_radius}}; if \(\lambda_{\min}(\hat I)<a_n\), report the full parameter ranges and stop, and otherwise set the calibration rows to \(\hat f_i(h)=\hat U_i(h)-\hat C_h\hat I^{-1}\tilde S_i\). In the other regimes set \(\hat f_i(h)=\hat U_i(h)\).
\STATE Compute \(\hat s_h=\max\{(n^{-1}\sum_{i=1}^n\hat f_i(h)^2)^{1/2},s_0\}\) from the calibration rows.
\STATE For each multiplier draw, use one Gaussian vector shared across arms, anchors, and radii, and compute the supremum exactly over all jump points, ties, fallback supports, and endpoint states.
\STATE Use the conditional quantile, or the finite-draw rule of \hyperref[app:mc_quantiles]{Appendix~\ref*{app:mc_quantiles}}, and set the cutoff to \(\max(\hat c_{1-\alpha},c_{\min})\).
\STATE Report \(\hat\theta_h\pm\max(\hat c_{1-\alpha},c_{\min})\hat s_h/\sqrt n\), intersected with \([0,1]\) for arms and \([-1,1]\) for contrasts.
\end{algorithmic}
\end{algorithm}

Monotonicity matters for the all-radii procedure of \hyperref[app:complete_radius]{Appendix~\ref*{app:complete_radius}}, which requires proper fitted CDFs; the finite-grid theory of \hyperref[app:regularity_proofs]{Appendix~\ref*{app:regularity_proofs}} does not use it. Threshold-wise binary learners fitted separately across radii therefore remain compatible with the finite-grid procedure, subject to its rate and correlation conditions, but need not produce the proper CDFs required for all radii.

\subsection{Diagnostics for Nuisance Estimation}\label{subsec:nuisance_diag}

Our estimator depends on nuisance components $\pi(x)$ and $m_a(g;x)$ that are scalar conditional probabilities, so standard out-of-fold diagnostics are natural.

\paragraph{Overlap.}
Let $\tilde\pi(X)$ denote the fitted propensity score before clipping. We report
\[
\min_i \tilde\pi(X_i),
\qquad
\max_i \tilde\pi(X_i),
\qquad
\frac1n\sum_{i=1}^n
\mathds{1}\{\tilde\pi(X_i)\notin[\epsilon_{\rm clip},1-\epsilon_{\rm clip}]\},
\]
the last being the fraction of units whose propensity is clipped. A substantial clipped fraction indicates weak overlap.

\paragraph{Calibration for the profile regressions.}
Since $\hat m_a(g;X)$ predicts the binary outcome $Z(g)=\mathds{1}\{d(\omega,Y)\le t\}$, calibration can be assessed by Brier scores and calibration curves. A grid-aggregated Brier score is
\begin{equation}\label{eq:int_brier}
\mathrm{IBS}_a
:=
\sum_{g\in\mathcal G_n^{\mathrm{est}}}
\tilde w_g\,
\mathbb P_n^{(a)}
\Big(\hat m_a(g;X)-Z(g)\Big)^2,
\end{equation}
where $\{\tilde w_g\}$ are deterministic aggregation weights over $\mathcal G_n^{\mathrm{est}}$ and $\mathbb P_n^{(a)}$ denotes the held-out empirical average in arm $a$.

\paragraph{Structural validity of the conditional CDF learner.}
For each anchor, we also monitor whether the fitted nuisance path is a valid CDF in $t$. This property is built into both the DRF weighted ECDF learner and the anchor-specific kernel ECDF learner, but it can fail for threshold-wise binary learners. This diagnostic therefore directly supports the algorithmic choice made in \hyperref[sec:nuisance_est]{Appendix~\ref*{sec:nuisance_est}}.

\subsection{Strategy Naming and Design Rules}
\label{subsec:design_rules}

\paragraph{Empirical anchor-selection rules.}
The results of \hyperref[app:geometry_anchor_theory]{Appendix~\ref*{app:geometry_anchor_theory}} suggest the following practical protocol; \hyperref[tab:anchor_strategy_catalog]{Table~\ref*{tab:anchor_strategy_catalog}} names the strategies used in the experiments. Below, \(\mathcal D_n^{\rm des}\) is the anchor-design sample, \(\mathcal W_K\) is a set of \(K\) selected anchors, \(\mathcal G_K\) is the grid formed by these anchors and the shared radii, \(\widehat\Delta_{\rm pilot}(g)\) and \(\widehat\sigma_{\rm pilot}(g)\) are the profile estimate and the score standard deviation computed on the pilot sample, and \(\widehat c_{K,1-\alpha}\) is a critical value for the maximum over \(\mathcal G_K\).

\begin{enumerate}
\item Use an independent design or pilot sample. Never choose anchors
using the same main sample that is used for the final doubly robust
score average.

\item Start with a geometry-coverage design. A medoid-start MaxMin rule
is the default because it reduces the empirical covering radius
\[
    \widehat\rho_K=\max_{y\in\mathcal D_n^{\rm des}}
    \min_{\omega\in \mathcal W_K}d(y,\omega).
\]

\item Use one radius vector shared by all anchors: quantiles of the
pooled auxiliary distances to the selected anchors, typically at levels
between 0.05 and 0.95. This gives interpretable profile levels
and avoids extreme tails where the profile indicators are nearly
constant.

\item Apply marginal pilot trimming:
\[
    \widehat p_{\rm pilot}(\omega,t)\in[c_p,1-c_p].
\]
The rule uses only marginal ball probabilities, not treatment contrasts.

\item If signal-aware anchors are desired, use only the pilot sample to
estimate the standardized signal \(|\widehat\Delta_{\rm pilot}(g)|/\widehat\sigma_{\rm pilot}(g)\), and combine this
ranking with guardrails: minimum anchor separation, shell quotas, and
the same marginal trimming rule. This yields an adaptive benchmark, not
the formal default design.

\item The grid size can be guided by the standardized power proxy
\[
    \widehat{\mathcal M}_K
    =
    \max_{g\in \mathcal G_K}
    \frac{\sqrt n|\widehat\Delta_{\rm pilot}(g)|}
         {\widehat\sigma_{\rm pilot}(g)}
    -
    \widehat c_{K,1-\alpha}.
\]
This proxy is a working diagnostic rather than a necessary condition
for a power gain: a larger grid can raise power even when the strongest
standardized signal is unchanged (see the example after
\hyperref[prop:max-detection-boundary]{Proposition~\ref*{prop:max-detection-boundary}}).

\item Report anchor diagnostics: covering radius, inferential grid size,
trimmed fraction, minimum estimated score standard deviation, bootstrap
critical value, and sensitivity across anchor seeds.
\end{enumerate}

\section{Interpreting Simultaneous Profile Bands}
\label{app:effect_outputs}

\subsection{Signs, Shells, and Largest Effects}

Let \([L_k(t),U_k(t)]\) be a simultaneous band for \(\Delta_k(t)=\Delta(\omega_k,t)\), either on a finite grid or over all radii at fixed anchors. For a set of radii \(T\), define
\[
\underline\kappa_k(T)=\sup_{t\in T}\max\{L_k(t),-U_k(t),0\},
\qquad
\overline\kappa_k(T)=\sup_{t\in T}\max\{|L_k(t)|,|U_k(t)|\}.
\]
On the event that the band covers every coordinate, \(\underline\kappa_k(T)\le\kappa_k(T)\le\overline\kappa_k(T)\) for \(\kappa_k(T)=\sup_{t\in T}|\Delta_k(t)|\), simultaneously for all anchors and all sets \(T\) of covered radii. Thus \(\overline\kappa_k(T)\le\tau\) establishes equivalence at the probability tolerance \(\tau\) on that set, whereas failure to reject zero is not an equivalence statement. For the shell \(\{y:s<d(\omega_k,y)\le t\}\) with \(s<t\), the treatment effect on its probability is \(\Delta_k(t)-\Delta_k(s)\), which lies in \([L_k(t)-U_k(s),\,U_k(t)-L_k(s)]\) on the same event. These are projections of one confidence set rather than separate coverage statements; scalar distributional-effect methods use the same interval-probability operation~\citep{byambadalai2024distributional}.

\subsection{What Finitely Many Profile Coordinates Cannot Recover}

\begin{proposition}[A radius grid can miss a mean-matched effect]
\label{prop:radius_blindness}
For any finite radius grid there is a one-anchor alternative for which every grid effect and the mean-distance effect are zero, but \(\sup_{t\ge0}|\Delta(t)|=1/2\).
\end{proposition}

\begin{proof}
Use nonnegative real outcomes with anchor zero, and choose \(0<h<c\) such that \([c-h,c+h]\) contains no grid radius. Let \(Y(0)\) put mass \(1/2\) at \(c-h\) and at \(c+h\), and let \(Y(1)=c\). Both mean distances equal \(c\). The profile contrast is \(-1/2\) on \([c-h,c)\), \(1/2\) on \([c,c+h)\), and zero elsewhere, so it vanishes at every grid radius.
\end{proof}

The example concerns radius resolution only. A scalar distributional-effect procedure over all thresholds of the same distance can detect this effect.

\begin{proposition}[Complete profiles at finitely many anchors]
\label{prop:anchor_blindness}
There exist two distinct laws on a four-point subset of \(\mathbb R^2\) and two anchors such that the map from a support point to its pair of anchor distances is injective, the two laws have identical complete profiles at both anchors, and a ball with another center has probability difference one.
\end{proposition}

\begin{proof}
Take anchors \((0,0)\) and \((1,0)\), and let \(r_1=1\) and \(r_2=\sqrt2\). For \(i,j\in\{1,2\}\) put
\[
x_{ij}=(r_i^2-r_j^2+1)/2,
\qquad
y_{ij}=\bigl(x_{ij},\sqrt{r_i^2-x_{ij}^2}\bigr).
\]
The four points \(y_{11}=(1/2,\sqrt3/2)\), \(y_{22}=(1/2,\sqrt7/2)\), \(y_{12}=(0,1)\), and \(y_{21}=(1,1)\) are distinct, and the distances of \(y_{ij}\) to the two anchors are \((r_i,r_j)\), so the distance-pair map is injective. Let \(P_0=(\delta_{y_{11}}+\delta_{y_{22}})/2\) and \(P_1=(\delta_{y_{12}}+\delta_{y_{21}})/2\). Under either law, the distance to each anchor equals \(r_1\) or \(r_2\) with probability \(1/2\), so the complete profiles agree at both anchors. The ball with center \((1/2,1)\) and radius \(1/3\) contains \(y_{11}\) and \(y_{22}\) but neither \(y_{12}\) nor \(y_{21}\), so its probability differs by one between the two laws.
\end{proof}

What the two anchors miss is the joint law of the distance vector, not the identity of individual points. The choice of center matters in a similar way: if \(Y(0)=-1\) and \(Y(1)=1\) on the real line, the complete profile contrast at anchor zero vanishes, whereas for any center \(x\in(0,1)\) the contrast at radius one equals one.

\subsection{Ball Features and Mean Distances}
\label{app:ball_kernel}

\begin{proposition}[Ball features as a kernel mean embedding]
\label{prop:ball_kernel}
For a finite nonnegative measure \(\mu\) on anchor--radius pairs, the kernel
\[
k_\mu(y,y')=\int\mathds{1}\{d(\omega,y)\le t\}\,\mathds{1}\{d(\omega,y')\le t\}\,d\mu(\omega,t)
\]
is positive semidefinite, and its squared mean-embedding distance between the laws of \(Y(1)\) and \(Y(0)\) equals \(\int\Delta(\omega,t)^2\,d\mu(\omega,t)\).
\end{proposition}

\begin{proof}
The map \(y\mapsto(\mathds{1}\{d(\omega,y)\le t\})_{(\omega,t)}\) takes values in \(L^2(\mu)\) and is bounded. For real coefficients \(c_i\), the quadratic form \(\sum_{i,j}c_ic_jk_\mu(y_i,y_j)\) equals \(\int[\sum_ic_i\mathds{1}\{d(\omega,y_i)\le t\}]^2\,d\mu\ge0\). The difference of the two mean feature maps is \(\Delta\), whose squared \(L^2(\mu)\) norm is the stated distance.
\end{proof}

Positive semidefiniteness of this kernel does not require strong negative type. For finitely many ball features \(b(y)\in\{0,1\}^p\) and a positive diagonal weight matrix \(W\), so that \(k(y,y')=b(y)^{\mathsf T}Wb(y')\), the kernel witness evaluated at locations \(v_1,\ldots,v_q\) equals \(B_VW\Delta\), where \(\Delta\in\mathbb R^p\) collects the profile contrasts and the \(\ell\)th row of \(B_V\) is \(b(v_\ell)^{\mathsf T}\). The witness vector is therefore a linear transformation of the profile vector, and the two null hypotheses agree when \(B_VW\) is injective on the attainable profile contrasts.

The profile curve at one anchor also determines a truncated mean-distance effect. For every finite \(R>0\),
\[
\int_0^R\Delta(\omega,t)\,dt
=
E\min\{d(\omega,Y(0)),R\}-E\min\{d(\omega,Y(1)),R\},
\]
because \(E\min\{D,R\}=\int_0^R\Pr(D>t)\,dt\) for a nonnegative random variable \(D\); atoms are allowed. Letting \(R\to\infty\) requires finite first moments of the distances. This scalar summary can be estimated as a bounded average treatment effect and may suffice for some scientific questions; \hyperref[prop:radius_blindness]{Proposition~\ref*{prop:radius_blindness}} shows that it can vanish while the radius curve does not.

\section{Additional Geometry and Anchor-Design Theory}
\label{app:geometry_anchor_theory}

This appendix collects the geometry used to interpret anchor profiles and the design theory used to select finite inferential grids. These results support the population-level profile discussion and the auxiliary-selected design rules; the formal simultaneous inference is given in Appendices~\ref{app:regularity_proofs} and~\ref{app:complete_radius}.

\begin{definition}[Strong negative type \citep{lyons2013distance}]\label{def:snt}
A metric space \((\Omega,d)\) has negative type if, for every finite collection \(x_1,\ldots,x_m\in\Omega\) and real coefficients \(a_1,\ldots,a_m\) with \(\sum_{i=1}^m a_i=0\),
\[
\sum_{i=1}^m\sum_{j=1}^m a_i a_j d(x_i,x_j)\le 0.
\]
It has strong negative type if it has negative type and, for any Borel probability measures \(P\) and \(Q\) with finite first distance moments,
\[
2\E d(X,Y)-\E d(X,X')-\E d(Y,Y')=0
\quad\Longrightarrow\quad
P=Q,
\]
where \(X,X'\stackrel{i.i.d.}{\sim}P\), \(Y,Y'\stackrel{i.i.d.}{\sim}Q\), and the two pairs are independent.
\end{definition}

\subsection{Anchor Selection: Theoretical Standards and Empirical Rules}

The profile causal effect map is indexed by anchor--radius pairs
\(g=(\omega,t)\). Hence, the choice of anchors is not a purely
computational detail: it determines which geometric aspects of the
object-valued causal effect are visible, how stable the standardized
scores are, and how much multiplicity is paid by the max statistic.
This section formalizes anchor selection as a design problem.

Let \(\mathcal F_n^{\rm des}\subseteq\mathcal F_n^{\mathrm{aux}}\)
denote the auxiliary design sigma-field used to construct anchors,
radii, and the trimming mask. In sample-based designs, \(\mathcal F_n^{\rm des}\) may be
generated by an auxiliary design sample \(\mathcal D_n^{\rm des}\),
independent of the main analysis sample \(\mathcal D_n^{\rm main}\).
After conditioning on \(\mathcal F_n^{\rm des}\), the selected grid is
fixed and independent of the main score-evaluation sample. An
anchor/radius design rule is an \(\mathcal F_n^{\rm des}\)-measurable
selection of a candidate grid
\[
    \mathcal G_n^{\rm cand}\subset \mathcal W\times\mathcal T .
\]
The inferential grid
\(\mathcal G_n^{\rm inf}\subseteq \mathcal G_n^{\rm cand}\) is obtained by the trimming rule of \hyperref[sec:defs]{Section~\ref*{sec:defs}}, whose inputs are \(\mathcal F_n^{\rm des}\)-measurable. The key inferential requirement is not
that the grid be deterministic unconditionally, but that it be
independent of the main analysis sample and fixed after conditioning on
\(\mathcal F_n^{\rm des}\).

\paragraph{Design principles.}
A valid and useful anchor design should satisfy five principles.

\begin{enumerate}
\item \textbf{Design exogeneity.} Anchors and radii must be fixed
before the main-sample doubly robust scores are evaluated. They may be
deterministic or pilot-measurable, but not selected using the same
outcomes that enter the final score averages.

\item \textbf{Geometric coverage.} The anchor set should cover the
working anchor domain \(\mathcal W\) at a resolution fine enough to approximate
the continuous profile map when such an approximation is intended.

\item \textbf{Score nondegeneracy.} Points of the inferential grid should avoid
near-deterministic indicators. Operationally, this motivates trimming
profile cells with marginal pilot probabilities close to 0 or 1.

\item \textbf{Standardized signal relevance.} For testing, the relevant
quantity is not the raw effect \(|\Delta(g)|\), but the standardized
effect \(|\Delta(g)|/\sigma(g)\), where \(\sigma^2(g)\) is the efficient
score variance.

\item \textbf{Multiplicity budget.} Adding anchors can improve signal
coverage, but the critical value of the max statistic can only increase.
Anchor selection therefore involves a power--stability tradeoff rather
than a universally optimal densification rule.
\end{enumerate}

\begin{proposition}[Design-conditional transfer of fixed-grid validity]
\label{prop:design-valid-anchors}
Suppose \(\mathcal G_n^{\rm inf}\) is measurable with respect to an auxiliary
design sigma-field
\(\mathcal F_n^{\rm des}\subseteq\mathcal F_n^{\mathrm{aux}}\),
independent of the main analysis sample, and that the dimension and nondegeneracy
conditions
\[
    \log^7(|\mathcal G_n^{\rm cand}|n)=o(n),
    \qquad
    \inf_{g\in \mathcal G_n^{\rm inf}}\sigma(g)\ge c_\sigma>0
\]
hold, the second with probability tending to one. Assume the measurable-enumeration and
high-probability event conventions of Assumptions~\ref{ass:meas_design}
and~\ref{ass:grid_det}. Assume the conditions of \hyperref[thm:clt]{Theorem~\ref*{thm:clt}} and
\hyperref[cor:scb]{Corollary~\ref*{cor:scb}} hold uniformly over \(\mathcal G_n^{\rm inf}\). Then,
conditionally on \(\mathcal F_n^{\rm des}\), the cross-fitted profile
estimator admits the same uniform linearization, Gaussian approximation,
multiplier-bootstrap validity, simultaneous confidence bands, and
max-test level control as in \hyperref[thm:clt]{Theorem~\ref*{thm:clt}},
\hyperref[cor:scb]{Corollary~\ref*{cor:scb}}, and \hyperref[thm:global_test]{Theorem~\ref*{thm:global_test}}.
\end{proposition}

\begin{proof}
Condition on \(\mathcal F_n^{\rm des}\). After conditioning,
\(\mathcal G_n^{\rm inf}\) is fixed and independent of the main sample. The main
sample remains i.i.d., and the cross-fitting folds are formed only
inside the main sample. The proof of \hyperref[thm:clt]{Theorem~\ref*{thm:clt}} therefore applies
without modification to the fixed finite vector
\(\{\psi(O;g):g\in \mathcal G_n^{\rm inf}\}\). The assumed nuisance rates imply
that the second-order bias and plug-in empirical-process remainders are
uniformly negligible over the selected grid. The dimension condition and
the bounded standardized scores yield the high-dimensional Gaussian
approximation for the maximum of the standardized score vector. The
studentization bound \eqref{eq:correct_scale_rate}, the correlation condition of \hyperref[ass:gauss_boot]{Assumption~\ref*{ass:gauss_boot}}, and Gaussian anti-concentration \eqref{eq:anticoncentration} transfer this
approximation to the studentized estimator and to the multiplier
bootstrap quantile. Since the conditional coverage and level statements
hold on an event whose auxiliary-design probability tends to one,
unconditional validity follows by averaging over
\(\mathcal F_n^{\rm des}\).
\end{proof}

\hyperref[prop:design-valid-anchors]{Proposition~\ref*{prop:design-valid-anchors}} is a transfer principle for
auxiliary-sample-based anchor construction. An auxiliary or pilot sample may be used to
stabilize the grid or to form a signal-aware benchmark, but the
resulting inference is conditional on the selected grid and is valid only
because that grid is external to the main score evaluation.

We now state a coverage criterion showing why anchor dispersion matters.

\begin{assumption}[Uniform profile anti-concentration]
\label{assump:profile-anticons}
There exists \(C_{\rm prof}<\infty\) such that, for each treatment arm
\(a\in\{0,1\}\), each anchor \(\omega\in\mathcal W\), each \(u\ge0\), and each
\(r>0\),
\[
    \mathbb P\{u<d(\omega,Y(a))\le u+r\}\le C_{\rm prof} r .
\]
\end{assumption}

This assumption is used only for continuous-domain approximation. It is
not needed for pointwise identification or finite-grid inference.

\begin{proposition}[Lipschitz profile effects and grid approximation]
\label{prop:grid-approx-anchor}
Under \hyperref[assump:profile-anticons]{Assumption~\ref*{assump:profile-anticons}}, the profile causal
effect satisfies
\[
    |\Delta(\omega,t)-\Delta(\omega',t')|
    \le
    2C_{\rm prof}\{d(\omega,\omega')+|t-t'|\}.
\]
Consequently, if \(\mathcal W_K\subset \mathcal W\) is a finite
measurable \(\rho_W\)-net of \(\mathcal W\) and
\(\mathcal T_L\subset \mathcal T\) is a finite measurable \(\rho_T\)-net of
\(\mathcal T\), then
\[
    \left|
    \sup_{(\omega,t)\in \mathcal W\times\mathcal T}|\Delta(\omega,t)|
    -
    \max_{(\omega,t)\in \mathcal W_K\times \mathcal T_L}|\Delta(\omega,t)|
    \right|
    \le
    2C_{\rm prof}(\rho_W+\rho_T).
\]
Moreover, if a measurable nearest-neighbor tie-breaking rule induces a
partition of \(\mathcal W\times\mathcal T\) with cells \(\{C_g\}\) and
weights \(q_g=\mu(C_g)\), then
\[
    \left|
    \int_{\mathcal W\times\mathcal T}\Delta(g)^2\,d\mu(g)
    -
    \sum_{g\in \mathcal W_K\times \mathcal T_L}q_g\Delta(g)^2
    \right|
    \le
    4C_{\rm prof}(\rho_W+\rho_T),
\]
where \(\mu\) is any probability measure on
\(\mathcal W\times\mathcal T\).
\end{proposition}

\begin{proof}
Fix \(a\in\{0,1\}\) and write
\(F_{a,\omega}(t)=\mathbb P\{d(\omega,Y(a))\le t\}\). Let
\(\delta=d(\omega,\omega')+|t-t'|\). By the triangle inequality,
\[
    d(\omega',Y(a))\le d(\omega,Y(a))+d(\omega,\omega').
\]
Hence
\[
    F_{a,\omega}(t)
    \le
    F_{a,\omega'}(t+\!d(\omega,\omega'))
    \le
    F_{a,\omega'}(t'+\delta).
\]
Therefore
\[
    F_{a,\omega}(t)-F_{a,\omega'}(t')
    \le
    F_{a,\omega'}(t'+\delta)-F_{a,\omega'}(t')
    \le
    C_{\rm prof}\delta .
\]
Interchanging \((\omega,t)\) and \((\omega',t')\) gives
\[
    |F_{a,\omega}(t)-F_{a,\omega'}(t')|
    \le C_{\rm prof}\delta .
\]
Taking the difference between arms \(a=1\) and \(a=0\) yields
\[
    |\Delta(\omega,t)-\Delta(\omega',t')|
    \le 2C_{\rm prof}\delta .
\]

For the supremum approximation, take any
\(g=(\omega,t)\in \mathcal W\times\mathcal T\)
and choose \(g_K=(\omega_K,t_L)\in \mathcal W_K\times\mathcal T_L\) with
\(d(\omega,\omega_K)\le \rho_W\) and \(|t-t_L|\le \rho_T\). The
Lipschitz bound gives
\[
    |\Delta(g)|\le |\Delta(g_K)|+2C_{\rm prof}(\rho_W+\rho_T).
\]
Taking the supremum over \(g\) proves the upper inequality; the reverse
inequality is immediate because
\(\mathcal W_K\times\mathcal T_L\subset\mathcal W\times\mathcal T\).

For the integrated squared approximation, let \(g^\star\) be the nearest
grid point to \(g\). Since \(|\Delta(g)|\le 1\), we have
\[
    |\Delta(g)^2-\Delta(g^\star)^2|
    \le
    2|\Delta(g)-\Delta(g^\star)|
    \le
    4C_{\rm prof}(\rho_W+\rho_T).
\]
Integrating over the nearest-neighbor cells gives the stated bound.
\end{proof}

Discrete object spaces may violate \hyperref[assump:profile-anticons]{Assumption~\ref*{assump:profile-anticons}}
because distance distributions can have atoms. This does not affect the
finite-grid inference results above; it only limits when the
continuous-domain approximation bound should be invoked. \hyperref[prop:grid-approx-anchor]{Proposition~\ref*{prop:grid-approx-anchor}} concerns the population map; it does not control the empirical or multiplier processes between grid points, so it does not by itself give a band over a continuum of anchors.

\hyperref[prop:grid-approx-anchor]{Proposition~\ref*{prop:grid-approx-anchor}} gives a theoretical standard:
anchors should be selected so that the covering radius of the working
domain is small relative to the effect size one wants to detect. A
MaxMin or farthest-point medoid rule is natural because it directly
targets the covering radius. For exact farthest-first selection on a
finite auxiliary cloud started from a cloud point, such as the medoid, let
\(r_K\) be the covering radius of the cloud by the first \(K\) anchors and
\(r_*\) the smallest covering radius of the cloud by any \(K\) centers.
Adding the farthest remaining point gives \(K+1\) points that are pairwise
at distance at least \(r_K\), because the covering radius does not increase
along the recursion. In an optimal cover two of these points share a ball,
so they are within \(2r_*\) of each other and \(r_K\le2r_*\). The argument
uses the triangle inequality and an initial cloud point; it does not apply
to completions of externally specified scientific anchors.

For power, coverage alone is insufficient. The max test is driven by
standardized signal.

\begin{proposition}[Near-optimal pilot signal selection]
\label{prop:pilot-signal-selection}
Let \(\mathcal G_n^{\rm guard}\) be a nonempty guarded subset of a finite candidate
grid \(\mathcal G_n^{\rm cand}\) with \(\sigma(g)>0\) for every
\(g\in \mathcal G_n^{\rm guard}\), let \(S(g)=|\Delta(g)|/\sigma(g)\) be the
standardized signal on \(\mathcal G_n^{\rm guard}\), and let \(\widehat S(g)\) be a
pilot estimate satisfying
\[
    \sup_{g\in \mathcal G_n^{\rm guard}}
    |\widehat S(g)-S(g)|\le \varepsilon_n
\]
with probability tending to one. Let \(\widehat g\) be any maximizer of
\(\widehat S(g)\) over \(\mathcal G_n^{\rm guard}\). Then, on the event where the
displayed bound holds,
\[
    S(\widehat g)
    \ge
    \max_{g\in \mathcal G_n^{\rm guard}}S(g)-2\varepsilon_n .
\]
More generally, on the same event, the top-\(K\) pilot-selected set contains at least one
grid point whose standardized signal is within \(2\varepsilon_n\) of the
best guarded-grid signal.
\end{proposition}

\begin{proof}
Let \(g^\star\in\arg\max_{g\in \mathcal G_n^{\rm guard}}S(g)\). Since
\(\widehat g\) maximizes \(\widehat S\),
\[
    \widehat S(\widehat g)\ge \widehat S(g^\star).
\]
Using the uniform pilot error twice,
\[
    S(\widehat g)
    \ge
    \widehat S(\widehat g)-\varepsilon_n
    \ge
    \widehat S(g^\star)-\varepsilon_n
    \ge
    S(g^\star)-2\varepsilon_n .
\]
The same argument applies to the top-\(K\) set because it contains a
maximizer of \(\widehat S\).
\end{proof}

\paragraph{Task-relative designs.}
Geometric coverage gives a default design when no scientific distinction
is prespecified. When the question concerns a declared family of marginal
law changes \(\nu_\theta=P_{1,\theta}-P_{0,\theta}\), write
\(s_G(\theta)=(\nu_\theta(B_g):g\in G)\) for the signature of
\(\nu_\theta\) on a grid \(G\), where \(B_g\) is the ball of coordinate
\(g\). The grid resolves a target \(q(\theta)\) within the family if
\(s_G(\theta)=s_G(\theta')\) implies \(q(\theta)=q(\theta')\).
Equivalently, with \(\mathsf S_G\nu=(\nu(B_g):g\in G)\), two law changes
have the same profile on \(G\) exactly when their difference lies in the
kernel of \(\mathsf S_G\). Labels such as location, dispersion, tail, or
local structure therefore require a declared family and its separation on
the chosen grid; \hyperref[prop:anchor_blindness]{Proposition~\ref*{prop:anchor_blindness}} shows that this
separation can fail even with complete radius profiles at several anchors.
For such a family, a working design may maximize the smallest
coordinatewise standardized separation between the signatures of the
declared alternative classes, relative to a working max-test cutoff, with
all inputs fixed by external or auxiliary information. A union
\(G_{\rm coverage}\cup G_{\rm task}\) of a coverage design and a task
design can be analyzed as one frozen grid with one common cutoff. These
are design principles; the experiments in \hyperref[sec:experiments]{Section~\ref*{sec:experiments}}
do not optimize this criterion.

\subsection{Efficiency--Power Pareto Frontier under a Fixed Metric}

After the outcome metric \(d\) is fixed, the remaining statistical
choices include the anchor grid, the radius grid, trimming threshold,
nuisance learners, and testing functional. These choices affect several
objectives simultaneously: pointwise estimation efficiency,
simultaneous-band width, global-test power, representation error, and
computational cost.

For an admissible design \(D\), let \(G_D\) be its inferential grid and
let
\[
    \sigma_D^2(g)
    =
    {\rm Var}\{\psi(O;g)\},\qquad g\in G_D,
\]
be the efficient score variance. Let \(c_D(1-\alpha)\) be the
\((1-\alpha)\)-quantile of the corresponding Gaussian max statistic.
We define the first-order simultaneous-band width criterion
\[
    \mathcal B(D)
    =
    c_D(1-\alpha)
    \cdot
    |G_D|^{-1}\sum_{g\in G_D}\sigma_D(g),
\]
so that, for a fixed design satisfying the conditions of \hyperref[cor:scb]{Corollary~\ref*{cor:scb}}, the average simultaneous confidence-band width is
\(2\mathcal B(D)/\sqrt n+o_P(n^{-1/2})\). Smaller
\(\mathcal B(D)\) means higher simultaneous estimation efficiency.

For local alternatives
\[
    \Delta_n(g)=h(g)/\sqrt n,
\]
we use as a design criterion the Gaussian power
\[
    \mathcal P(D;h)
    =
    \mathbb P\left\{
    \max_{g\in G_D}
    \left|Z_D(g)+\frac{h(g)}{\sigma_D(g)}\right|
    >
    c_D(1-\alpha)
    \right\},
\]
where \(Z_D\) is the centered Gaussian process with the correlation
matrix of the standardized efficient score vector. It approximates the rejection probability of the max test along these alternatives whenever the shifted maximum and the feasible cutoff admit the corresponding Gaussian approximations.

\subsubsection{Pointwise Efficiency Implies Pointwise Local Power}

We first record a standard benchmark. For a fixed profile location, the
efficient estimator is also locally most powerful among regular
asymptotically linear estimators using the same Wald calibration.

\begin{proposition}[Pointwise efficiency and local power]
\label{thm:pointwise-eff-power}
Fix \(g\). Consider a regular asymptotically linear estimator
\(\widetilde\Delta(g)\) with asymptotic variance \(0<v(g)<\infty\), and the
two-sided level-\(\alpha\) Wald test studentized by a consistent estimator
of \(v(g)\). Along a regular parametric submodel with
\(\Delta_n(g)=h(g)/\sqrt n\), this test has asymptotic power
\[
    \pi_v(h)
    =
    1-\Phi\left(z_{1-\alpha/2}-\frac{|h(g)|}{\sqrt{v(g)}}\right)
    +
    \Phi\left(-z_{1-\alpha/2}-\frac{|h(g)|}{\sqrt{v(g)}}\right).
\]
This power is nonincreasing in \(v(g)\). Therefore, any estimator
attaining the semiparametric efficiency bound
\(\sigma_{\rm eff}^2(g)\) maximizes first-order pointwise local power
among regular estimators.
\end{proposition}

\begin{proof}
Let
\[
    \nu=\frac{|h(g)|}{\sqrt{v(g)}},\qquad c=z_{1-\alpha/2}.
\]
The limiting Wald statistic is \(|Z+\nu|\), with \(Z\sim N(0,1)\). Hence
\[
    \pi_v(h)
    =
    \mathbb P(|Z+\nu|>c)
    =
    1-\Phi(c-\nu)+\Phi(-c-\nu).
\]
Differentiating with respect to \(\nu\ge0\) gives
\[
    \frac{d}{d\nu}\pi_v(h)
    =
    \phi(c-\nu)-\phi(c+\nu)\ge0,
\]
because \(c+\nu\ge |c-\nu|\) for \(c,\nu\ge0\) and the standard normal
density decreases with \(|x|\); the derivative is positive when \(c>0\) and \(\nu>0\), since then \(c+\nu>|c-\nu|\). Thus power increases with \(\nu\), and
\(\nu\) decreases with \(v(g)\). Since the semiparametric efficiency
bound is the smallest attainable asymptotic variance among regular
estimators, an efficient estimator maximizes the first-order pointwise
local power.
\end{proof}

\hyperref[thm:pointwise-eff-power]{Proposition~\ref*{thm:pointwise-eff-power}} clarifies the role of the doubly
robust efficient score: at a fixed anchor--radius pair, efficiency and
power are aligned. The conflict appears once the coordinate itself is a design choice,
already between two single-coordinate designs (\hyperref[prop:non-dominance]{Proposition~\ref*{prop:non-dominance}}).

\subsubsection{Global Power and Multiplicity}

For a grid \(G_D\), define the maximal standardized local signal
\[
    \mu_D(h)
    =
    \max_{g\in G_D}
    \frac{|h(g)|}{\sigma_D(g)}.
\]

\begin{proposition}[A sufficient condition for max-test power]
\label{prop:max-detection-boundary}
For every \(c\in\mathbb R\),
\[
    \mathbb P\left\{
    \max_{g\in G_D}
    \left|Z_D(g)+\frac{h(g)}{\sigma_D(g)}\right|
    >c\right\}
    \ge
    1-\Phi\{c-\mu_D(h)\}.
\]
In particular, \(\mathcal P(D;h)\ge1-\Phi\{c_D(1-\alpha)-\mu_D(h)\}\), so that, with \(z_\beta=\Phi^{-1}(\beta)\), the condition \(\mu_D(h)\ge c_D(1-\alpha)+z_\beta\) gives \(\mathcal P(D;h)\ge\beta\).
\end{proposition}

\begin{proof}
Let \(g^\star\in\arg\max_{g\in G_D}|h(g)|/\sigma_D(g)\), and define
\[
    s^\star=
    \begin{cases}
    1, & h(g^\star)\ge0,\\
    -1, & h(g^\star)<0,
    \end{cases}
    \qquad
    \mu_D(h)=s^\star h(g^\star)/\sigma_D(g^\star).
\]
Since \(s^\star Z_D(g^\star)\sim N(0,1)\) and the maximum dominates the
signed coordinate,
\[
\begin{aligned}
    \mathbb P\left\{
    \max_{g\in G_D}
    \left|Z_D(g)+\frac{h(g)}{\sigma_D(g)}\right|
    >c\right\}
    &\ge
    \mathbb P\left\{
    s^\star Z_D(g^\star)+\mu_D(h)>c
    \right\} \\
    &=
    1-\Phi(c-\mu_D(h)).
\end{aligned}
\]
Taking \(c=c_D(1-\alpha)\) gives the bound for \(\mathcal P(D;h)\), and the sufficient condition follows.
\end{proof}

The bound motivates the signal-minus-critical-value proxy
\[
    \mu_D(h)-c_D(1-\alpha).
\]
Adding anchors may increase \(\mu_D(h)\) by covering a stronger signal
region, but it may also increase \(c_D(1-\alpha)\). The bound is not a
necessary condition for power, because the maximum uses the whole signal
vector and its dependence. With independent coordinates that each have
standardized shift \(2\) and \(\alpha=0.05\), one coordinate gives
\(c=1.960\) and power \(0.516\), whereas two coordinates give
\(c=2.236\) and power \(0.648\), although the largest shift is the same.
Profile scores attain this configuration: for a Euclidean outcome on
\(\{0,1\}^2\) with independent fair coordinates, the anchors \((0,1/2)\)
and \((1/2,0)\) with radius \(0.6\) select the events that the first and
the second coordinate equal zero; under randomization with probability
\(1/2\) the two oracle contrast scores have unit variance and are
uncorrelated, and raising each of these probabilities to
\(1/2+2/\sqrt n\) under treatment gives the stated shifts.

\begin{lemma}[Multiplicity monotonicity under a common Gaussian process]
\label{lem:multiplicity}
If \(G_D\subseteq G_{D'}\) and both maxima are evaluated under the same
standardized Gaussian score process \(Z(\cdot)\), then
\[
    c_{D'}(1-\alpha)\ge c_D(1-\alpha).
\]
\end{lemma}

\begin{proof}
Since \(G_D\subseteq G_{D'}\),
\[
    \max_{g\in G_{D'}}|Z(g)|
    \ge
    \max_{g\in G_D}|Z(g)|
\]
almost surely. Therefore the distribution of the larger-grid maximum
first-order stochastically dominates the distribution of the smaller-grid
maximum, and its \((1-\alpha)\)-quantile cannot be smaller.
\end{proof}

\hyperref[lem:multiplicity]{Lemma~\ref*{lem:multiplicity}} explains why denser nested grids under a
fixed score process are not automatically better: a new anchor can raise the
simultaneous critical value, and whether power improves depends on the
whole standardized signal vector and its dependence. The example after
\hyperref[prop:max-detection-boundary]{Proposition~\ref*{prop:max-detection-boundary}} shows a power gain without
any increase in the strongest standardized signal. Across different
designs the standardized covariance structure also changes, so the lemma
does not compare arbitrary design classes.

\subsubsection{Pareto Efficiency of Designs}

Let \(\mathfrak D\) denote a collection of admissible designs. Besides
\(\mathcal B(D)\) and \(\mathcal P(D;h)\), let
\(\mathcal R(D)\) denote a representation error, such as the covering
radius bound in \hyperref[prop:grid-approx-anchor]{Proposition~\ref*{prop:grid-approx-anchor}}, and let
\(\mathcal C(D)\) denote computational cost, for example proportional to
\(|G_D|\) times the cost of nuisance evaluation.

We say that \(D_1\) dominates \(D_2\) if
\[
    \mathcal B(D_1)\le \mathcal B(D_2),\qquad
    \mathcal R(D_1)\le \mathcal R(D_2),\qquad
    \mathcal C(D_1)\le \mathcal C(D_2),\qquad
    \mathcal P(D_1;h)\ge \mathcal P(D_2;h),
\]
with at least one strict inequality. A design is Pareto efficient if no
other admissible design dominates it.

\begin{proposition}[Weighted scalarization yields Pareto-efficient designs]
\label{thm:pareto-scalarization}
Let \(\lambda_B,\lambda_R,\lambda_C,\lambda_P>0\). Suppose the
optimization problem
\[
    \max_{D\in\mathfrak D}
    \left[
        \lambda_P\mathcal P(D;h)
        -
        \lambda_B\mathcal B(D)
        -
        \lambda_R\mathcal R(D)
        -
        \lambda_C\mathcal C(D)
    \right]
\]
has a solution \(D^\star\). Then \(D^\star\) is Pareto efficient.
\end{proposition}

\begin{proof}
Suppose \(D^\star\) is not Pareto
efficient. Then there exists \(D'\in\mathfrak D\) such that
\[
    \mathcal P(D';h)\ge \mathcal P(D^\star;h),\quad
    \mathcal B(D')\le \mathcal B(D^\star),\quad
    \mathcal R(D')\le \mathcal R(D^\star),\quad
    \mathcal C(D')\le \mathcal C(D^\star),
\]
with at least one strict inequality. Since all weights are strictly
positive, the scalarized objective is strictly larger at \(D'\) than at
\(D^\star\), contradicting the optimality of \(D^\star\). Hence
\(D^\star\) must be Pareto efficient.
\end{proof}

In nonconvex design classes, some Pareto-efficient designs need not solve any weighted scalarization with nonzero nonnegative weights.

The following example shows that efficiency and power can conflict even
in the simplest one-coordinate designs.

\begin{proposition}[Non-dominance of efficiency and power]
\label{prop:non-dominance}
Fix a level \(\alpha\) and let \(c=z_{1-\alpha/2}\). Consider two
single-coordinate designs \(D_E=\{g_E\}\) and \(D_P=\{g_P\}\) with
score standard deviations
\[
    0<\sigma_E<\sigma_P.
\]
Let the local alternative satisfy \(h(g_E)=0\) and \(h(g_P)>0\). Then
\(D_E\) has smaller simultaneous-band width, while \(D_P\) has strictly
larger local power. Hence neither design dominates the other.
\end{proposition}

\begin{proof}
For a one-coordinate design, the simultaneous-band width multiplier is
\(c\sigma\). Since \(\sigma_E<\sigma_P\),
\[
    c\sigma_E<c\sigma_P,
\]
so \(D_E\) is more efficient in the band-width sense.

Under \(D_E\), the local signal is zero, so the two-sided Wald test has
asymptotic power \(\alpha\). Under \(D_P\), the local noncentrality is
\[
    \nu_P=\frac{h(g_P)}{\sigma_P}>0.
\]
The power is
\[
    1-\Phi(c-\nu_P)+\Phi(-c-\nu_P).
\]
As shown in the proof of \hyperref[thm:pointwise-eff-power]{Proposition~\ref*{thm:pointwise-eff-power}}, this
quantity is strictly increasing in \(\nu_P\) for \(\nu_P>0\), and equals
\(\alpha\) when \(\nu_P=0\). Therefore \(D_P\) has strictly larger power
than \(D_E\). One design is better for efficiency, the other for power,
so neither dominates the other.
\end{proof}

\paragraph{Empirical Pareto protocol.}
In practice, we recommend evaluating a finite family of candidate designs
\(\{D_\ell\}\), such as geometry-only MaxMin designs, guarded
pilot-signal designs, and hybrid shell-balanced designs. For each
candidate, compute pilot or simulation estimates of
\[
    \widehat{\mathcal B}(D_\ell),\qquad
    \widehat{\mathcal P}(D_\ell;h)
    \ \text{or}\
    \widehat{\mathcal M}(D_\ell)
    =
    \max_{g\in G_{D_\ell}}
    \frac{\sqrt n|\widehat\Delta_{\rm pilot}(g)|}
         {\widehat\sigma_{\rm pilot}(g)}
    -
    \widehat c_{D_\ell,1-\alpha},
\]
together with representation error, trimming fraction, inferential grid
size, and computational cost. Designs dominated on all reported metrics
should be discarded. The final design can then be selected either by a
prespecified weight vector or by choosing a stable knee point on the
empirical Pareto frontier. This selection must be completed before the
main analysis sample is used.

\FloatBarrier
\section{Simulation Details}
\label{app:simulation_details}

\subsection{Five-Space Configuration}

The five-space calibration study uses the same inferential construction across all outcome spaces and anchor strategies. The main analysis sample is used only for nuisance fitting and doubly robust score evaluation; auxiliary design information is used only to freeze anchors, radii, and the inferential grid before main-sample inference. The anchor, pilot, and reference sizes in \hyperref[tab:full_config]{Table~\ref*{tab:full_config}} are components of the auxiliary information summarized theoretically by \(\mathcal F_n^{\mathrm{aux}}\). The high-sample block in the main text uses the DRF weighted ECDF learner at \(n=3200,6400,12800\) and a separate design with bivariate Euclidean outcomes: \(X\sim N(0,I_2)\), propensity score \(\mathrm{expit}(0.6X_1-0.6X_2)\), a \(10\times10\) grid frozen from an independent auxiliary sample of size \(10{,}000\), trimming probabilities computed on a second independent sample of the same size, true profiles approximated from \(100{,}000\) oracle draws, \(500\) null and \(300\) alternative repetitions per sample size, \(500\) multiplier draws, and two folds. The propensity score is estimated by logistic regression with ridge penalty \(10^{-4}\), and the fitted values are clipped to \([0.05,0.95]\). In every reported max-test analysis, the trimming probabilities \(\hat p_{\mathrm{pilot}}\) of \hyperref[sec:defs]{Section~\ref*{sec:defs}} are computed on a reference sample, which plays the role of \(\mathcal I_{\mathrm{pilot}}\) and is separate from the samples that construct and score the anchors.

\begin{table}[t]
\centering
\caption{Configuration of the five-space calibration study.}
\label{tab:full_config}
\scriptsize
\begin{tabular}{@{}lp{0.62\textwidth}@{}}
\toprule
Component & Value \\
\midrule
Outcome spaces & Euclidean, Wasserstein, SPD, graph, rooted tree \\
Sample sizes & $200,400,800,1600,3200$ \\
Monte Carlo repetitions & $500$ \\
Multiplier-bootstrap draws & $1000$ \\
Anchors and radii & $K_\omega=10$, $K_t=10$ \\
Radius rule & quantiles at $10$ equally spaced levels from $0.05$ to $0.95$ of the pooled distances from the anchor-design sample to the selected anchors, shared across anchors \\
Inferential-grid trimming & pilot probabilities in $[0.05,0.95]$ \\
Cross-fitting folds & $2$ \\
True propensity score & $0.05+0.9\,\mathrm{expit}(0.8X_1-0.7X_2+0.5X_3)$, $X\sim N(0,I_3)$ \\
Fitted propensity score & logistic regression (ridge penalty $10^{-4}$), clipped to $[0.05,0.95]$ \\
SCB level & $0.90$ \\
Global test level & $0.05$ \\
Conditional CDF learner & kernel-weighted ECDF; covariates standardized within each training set; bandwidth the median pairwise distance, at least \(0.15\) \\
Auxiliary design sizes & anchor/pilot/reference sizes $20{,}000$; signal candidates $3000$ \\
Oracle sample for true profiles & $200{,}000$ \\
Tree dissimilarity & recursive alignment with optional child swaps; $1.0$ per unmatched node plus $0.25$ times the node-value difference, divided by the $15$ node slots \\
\bottomrule
\end{tabular}
\end{table}

In the five-space block, the multiplier draws for the tree-space runs and for the graph Guarded Shell-Signal alternative at \(n=3200\) used score rows centered within each cross-fitting fold, while the standard deviations in these runs were computed from globally centered rows. In the high-sample, separation, confounding, ABIDE, and insurance experiments, coordinates whose main-sample estimated standard deviation is below \(0.02\) are excluded from the inferential grid in addition to the pilot-probability trimming. Where the retained counts were recorded, this screen rarely changed the grid: it removed no probability-retained coordinate in the high-sample runs and a single coordinate in \(12\) of the \(6400\) ABIDE fits. The tree dissimilarity can violate the triangle inequality, so the tree space is a dissimilarity space rather than a metric space; the procedure uses its ball indicators as features in the same way as for the other spaces. In every experiment, the radii form one vector shared by all anchors of a design. They are quantiles of pooled auxiliary distances, except in the mean-zero examples, where eight radii are equally spaced between \(0.1\) and \(0.9\) times the largest pilot distance to the central anchor. The Guarded Shell-Signal shells are defined by distance to a medoid of the pilot outcomes. In the five-space block, the candidate anchors are sorted by this distance and split into three groups of nearly equal size; in the insurance and ABIDE benchmarks, the shells are cut at the tertiles of the distances, so that tied distances fall in the same shell. Medoids of large pools, including the starting anchor of Medoid MaxMin, are computed approximately: a random subset of candidate points is scored by total distance to a random subset of support points.

\begin{table}[t]
\centering
\caption{Anchor-design strategies.}
\label{tab:anchor_strategy_catalog}
\scriptsize
\begin{tabular}{@{}p{0.26\textwidth}p{0.68\textwidth}@{}}
\toprule
Strategy & Rule \\
\midrule
Medoid MaxMin & Select anchors by a medoid-start MaxMin rule on an independent anchor-design sample; radii use fixed empirical quantiles. \\
Guarded Shell-Signal & Split candidate anchors into three shells by distance to a pilot medoid, score them on a pilot sample independent of the main sample, apply guarded selection within shells, and freeze the resulting grid before main-sample score evaluation. \\
Balanced LHS & Select the balanced winner from a Latin-hypercube-screened family of pilot-conditioned candidate rules, balancing power, calibration, coverage, and error diagnostics. \\
Calibration LHS & Select the Latin-hypercube-screened candidate emphasizing null calibration and simultaneous-band coverage. \\
Power LHS & Select the Latin-hypercube-screened candidate emphasizing alternative rejection subject to validity diagnostics. \\
\bottomrule
\end{tabular}
\end{table}

\subsection{Exploratory Anchor-Design Screening}

We additionally ran an exploratory anchor-design screening to assess whether pilot-conditioned anchor selection can improve power while preserving the same max-statistic inferential calibration. The screening generated a family of signal-aware LHS candidate rules from independent pilot data, evaluated them under the same studentized max-statistic calibration, and selected representative winners according to balanced, calibration-oriented, and power-oriented criteria. All screened variants freeze anchors before main-sample inference, but they are treated as exploratory because their selection rule is more adaptive than the formal Medoid MaxMin design. \hyperref[tab:anchor_strategy_catalog]{Table~\ref*{tab:anchor_strategy_catalog}} describes the strategies.

\begin{table}[t]
\centering
\caption{Exploratory anchor-design screening at \(n=800\) with the kernel-weighted ECDF learner, averaged over the five outcome spaces (\(500\) repetitions per cell). Power and Size: rejection rates of the level-\(0.05\) max test; SCB: null coverage of the \(90\%\) simultaneous band; \(L^2\) and Sup: estimation errors under the alternative, relative to a Monte Carlo approximation of the true profile.}
\label{tab:anchor_screening_summary}
\scriptsize
\begin{tabular}{@{}lccccc@{}}
\toprule
Strategy & Power & Size & SCB & $L^2$ & Sup \\
\midrule
Medoid MaxMin & 0.860 & 0.040 & 0.917 & 0.0264 & 0.0802 \\
Guarded Shell-Signal & 0.876 & 0.035 & 0.926 & 0.0292 & 0.0776 \\
Balanced LHS & 0.893 & 0.038 & 0.923 & 0.0286 & 0.0740 \\
Calibration LHS & 0.876 & 0.035 & 0.928 & 0.0276 & 0.0776 \\
Power LHS & 0.881 & 0.041 & 0.925 & 0.0278 & 0.0760 \\
\bottomrule
\end{tabular}
\end{table}

In this exploratory screening (\hyperref[tab:anchor_screening_summary]{Table~\ref*{tab:anchor_screening_summary}}), the pilot-conditioned rules increased average power relative to Medoid MaxMin, with null rejection rates and band coverage in a similar range. We view Balanced LHS as an exploratory extension, and Calibration LHS and Power LHS as calibration- and power-oriented variants. Medoid MaxMin, whose anchors depend only on auxiliary geometry, remains the default design. In this screening, the multiplier draws used score rows centered within each cross-fitting fold.

\subsection{Additional Experimental Figures}\label{app:exp_figures}

Figures~\ref{fig:app_convergence_diagnostics}--\ref{fig:app_anchor_screening} show the high-sample error and calibration diagnostics, an estimator-level double-robustness check, pilot trimming fractions in the five-space study, and the exploratory anchor-design screening.

\begin{figure}[t]
\centering
\includegraphics[width=0.82\textwidth]{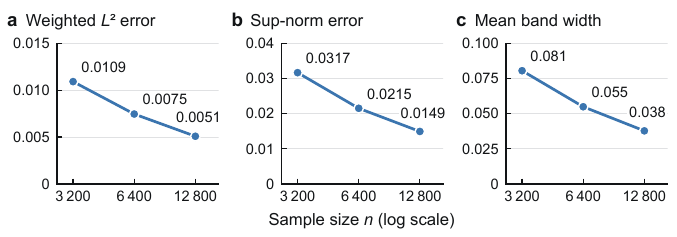}
\caption{DRF finite-grid estimation error under the alternative in the high-sample block: (a) weighted \(L^2\) error, (b) sup-norm error, and (c) mean band width, averaged over \(300\) repetitions per sample size; all \(300\) reject at each sample size. Errors are measured against a Monte Carlo approximation of the true profile from \(100{,}000\) oracle draws.}
\label{fig:app_convergence_diagnostics}
\end{figure}

\begin{figure}[t]
\centering
\includegraphics[width=0.455\textwidth]{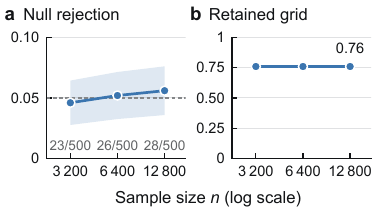}
\caption{DRF null calibration of the studentized max-statistic test on the inferential grid: (a) null rejection rate with counts and (b) retained grid fraction, \(500\) repetitions per sample size. Shaded bands are \(95\%\) Wald intervals and the dashed line marks the nominal level. Because the true null curve is zero, the \(95\%\) simultaneous band covers it exactly when the max test does not reject, so its coverage is \(0.954\), \(0.948\), and \(0.944\) at the three sample sizes.}
\label{fig:app_null_calibration_by_n}
\end{figure}

\begin{figure}[t]
\centering
\includegraphics[width=0.72\textwidth]{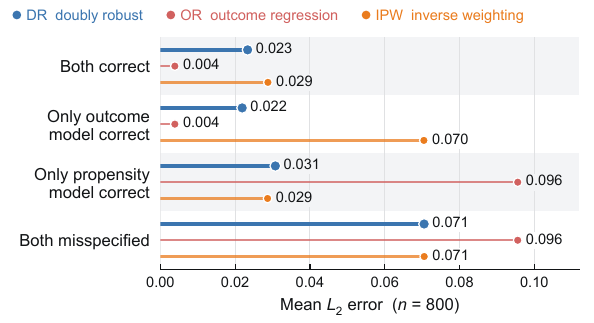}
\caption{Estimator-level double-robustness check in a Euclidean simulation design at $n=800$: mean \(L^2\) error of the doubly robust (DR), outcome-regression (OR), and inverse-weighting (IPW) estimators under four specification regimes, with \(500\) repetitions per regime. DR stays accurate when either nuisance model is correct (\(0.022\)--\(0.031\)). When the outcome model is correct, OR alone is most accurate (\(0.004\)); when both models are misspecified, DR and IPW coincide (\(0.071\)).}
\label{fig:app_double_robustness_sanity}
\end{figure}

\begin{figure}[t]
\centering
\includegraphics[width=0.72\textwidth]{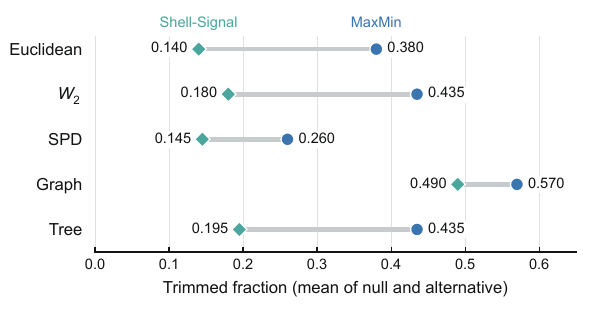}
\caption{Pilot trimming by outcome space and anchor strategy. Each value averages the trimmed fractions of the null and alternative pilots; for graphs the two move in opposite directions (Medoid MaxMin \(0.27\) and \(0.87\), Guarded Shell-Signal \(0.80\) and \(0.18\)).}
\label{fig:app_trimming_diagnostics}
\end{figure}

\begin{figure}[t]
\centering
\includegraphics[width=0.70\textwidth]{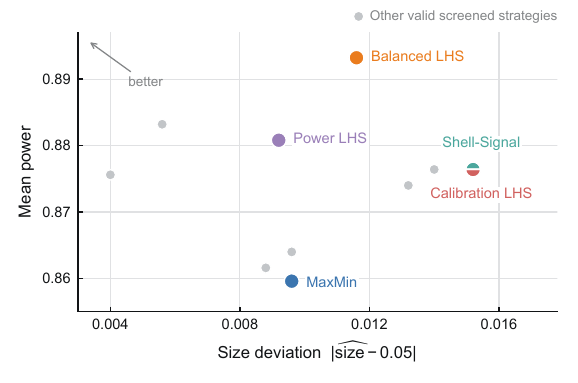}
\caption{Exploratory anchor-design screening: mean power against the absolute deviation of size from \(0.05\), averaged over the five outcome spaces, for candidate anchor-selection rules. Gray points are the other screened rules that passed the validity gate. Guarded Shell-Signal and Calibration LHS have identical coordinates and share a split marker.}
\label{fig:app_anchor_screening}
\end{figure}

\clearpage
\FloatBarrier
\section{Real-Geometry and Related-Work Supplementary Results}
\label{app:real_object_related_work}

This appendix records supplementary figures and interpretation boundaries for the evaluations summarized in \hyperref[sec:experiments]{Section~\ref*{sec:experiments}}. These analyses use the same profile estimator, inferential-grid mask, and max-statistic calibration as the main analysis; the difference is the empirical substrate or the neighboring baseline being tested.

\subsection{Real-Geometry Semi-Synthetic Benchmarks}
\label{app:realgeom_benchmarks}

The insurance benchmark uses restricted processed data derived from the social-network insurance experiment of~\citet{cai2015social}. We keep the real network, real covariates, and observed treatment structure as the substrate for known-truth graph-valued potential outcomes (Figures~\ref{fig:app_insurance} and~\ref{fig:app_insurance_metric}). The benchmark has \(1{,}844\) ego networks, of which \(922\) form the main sample; outcomes are compared by the Frobenius distance between adjacency matrices, and each anchor design uses \(500\) repetitions. Null rejection stays below the nominal level in every design, and the pilot-conditioned designs increase structural-alternative rejection relative to the geometry-only Medoid MaxMin baseline.

\begin{figure}[!htb]
\centering
\includegraphics[width=0.92\textwidth]{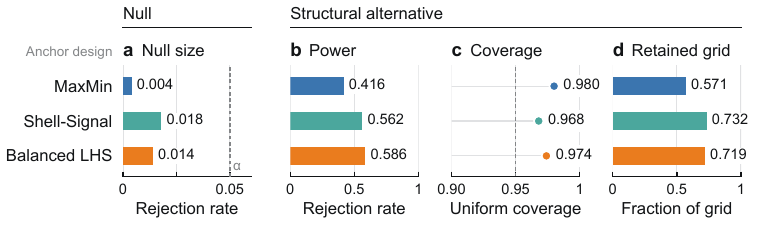}
\caption{Insurance real-network semi-synthetic benchmark under DRF across anchor designs: (a) null size, (b) structural-alternative power, (c) uniform coverage, and (d) retained grid fraction. Panels (c) and (d) use the structural-alternative runs; coverage is evaluated against the average profile effect over the main-sample units. Each design uses \(500\) repetitions.}
\label{fig:app_insurance}
\end{figure}

\begin{figure}[!htb]
\centering
\includegraphics[width=0.82\textwidth]{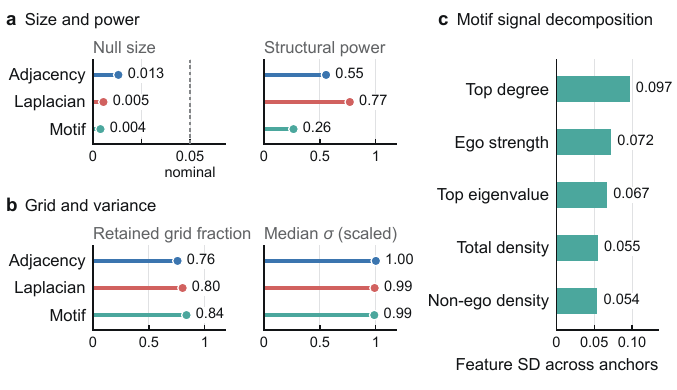}
\caption{Insurance metric diagnostics, averaged over \(26\) anchor strategies. (a) All three metrics are conservative under the null (sizes \(0.013\), \(0.005\), and \(0.004\)); the Laplacian metric has the highest structural power and the motif-summary metric the lowest. (b) Retained grid fraction and scaled median standard deviation. (c) The five of the eight motif-summary features with the largest standard deviation across anchors. These diagnostics complement the adjacency-Frobenius results.}
\label{fig:app_insurance_metric}
\end{figure}

The ABIDE benchmark uses ABIDE PCP CPAC functional-connectivity derivatives with Craddock-200 ROI time series~\citep{dimartino2014autism,craddock2013neuro,craddock2012whole} as high-dimensional SPD objects. The known-truth semi-synthetic layer retains real connectomes and covariates while injecting controlled prototype-shift and dispersion-collapse alternatives (\hyperref[fig:main_abide_semisynth_drf]{Figure~\ref*{fig:main_abide_semisynth_drf}}). The benchmark has \(845\) units, of which \(422\) form the main sample, and compares connectomes with the log-Euclidean metric. The effect strength \(s\) is a mixing weight in log-matrix coordinates: under dispersion collapse, \(Y(0)\) is the observed connectome and \(Y(1)\) moves toward the midpoint of the medoids of the two diagnostic groups with weight \(0.55s\); under prototype shift, \(Y(0)\) and \(Y(1)\) move toward the typically developing and autism medoids with weights \(0.20s\) and \(0.45s\). The anchor design is Guarded Shell-Signal on an \(8\times8\) grid, with \(200\) repetitions per effect strength. The observational ASD/TD comparison is reported as a descriptive profile analysis (\hyperref[fig:app_abide_observational]{Figure~\ref*{fig:app_abide_observational}}).

\begin{figure}[!htb]
\centering
\includegraphics[width=\textwidth]{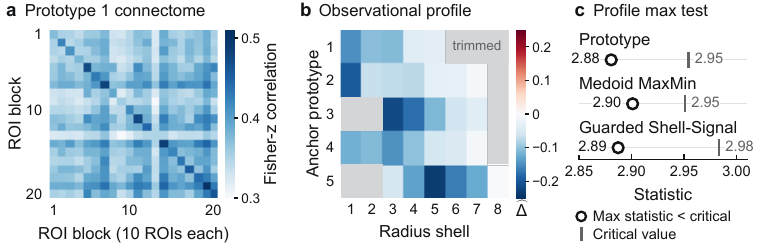}
\caption{ABIDE observational profile analysis (ASD versus TD). (a) First prototype connectome. (b) Profile map for the prototype anchors; gray cells are trimmed. (c) The profile max test does not reject in any of the three anchor designs. A mean-connectome proxy, not shown, rejects (statistic \(1.645\), critical value \(1.544\)). The analysis is descriptive; calibration and power are assessed in the semi-synthetic benchmark.}
\label{fig:app_abide_observational}
\end{figure}

\subsection{Auxiliary-Sample Allocation}
\label{subsec:allocation}

The anchor design uses auxiliary observations that are not available for score evaluation. To examine this cost, we used the ABIDE semi-synthetic substrate with a fixed pool of \(845\) units, the prototype-shift alternative, and three allocations of the pool (\hyperref[tab:allocation]{Table~\ref*{tab:allocation}}). State L uses a larger auxiliary sample; state C uses a smaller auxiliary sample with the same \(422\)-unit main sample and leaves \(170\) units unused; state D returns these units to the main sample. Each state uses an \(8\times8\) candidate grid, the DRF learner, two folds, \(500\) multiplier draws, and \(500\) paired repetitions at each effect strength. These runs are separate from those of \hyperref[fig:main_abide_semisynth_drf]{Figure~\ref*{fig:main_abide_semisynth_drf}}.

\begin{table}[htbp]
\centering
\small
\caption{Allocation of a fixed pool of \(845\) ABIDE units. Pilot and reference samples form the auxiliary design, and trimming uses the reference sample; entries in the last two columns are rejection rates of the level-\(0.05\) max test with the DRF learner, \(500\) paired repetitions each.}
\label{tab:allocation}
\begin{tabular}{@{}lrrrrrr@{}}
\toprule
State & Pilot & Reference & Main & Unused & Effect \(0\) & Effect \(0.10\) \\
\midrule
L & 254 & 169 & 422 & 0 & 0.020 & 0.448 \\
C & 152 & 101 & 422 & 170 & 0.020 & 0.388 \\
D & 152 & 101 & 592 & 0 & 0.018 & 0.660 \\
\bottomrule
\end{tabular}
\end{table}

At effect \(0.10\), the paired difference C\(-\)L is \(-0.060\), with a conservative \(95\%\) Monte Carlo interval \([-0.119,-0.001]\), and D\(-\)L is \(0.212\), with interval \([0.139,0.282]\). Each interval combines two \(97.5\%\) Clopper--Pearson intervals for the discordant pairs through a Bonferroni bound. The comparison concerns the complete procedures: the states differ in their auxiliary designs and, for D, in the main sample, and in one null replication of state C the main-sample standard-deviation screen removed one of \(46\) coordinates that state D retained. In this setting, shrinking the auxiliary sample without using the freed units lowered power, whereas returning them to the main sample raised it substantially.

Three elementary calculations clarify the cost of the auxiliary sample. Write \(N=m_p+m_r+n\) for the pilot, reference, and main sizes. For a grid that is already fixed, with the same score variance and nuisance quality, reserving auxiliary units inflates the first-order pointwise standard error by the factor \(\sqrt{N/n}\) relative to using all \(N\) units; this comparison ignores the max-test cutoff and the design choices themselves. Conditional on a pilot-fixed set of \(p\) candidate coordinates whose pooled ball probabilities lie at least \(\gamma\) from both trimming thresholds, and with the reference sample independent of the pilot sample, Hoeffding's inequality and a union bound show that the reference-sample trimming differs from population trimming with probability at most \(2p\exp(-2m_r\gamma^2)\). Finally, a pilot of size \(m_p\) contains no point of a region of probability \(\rho\) with probability \((1-\rho)^{m_p}\), so a rule that places anchors at pilot points needs \(m_p\ge\log(1/\delta)/\{-\log(1-\rho)\}\) to see such a region with probability at least \(1-\delta\). Grids specified from external information do not consume study units, and auxiliary units could also be reused for off-fold nuisance training; such extensions require their own independence arguments and are not used in this paper.

\subsection{Distribution-Valued Observational Analysis}

The NHANES analysis follows the Lin-style processed sample from the 2005--2006 NHANES Physical Activity Monitor data~\citep{lin2023causal,cdc2008nhanespaxraw}, keeps the covariate specification with age and gender, and treats marriage as the exposure, coded as married or living with a partner versus otherwise. The sample contains \(5{,}524\) participants aged \(14\) to \(85\) and is analyzed without survey weights. The exposure is rare among the \(966\) participants aged \(14\)--\(17\) (\(15\) exposed), so overlap is weak in this age group, and the analysis is descriptive. The profile analysis answers a different question from a horizontal quantile-effect curve: whether the outcome distribution becomes closer to prototype neighborhoods such as low-activity, high-activity, right-tail-active, or concentrated activity profiles (\hyperref[fig:app_nhanes]{Figure~\ref*{fig:app_nhanes}}).

\begin{figure}[t]
\centering
\includegraphics[width=\textwidth]{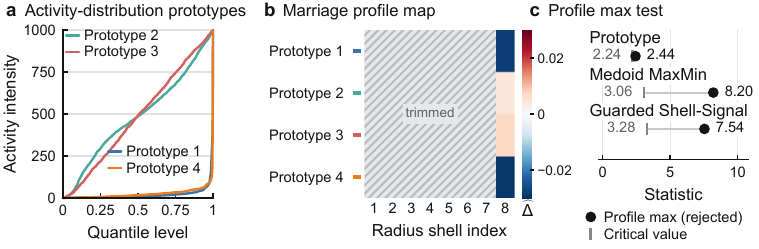}
\caption{NHANES activity distributions with marriage as the exposure. (a) Four activity-distribution prototypes, shown as quantile functions. (b) Profile map for the prototype anchors; hatched cells are trimmed, leaving \(4\) of \(32\) cells. (c) Profile max statistic and critical value for three anchor designs; the profile max test rejects in all three. The mean-quantile proxy, not shown, has statistic \(19.81\) and critical value \(4.83\). The profile view describes movement toward or away from prototype neighborhoods and complements the horizontal quantile contrast; the analysis is descriptive.}
\label{fig:app_nhanes}
\end{figure}

\subsection{Related-Work Stress Tests}

In the related-work simulations, each sample of total size \(n\) is split at random into an anchor-design part (\(30\%\)), a trimming-reference part (\(20\%\)), and a main part (\(50\%\)) used for nuisance fitting and score evaluation. The simulations isolate three points. First, the Wasserstein collapse example follows the separating construction in which the mean-quantile contrast is weak or zero, but the center-anchor profile effect is large. Second, the graph and SPD mean-zero examples keep the Fr\'echet-mean proxy near zero while changing dispersion or multimodality (\hyperref[fig:app_related_work]{Figure~\ref*{fig:app_related_work}}). Third, under confounded causal nulls, a metric two-sample statistic that ignores covariates rejects frequently, because it compares observed groups (\hyperref[fig:main_estimand_confounding_drf]{Figure~\ref*{fig:main_estimand_confounding_drf}}(b)).

\begin{figure}[t]
\centering
\includegraphics[width=0.90\textwidth]{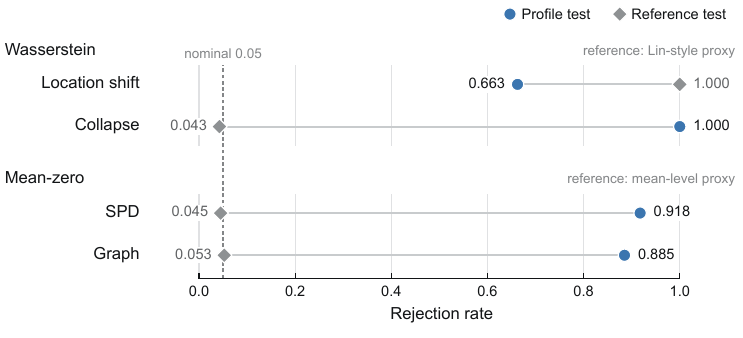}
\caption{Related-work separating examples under DRF. The Lin-style location setting is a positive control; the Wasserstein collapse setting shows a profile signal with weak mean-quantile contrast. The graph and SPD mean-zero examples show profile power under object-law collapse while the Fr\'echet-mean (mean-level) proxy remains close to nominal; the Lin-style proxy is the mean-quantile proxy. Rejection rates pool the total sample sizes \(n=200\) and \(n=400\), with \(200\) repetitions each; separately, the profile test rejects at \(0.470\) and \(0.855\) in the location setting, \(0.790\) and \(0.980\) for graphs, and \(0.845\) and \(0.990\) for SPD matrices.}
\label{fig:app_related_work}
\end{figure}

\paragraph{Matched tasks.}
We also compared the profile test with neighboring proxies on two randomized tasks: a location shift of Wasserstein quantile outcomes, against the mean-quantile proxy, and a log-location shift of SPD outcomes, against a log-Euclidean Fr\'echet-mean proxy, each with shift \(\theta\in\{0,0.5,1\}\) (\hyperref[tab:matched_tasks]{Table~\ref*{tab:matched_tasks}}). Both procedures use the same \(400\) main observations, with treatment assigned independently with probability \(0.5\); the profile procedure additionally uses \(200\) auxiliary observations, \(120\) for anchor design and \(80\) for trimming. Its grid combines one scientific anchor (the standard normal quantile object or the SPD identity) with seven farthest-first anchors and eight shared radii. The propensity score is estimated by logistic regression clipped at \(0.05\), the profile calibration uses \(500\) multiplier draws, and the proxies use \(500\) permutations. Under the null, the retained fraction of the \(64\) candidate coordinates averaged \(0.735\) and \(0.729\) for the Wasserstein and SPD tasks, the empirical covering radius of the anchors on the design sample \(0.500\) and \(1.434\), the radii ranged on average from \(0.262\) to \(2.444\) and from \(0.860\) to \(2.664\), and the max-test critical value averaged \(3.135\) and \(3.216\). Both tests reject almost always under the Wasserstein shifts, whereas the Fr\'echet-mean proxy is more powerful for the SPD shift at \(\theta=0.5\). The two procedures test different functionals with different multiplicity, so the comparison concerns rejection behavior on a shared task rather than efficiency for a common estimand.

\begin{table}[htbp]
\centering
\small
\caption{Matched randomized tasks: rejection rates at level \(0.05\) with Wilson \(95\%\) Monte Carlo intervals, \(200\) repetitions per cell. The profile procedure uses \(200\) auxiliary observations in addition to the \(400\) shared main observations.}
\label{tab:matched_tasks}
\begin{tabular}{@{}llll@{}}
\toprule
Task & \(\theta\) & Profile test & Proxy test \\
\midrule
Wasserstein location & \(0\) & 0.035 [0.017, 0.070] & 0.060 [0.035, 0.102] \\
 & \(0.5\) & 1.000 [0.981, 1.000] & 0.975 [0.943, 0.989] \\
 & \(1\) & 1.000 [0.981, 1.000] & 1.000 [0.981, 1.000] \\
SPD log-location & \(0\) & 0.035 [0.017, 0.070] & 0.035 [0.017, 0.070] \\
 & \(0.5\) & 0.710 [0.644, 0.768] & 0.895 [0.845, 0.930] \\
 & \(1\) & 1.000 [0.981, 1.000] & 1.000 [0.981, 1.000] \\
\bottomrule
\end{tabular}
\end{table}

These supplementary experiments complement the five-space simulation block: they examine portability to real geometries and the relationship between the profile estimand and neighboring baselines.

\clearpage
\FloatBarrier
\section{Conditional-CDF Learner Comparison: DRF and Kernel ECDF}
\label{app:learner_robustness}

The high-sample, separation, confounding, and real-geometry experiments use the DRF weighted ECDF learner, and the five-space table uses the anchor-specific kernel-weighted ECDF learner. This appendix records the paired comparison of the two learners. Both learners estimate the same conditional-CDF nuisance target and preserve monotonicity in \(t\); they differ in finite-sample nuisance behavior and runtime.

\subsection{High-\texorpdfstring{\(n\)}{n} Finite-Sample Diagnostics}

The high-sample design is given in \hyperref[app:simulation_details]{Appendix~\ref*{app:simulation_details}}. Its diagnostics are empirical counterparts of the rates in \hyperref[ass:uniform_remainder]{Assumption~\ref*{ass:uniform_remainder}}, computed on the main sample against the true nuisance functions: \(\hat r_\pi\) is the root mean squared error of the out-of-fold clipped propensity predictions, \(\hat r_m\) is the largest such error of the out-of-fold conditional-CDF predictions over arms and grid cells, the squared-rate and product-rate diagnostics are \((\hat r_m^2+\hat r_\pi^2)L^2\) and \(\sqrt n\,\hat r_m\hat r_\pi L^{1/2}\), and the scaled remainder is \(L^{1/2}\max_g|n^{-1/2}\sum_{i=1}^n\{\hat\phi_i(g)-\phi(O_i;g,\eta_0)\}|\) over the cells retained by the trimming mask and the standard-deviation screen of \hyperref[app:simulation_details]{Appendix~\ref*{app:simulation_details}}, where \(L=\log100\) is the logarithm of the candidate-grid size.

Kernel ECDF gives conservative finite-sample max-statistic behavior and decreasing estimation error, but its empirical nuisance diagnostics do not clearly improve at the tested sample sizes. In the high-sample controlled runs, ECDF null rejection is \(0.020\), \(0.022\), and \(0.022\) for \(n=3200,6400,12800\), while SCB coverage is \(0.980\), \(0.978\), and \(0.978\). Its weighted \(L^2\) and sup-norm errors decrease, but the conditional-CDF error \(\hat r_m\) remains nearly flat around \(0.355\), and the product-rate diagnostics do not clearly decline.

DRF improves the nuisance diagnostics. Across \(n=3200,6400,12800\), \(\hat r_m\) decreases from \(0.193\) to \(0.120\), the squared-rate diagnostic from \(0.796\) to \(0.305\), and the product-rate diagnostic from \(0.447\) to \(0.272\). The scaled remainder also decreases, from \(1.554\) to \(1.058\). These empirical errors decrease over the reported sample sizes; they describe the fitted procedure and do not verify the rate conditions of \hyperref[ass:uniform_remainder]{Assumption~\ref*{ass:uniform_remainder}}.

\subsection{Learner Scorecard}

\begin{table}[!htbp]
\centering
\caption{Learner comparison summary.}
\label{tab:learner_robustness_scorecard}
\scriptsize
\setlength{\tabcolsep}{3pt}
\begin{tabular}{@{}>{\raggedright\arraybackslash}p{0.22\textwidth}>{\raggedright\arraybackslash}p{0.22\textwidth}>{\raggedright\arraybackslash}p{0.20\textwidth}>{\raggedright\arraybackslash}p{0.26\textwidth}@{}}
\toprule
Block & DRF weighted ECDF & Kernel ECDF & Use in paper \\
\midrule
High-\(n\) diagnostics & Null rejection \(0.046\)--\(0.056\); \(\hat r_m\), the squared-rate and product-rate diagnostics, and the scaled remainder decrease with \(n\) & More conservative null rejection \(0.020\)--\(0.022\); nuisance-rate diagnostics weak & DRF is used in the main text for this block. \\
Insurance graph & Structural power \(0.416\)--\(0.586\); null size \(0.004\)--\(0.018\) & Structural power \(0.392\)--\(0.550\); null size \(0.004\)--\(0.022\) & DRF is at least as strong in the graph benchmark. \\
ABIDE semi-synthetic & Collapse power \(0.430\) at effect \(0.05\), \(0.955\) at \(0.10\); prototype power \(0.970\) at \(0.20\) & Nearly identical power curve; slightly wider bands & Supports portability to high-dimensional connectomes. \\
Confounded null & Wasserstein \(n=400\): profile \(0.030\), naive \(0.674\) & Profile \(0.038\), naive \(0.674\) & The naive metric test, which targets the observational null, rejects frequently under the causal null. \\
Separating examples & Wasserstein collapse profile \(1.000\) versus mean-quantile proxy \(0.040\)--\(0.045\); SPD/graph mean-zero profile \(0.980\)--\(0.990\) at \(n=400\) & Similar separating-example behavior & Supports estimand separation; the tree space is not included in this block. \\
\bottomrule
\end{tabular}
\end{table}

The learner comparison (\hyperref[tab:learner_robustness_scorecard]{Table~\ref*{tab:learner_robustness_scorecard}}) supports the use of DRF outside the five-space table. DRF improves the finite-grid nuisance diagnostics and performs well on the ABIDE semi-synthetic power curves, the confounding stress tests, the separating examples, and the insurance graph benchmark. Kernel ECDF remains valuable because it is simple, fast, and structurally coherent, but its nuisance-rate diagnostics are weaker in the high-sample analysis.

\subsection{Relation to Reported Claims}

These paired results support three limited conclusions. First, the conditional-CDF nuisance target can be estimated by more than one structurally coherent learner. Second, DRF has smaller and more clearly decreasing empirical conditional-CDF prediction errors in the high-sample runs. Third, under confounded causal nulls the doubly robust profile test stays below the nominal level, whereas the naive metric two-sample test, which compares observed groups, rejects frequently.

\section{Experimental Reproducibility Details}
\label{app:reproducibility_details}

The code, configuration files, and figure-rendering scripts used for the reported simulation, semi-synthetic, and portability analyses are available from the authors upon request. The code is organized into reusable modules for object-space data generation, anchor and radius construction, nuisance estimation, doubly robust profile estimation, simultaneous bands, max-statistic tests, and visualization. The script directory contains runners for the five-space calibration experiments, confounding stress tests, related-work separating examples, insurance graph semi-synthetic benchmark, ABIDE semi-synthetic benchmark, and figure regeneration.

The configurations used for the reported results cover the null and alternative finite-grid calibration blocks, nuisance-rate diagnostics, Wasserstein separating examples, graph and SPD mean-zero examples, confounding-failure experiments, insurance graph semi-synthetic runs, and ABIDE semi-synthetic runs. Large raw datasets, restricted insurance-network data, and large per-repetition matrices are not redistributed; for public or externally governed data sources, provenance and access notes are provided instead of the raw data.

The reported experiments are CPU-based and require no GPU; their cost is dominated by repeated nuisance fitting, inferential-grid construction, and bootstrap max-statistic calibration. The distributional random forest nuisance learner uses an R backend. The realized backend is recorded for each fit, and no fallback to another backend occurred in the reported runs.

The figure-rendering script regenerates the main and appendix figures from lightweight derived outputs. These regenerated figures cover the finite-grid convergence and calibration diagnostics, estimator-level double-robustness check, trimming diagnostics, anchor-design screening, related-work stress tests, insurance metric diagnostics, real-geometry semi-synthetic benchmarks, and the observational figures. The observational ABIDE and NHANES figures are descriptive; the known-truth validation claims in the paper rely on the simulation and semi-synthetic blocks described above.

\end{document}